\documentclass[final,3p,times]{elsarticle}
\usepackage{epstopdf}
\usepackage{color}
\usepackage{appendix}
\usepackage{bm}
\usepackage{graphicx}
\usepackage{amsmath}
\usepackage{amsfonts}
\usepackage{amssymb}
\usepackage{bezier} 
\usepackage{color} 

\usepackage{graphicx,amsmath,mathrsfs}
\usepackage{amssymb}
\usepackage{amsthm,multirow}
\usepackage{bm,bbm}
\usepackage{subfigure,xcolor} 
\usepackage[makeroom]{cancel}
\usepackage{color}
\usepackage{float}
\usepackage{bbold}
\def \be{\begin{equation}}
\def \ee{\end{equation}}
\def \beq{\begin{equation}}
\def \eeq{\end{equation}}
\def \bea{\begin{eqnarray}}
\def \eea{\end{eqnarray}}

\def\btheta{{\boldsymbol{\theta}}}

\def\BZ{{\rm BZ}}

\def\bq{{\bf q}}
\def\bQ{{\bf Q}}

\def\bx{{\bf x}}

\def\bv{{\bf v}}
\def\bj{{\bf j}}
\def\br{{\bf r}}
\def\bE{{\bf E}}
\def\bj{{\bf j}}
\def\bJ{{\bf J}}

\def\bS{{\bf S}}
\def\bk{{\bf k}}

\def\ij{\langle i,j\rangle}
\def\nd{{\vphantom{\dagger}}}
\def\cP{{\cal P}}

\def\etal{{\it et al.~}}
\def\be{\begin{equation}}
\def\beq{\begin{equation}}
\def\ee{\end{equation}}
\def\eeq{\end{equation}}
\def\bea{\begin{eqnarray}}
\def\eea{\end{eqnarray}}
\def\bA{{\bf A}}
\def\bB{{\bf B}}

\def\cO{{\cal O}}
\def\cM{{\cal M}}
\def\cL{{\cal L}}

\def\half{{1\over 2}}

\def\tt{{t}}

\def\bnabla{{\boldsymbol{\nabla} }}
\def\bsigma{{\boldsymbol{\sigma} }}
\def\balpha{{\boldsymbol{\alpha} }}
\def\bkappa{{\boldsymbol{\kappa} }}
\def\bOmega{{\boldsymbol{\Omega} }}

 \def\bq{{\bf q}}
\def\bB{{\bf B}}
\def\bR{{\bf R}}
\def\bE{{\bf E}}
\def\bx{{\bf x}}

\def\bv{{\bf v}}
\def\bp{{\bf p}}
\def\br{{\bf r}}
\def\bz{{\bf z}}
\def\bj{{\bf j}}

\def\ba{{\bf a}}
\def\bA{{\bf A}}
\def\bk{{\bf k}}
\def\bQ{{\bf Q}}

\def\cL{{\cal L}}

\def\cP{{\cal P}}

\def\cO{{\cal O}}

\def\cM{{\cal M}}

\def\bP{{\bf P}}

\def\tc{{\tilde c}}
\def\Im{{\mbox{Im}}}
\def\Re{{\mbox{Re}}}
\def\Tr{{\mbox{Tr}}}

\def\bnabla{{\boldsymbol{\nabla}}}

 \def\bq{{\bf q}}
\def\bB{{\bf B}}
\def\bR{{\bf R}}
\def\bE{{\bf E}}
\def\bM{{\bf M}}
\def\bx{{\bf x}}

\def\bv{{\bf v}}
\def\bp{{\bf p}}
\def\br{{\bf r}}
\def\bz{{\bf z}}
\def\bj{{\bf j}}

\def\ba{{\bf a}}
\def\bs{{\bf s}}
\def\bA{{\bf A}}
\def\bk{{\bf k}}
\def\bQ{{\bf Q}}

\def\cL{{\cal L}}

\def\cV{{\cal V}}
\def\cP{{\cal P}}

\def\cO{{\cal O}}

\def\cM{{\cal M}}

\def\bP{{\bf P}}

\def\tc{{\tilde c}}
\def\Im{{\mbox{Im}}}
\def\Re{{\mbox{Re}}}
\def\Tr{{\mbox{Tr}}}

\def\bpsi{{\boldsymbol{\psi}}}

\def\bDelta{{\bar{\Delta}}}
\def\bmu{{\bar\mu}}

\newtheorem{lemma}{Lemma}
\def\ve{{\varepsilon}}
\def\vve{{\epsilon}}
\def\ij{{\langle ij\rangle}}
\def\jk{{\langle jk\rangle}}

\def\blangle{\boldsymbol{\Big<}}
\def\brangle{\boldsymbol{\Big>}}
\def\bvert{\boldsymbol{\big|}}
\def\kexpect#1#2#3{\blangle \, #1 \bvert #2 \bvert #3 \,\brangle}
\def\kbraket#1#2{\blangle \, #1 \bvert #2 \,\brangle}
\def\kbra#1{\blangle \, #1 \bvert }
\def\kket#1{\bvert #1 \,\brangle}
\def\vvv{\vphantom{\big |}}

\begin{document}
\title{Quantum Transport Theory of Strongly Correlated Matter}

\author{Assa Auerbach and Sauri Bhattacharyya}
\affiliation{Physics Department, Technion, 32000 Haifa, Israel.}

\date{\today }
\begin{abstract}  
This report reviews recent progress in 
computing Kubo formulas for general interacting Hamiltonians.  The aim is to  calculate electric and thermal magneto-conductivities in strong scattering regimes where Boltzmann equation and Hall conductivity proxies exceed their validity.  Three primary approaches are explained.
\begin{enumerate}
    \item   
Degeneracy-projected polarization formulas for Hall-type conductivities, which substantially reduce the number of calculated current matrix elements. These expressions generalize the Berry curvature integral formulas to imperfect lattices. 
\item Continued fraction representation of dynamical longitudinal conductivities. The calculations produce
a set of  thermodynamic averages, which can be controllably extrapolated using their mathematical relations to  low and high frequency conductivity asymptotics.
\item Hall-type  coefficients summation formulas, which are constructed from thermodynamic averages.
\end{enumerate}

The thermodynamic formulas are derived in the operator Hilbert space formalism,  which  avoids the opacity and high computational cost of the Hamiltonian eigenspectrum. The coefficients can be obtained by well established imaginary-time Monte Carlo sampling, high temperature expansion, traces of operator products, and variational wavefunctions at low temperatures.  

We demonstrate the power of  approaches 1--3  by their application to well known models of  lattice electrons and bosons. The calculations clarify the far-reaching influence of strong local interactions on the metallic transport near Mott insulators. Future directions for these approaches are discussed.
\end{abstract}

\maketitle
\newpage
\tableofcontents
\newpage
\part{Introduction}
\section{Why calculate conductivities?}
Theorists commonly describe phases of matter by their order parameters, since they can be calculated by tried and tested algorithms of equilibrium statistical mechanics. These  {\em ``thermodynamic methods''} include stochastic series expansions~\cite{SSE}, worm algorithm of Quantum Monte-Carlo simulations~\cite{prokofiev,Worm}, and variational methods such as
density matrix renormalization group~\cite{DMRG,DMRG-Schollwok}, projected entangled-pair states~\cite{PEPS} and  tensor networks~\cite{TN}.  

Experimentalists on the other hand, commonly probe phases of matter by transport measurements. For example:  Hall coefficients can characterize the current-carriers' density (near band extrema).  Temperature dependent resistivity may herald the onset of superconductivity or  charge localization. Quantized Hall conductivity is associated with topologically ordered  ground states.    

Electric, thermo-electric and thermal transport coefficients,  $\bsigma,\balpha, {\bar\balpha}$ and $\bkappa$,  are defined by linear response equations,
\bea
\bj   &=& \bsigma \cdot \left( \bE  -  \bnabla  \mu /e\right)  ~ + ~ \balpha  \cdot    \left( - \bnabla T   \right) \quad ,\nonumber\\
 \bj_Q   &=&  T  \bar{\balpha}   \cdot  \left(\bE  -  \bnabla \mu /e\right) ~ + ~\bkappa \cdot  ( - \bnabla T  ) \quad .
\label{Tr0}
\eea
$\bj$ and $\bj_Q$ are the  
electric and thermal currents respectively. $\bE-\bnabla\mu/e$ and $\bnabla T$  are the electrochemical field and temperature gradient respectively. In Eq.~(\ref{Tr0}), the wavevector $\bq$ and frequency $\omega$ dependence of all variables are suppressed.

For systems described by weakly scattered Bloch-band quasiparticles, the Boltzmann equation~\cite{KL,Ziman} and diagrammatic perturbation theory~\cite{Mahan} are adequate.  For incompressible quantum Hall phases~\cite{QHE} and topological insulators~\cite{Bernevig}, proxies such as the Chern number~\cite{TKNN,Avron}  and Streda formula~\cite{Streda} yield the Hall conductivity. 

However, for strongly correlated gapless systems which exhibit   ``bad  metal'' phenomenology~\cite{badmetals}, quasiparticle  descriptions may fail. The only Hamiltonian-based alternative to weak scattering approaches is the computation of the Kubo formulas~\cite{Kubo}.
However, dynamical Kubo formulas may be forbiddingly difficult.
Exact diagonalizations entail exponentially large memory costs, and analytic continuation of Quantum Monte Carlo data to real frequencies is ill-posed at low frequencies~\cite{Snir}. Simplifications of Kubo formulas which apply to gapless phases are in dire need.

This Report reviews recent advances  which sidestep some of the Kubo formula difficulties, and
render conductivity calculations in the presence of strong interactions more accessible. 

\section{Lingering issues clarified}
Before delving into details, we list certain questions which have permeated the common lore of transport theory, and are resolved in this Report. 

\subsection{Are DC Hall conductivities {\em on-shell} or {\em off-shell} expressions?} 
In the Lehmann (eigenstates) representation of Kubo formulas,  {\em on-shell} expressions involve current matrix elements $j_{nm}^\alpha$ between quasi-degenerate states (i.e. whose  energies' separation $E_n-E_m$ vanishes in the large volume limit).
On-shell expressions are  implemented by taking
\be
\lim_{\ve\to 0}\lim_{\bq\to 0}\lim_{\cV\to \infty} \Im {1\over E_n(\cV)-E_m(\cV)-i\ve} \to \pi \lim_{\bq\to 0}\lim_{\cV\to \infty}\delta(E_n(\cV)-E_m(\cV)) \quad .
\label{onshellIm}
\ee
Real longitudinal conductivities turn out to be on-shell expressions.

In contrast, Hall-type (antisymmetric transverse) conductivities, involve sums over the {\em real} part of  energy denominators, i.e.
\be
 \Re {1\over E_n-E_m-i\ve}={E_n-E_m\over (E_n-E_m)^2 + \ve^2} \quad .
\label{onshellRe}
\ee
By blithely setting $\ve\to 0$, and neglecting $\cO(\ve^2)$ terms, one may be wrongly led to believe that  ``Hall conductivities are off-shell expressions'', i.e. that they include $j_{nm}^\alpha$ which connect between well separated energies in the large volume limit.

This statement is misleading. As shown below by the degeneracy projected polarization (DPP) formulas in Part \ref{sec:DPP}, on open boundary conditions (OBC), terms of order $\ve^2$ in  (\ref{onshellRe}) are {\em essentially important} in the DC limit!
In fact, Hall-type conductivities with OBC are {\em purely  on-shell expressions}. Physically, this implies that at low temperatures, the Hall current is carried by low energy gapless excitations, which may be located in the bulk or at the system's edges.

Off-shell Kubo formulas have been used in the literature, most notably by Thouless \etal~\cite{TKNN} in their seminal derivation of the quantized Hall conductance of perfectly periodic lattices. However, they are only valid in limited cases such as a gapped ground state with periodic boundary conditions (PBC), as discussed in Section~\ref{sec:Proxies}.
\\

\subsection{What is the origin of  magnetization subtractions in $\alpha_{xy}$ and $ \kappa_{xy}$? Must we compute them?}

The infamous magnetization corrections for $\alpha_{xy}$ and $\kappa_{xy}$~\cite{Cooper,QNS}  (which inconveniently diverge at zero temperature)  are an artifact of introducing a static ``gravitational field'' $\bpsi_{\omega=0}$ in lieu of the  temperature gradient $-\bnabla T$, which is a {\em non-equilibrium} statistical force.  The static $\bpsi$ produces superfluous magnetization currents which must be  subtracted from the Kubo formula in the DC limit~\cite{Cooper,QNS}. 
However, magnetization subtractions can be eliminated using
the DPP formulas for $\alpha_{xy}$ and $ \kappa_{xy}$ as shown in subsection \ref{sec:MagCT}. The magnetization terms also fall out of the Hall coefficient summation formulas of Part~\ref{Part:Thermo II}.

\subsection{Can DC dissipative transport coefficients be expressed in terms of static thermodynamic coefficients?}
At first thought, one may suspect that expressions involving on-shell scattering rates (i.e. Fermi's golden rule),  may not be applicable to calculations of static thermodynamic averages.

On the other hand, it has long been realized that some dynamical conductivities may be computed by continued fraction representations~\cite{RXX-PRB,Viswanath}, which are constructed from the conductivity moments which are a set of thermodynamic averages.
The price to pay is a necessary {\em extrapolation} of the calculated moments up to infinite order. While this may turn out to be  a daunting task, extrapolation is occasionally facilitated by appealing to mathematical relations, which are reviewed in Section \ref{sec:Recurrents-Cond}. These relations connect between 
high order moments and high frequency limits of the dynamical conductivity. 

In a different approach, Part \ref{Part:Thermo II} derives 
thermodynamic summation formulas for the low magnetic field
 Hall-type coefficients. 
Truncation  of this sum may sometimes be justified by calculating leading order corrections, as demonstrated  in the model examples of Part \ref{part:SCE} and Part \ref{part:SCB}.

\subsection{What are the effects of a Mott insulator phase on the longitudinal and Hall transport of a nearby metal?} 

 Strong local interactions effects give rise to the Mott insulator and extend deep into the nearby metallic phase. The effects include large high temperature linear resisitivity slopes, and Hall coefficient sign reversal and divergence at low doping of the Mott insulator. These effects are captured by thermodynamic
 Kubo formula calculations of the two dimensional t-J model  and Hard Core Bosons model, in Part \ref{part:SCE} and Part \ref{part:SCB} respectively.  
 
\section{Organization of the report}
In Part \ref{sec:BH}, we present a  brief review of Drude, Boltzmann, and memory functional transport theories which apply to Hamiltonians of weakly scattered quasiparticles. In Parts~\ref{sec:KF}-\ref{Part:Thermo II}, we
lay the mathematical basis for approaches applicable to strongly interacting Hamiltonians. First, in Part \ref{sec:KF}, the Kubo formula with its proper DC order of limits is defined. The reversal of that order of limits in Chern number and Streda formula proxies is discussed.

New approaches are derived, starting from Part \ref{sec:DPP}, where the DPP formulas for DC Hall type conductivities are presented. The DPP formulas reduce to Berry curvature integrals in the disorder-free limit.
Part~\ref{Part:Thermo I} derives the continued fraction conductivities from the moments expansion, and presents viable extrapolation schemes. 
Part \ref{Part:Thermo II} derives the Hall coefficient summation formulas. 

\begin{figure}[t]
\begin{center}
\includegraphics[width=8.5cm,angle=0]{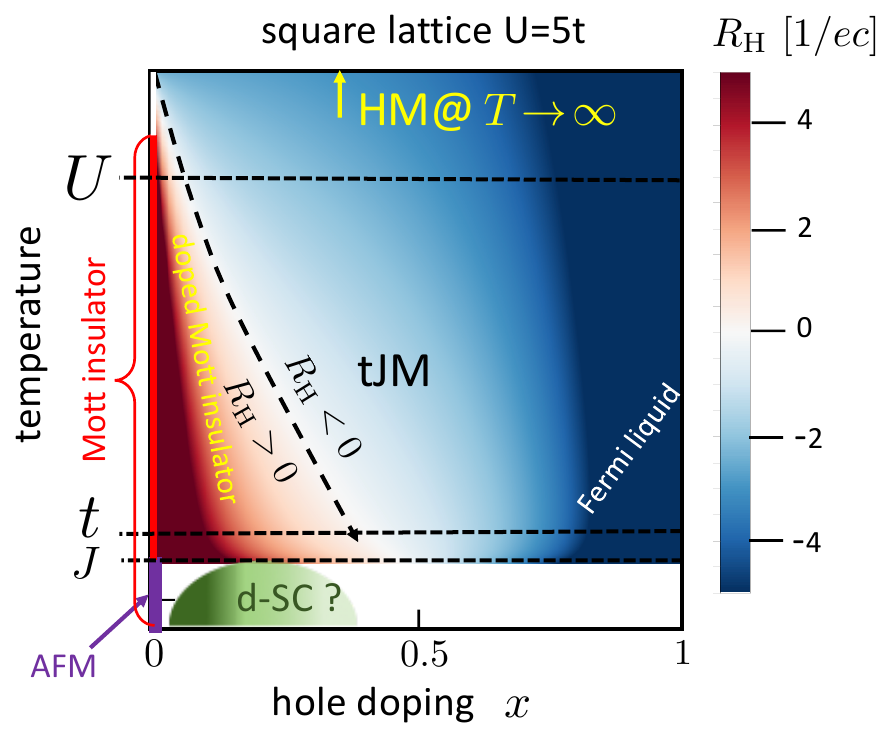}
\caption{{
Hall coefficient map  for the strongly interacting Hubbard model  on the square lattice.  The Hall coefficient was calculated~\cite{tJM} using the thermodynamic summation formula reviewed in Part~\ref{Part:Thermo II}.
The  model and its parameters are defined in Part~\ref{part:SCE}.   The anomalous (reddish-brown)  region near the Mott insulator is a consequence of the strong local electron-electron interactions which affect the current commutation relations and 
produce the sign reversal and divergence.
At low temperatures (beyond the present analysis),   we mark the  region of antiferromagnetic (AFM) order, and a hypothetical region of $d$-wave superconductivity (d-SC)~\cite{DMRG,Sorella,Troyer}. }}
\label{fig:HallMap-U}
\end{center}
\end{figure}

Part~\ref{part:SCE} demonstrates recent application~\cite{tJM} of the continued fraction conductivity and Hall coefficient summation formula to strongly interacting electrons near a Mott insulator, as modelled by the two dimensional Hubbard and t-J Hamiltonians. The resulting Hall coefficient map is depicted in Fig.~\ref{fig:HallMap-U}.
Part~\ref{part:SCB} demonstrates the application of the same formulas to obtain the resistivity and Hall coefficient of strongly interacting lattice bosons.

 Part~\ref{Part:Summary} summarizes the Report. A discussion of  strategies for optimizing the   thermodynamic approaches is given. Future applications and directions of research are proposed.

\part{Brief history of weak scattering}
\label{sec:BH}
Perturbative calculations of Kubo formulas  for DC conductivities  of metals suffer from the singular effect of scattering.
At leading order in impurity concentration  an infinite resummation of diagrams is needed to obtain a finite longitudinal conductivity at finite impurity concentration~\cite{Mahan}.  

Older (and simpler) alternatives for the weak scattering regimes are to apply the Drude and Boltzmann approaches.  In this Report,  unless otherwise specified, we use units of $\hbar=k_B=1$.
  
\section{Drude theory} 
Drude theory~\cite{AM} is  particularly useful for lightly doped semiconductors.
It is based on  a Fermi gas of electrons of mass $m^*$ and charge $e$ whose  kinetic energy is described by the single particle dispersion,
\be
\epsilon_\bk = {\bk^2\over 2m^*} \quad .
\ee
Collisions with disorder and other electrons are introduced by a scattering time $\tau$.
Solving for the single electron equation of motion in a time dependent field, the dynamical longitudinal conductivity is 
\be 
\sigma^{\rm Drude}_{xx}(\omega) =  {ne^2\over m^*} {\tau \over 1+(\omega\tau)^2} \quad .
\label{Sxx-Drude}
\ee
where $e$ is the electron charge.
The  DC  conductivities 
in a uniform magnetic field $\bB=B\hat{\bz}$ are
\bea 
\sigma^{\rm dc}_{xx} &=& {ne^2\over m^*} ~ {\tau\over 1+ (\omega_c \tau)^2} \quad ,\nonumber\\
\sigma^{\rm dc}_{xy} &=& {ne^2\over m^*}  ~  {\omega_c \tau^2 \over    1 + (\omega_c \tau)^2 } \quad ,
\label{S-Drude}
\eea
where $\omega_c= {e  B\over m^* c}$ is the cyclotron frequency,
and $c$ is the speed of light.
The zero field  Hall coefficient is defined as,
\be
R^{\rm Drude}_{\rm H} \equiv  -\lim_{B\to 0}  {\rho_{xy}(B)\over B}=  {1 \over \sigma_{xx}^2(0)}\left.{d\sigma^{\rm dc}_{xy} \over dB}\right|_{B=0} = {1 \over n e   c} \quad .
\label{RH-Drude}
\ee
Thus, $R_{\rm H}^{\rm Drude}$ is proportional to  the inverse charge density. We note that $R_{\rm H}^{\rm Drude}$, in contrast to $\sigma_{xy}$, is independent of dynamical parameters $m^*$ and $\tau$.
Later, in Part~\ref{Part:Thermo II},  the expression of the Hall coefficient in terms
of thermodynamic coefficients  is shown to be a general feature of non-separable Hamiltonians.

\section{Boltzmann equation} 
The single band Boltzmann equation (BE)~\cite{KL,Ziman,Girvin} for the quasiparticle distribution function deviation $\delta f_{\bk,\br}$ is based on small deviations from the  Fermi-Dirac distribution $f_0(\epsilon_\bk-\mu)$, where
 $\epsilon_{\bk}$ is the non-interacting
 Bloch band dispersion,  $\bk$ is a wavevector within the Brillouin zone (BZ), and $\mu$ is the chemical potential.
In the presence of  presence of an externally imposed  electrochemical field $\bE-{1\over e} \bnabla \mu$ and temperature gradient $ -\bnabla T$, the BE is,
\be
{\partial \delta f\over \partial t} + \dot{\bk}\cdot {\partial \delta f\over \partial \bk}+ \dot{\bx}_\bk \cdot \left[ -\bnabla\mu - {\epsilon_\bk-\mu\over T} \nabla T \right]\left(  {\partial f^0_\bk\over \partial \epsilon}\right)= {\cal I}^1_\bk [\delta f] \quad .
\label{BE-lin}
\ee 
where the semiclassical equations of motion are,
\bea
\dot{\bx}_\bk&=& \bnabla_\bk \epsilon_\bk -\dot{\bk}\times \bOmega_\bk\quad,\nonumber\\
\dot{\bk} &=& e\bE - {e\over c}\dot{\br}_\bk \times \bB\quad .
\label{Semiclass}
\eea
 The band and spin indices are suppressed. The band Berry curvature $\bOmega_\bk$, which modifies the velocity~\cite{Girvin,Niu-Sxy} is given by,
 \be
 \bOmega_\bk\equiv i \langle \bnabla_\bk u_\bk |\times |\bnabla_\bk u_\bk\rangle \quad ,
 \label{Berry-C}
 \ee
 where $|u_\bk\rangle$ is the periodic part of the Bloch state.
${\cal I}^1_\bk$ is the collision integral which is commonly simplified by the relaxation time approximation,
\be
{\cal I}^1_\bk [f]= -{\delta f_\bk\over \tau_\bk} \quad .
\ee
Weak electron-electron and electron-phonon interactions can be incorporated into BE by renormalizing the $\bv_\bk$,
and contributing to the quasiparticle scattering rate  $1/\tau_\bk$.

In the absence of time-reversal symmetry breaking in the equilibrium density matrix, the electric and thermal currents are given respectively by
\bea
\bj&=&  \sum_{\bk \in \BZ} \bv_{\bk} \ \delta f_\bk \quad , \nonumber\\
\bj_Q&=&  \sum_{\bk \in \BZ} (\epsilon_\bk-\mu)\bv_{\bk} \ \delta f_\bk \quad ,
\label{BE-currents}
\eea
where $\BZ$ is the Brillouin zone.
Using the solution of Eq.~(\ref{BE-lin}) in (\ref{BE-currents}), for $C_4$ symmetric bands, yields  BE expression for the DC longitudinal and Hall conductivities,
\bea
\sigma_{xx}^{\rm dc} &=& {e^2\over c \cV} \sum_\bk  \left(-{\partial f\over \partial \epsilon}\right) (v_\bk^x)^2 \tau_{\bk} \quad ,\nonumber\\
\sigma^{\rm dc}_{xy} &=&  { e^3 B\over c \cV} \sum_{\bk \in \BZ}   \left(- {\partial f_\bk^0 \over \partial \epsilon} \right)  v^y_\bk \tau_{\bk}  \left( v^y_\bk {\partial \over \partial k_{x}} - v_\bk^{x} {\partial \over \partial k_{y}} \right) (v_\bk^{x} \tau_{\bk}) \quad .
\label{Sab-Boltz}
\eea
For an isotropic (energy dependent)  scattering time $\tau_\bk=\tau(\epsilon_\bk)$, one obtains simplified expressions at low temperatures relative to the Fermi energy $\epsilon_F$, 
\bea
\sigma_{xx}^{\rm dc} &=& \tau(\epsilon_F)~\chi_{\rm csr} \quad , \nonumber\\
\sigma_{xy}^{\rm dc} &=&   
\tau^2(\epsilon_F)~ {e^3\over c \cV}  \sum_{\bk \in \BZ}   \left( -{\partial f_\bk^0 \over \partial \epsilon} \right) \left( v_\bk^{y} \left( v_\bk^{y} {\partial \over \partial k^{x}} - v_\bk^{x} {\partial \over \partial k^{y}}\right)  v_\bk^{x}\right) \quad ,
\label{Sab-tauconst}
\eea
where the conductivity sum rule (CSR) is,
\be
\chi_{\rm csr} = 
{e^2\over \cV} \sum_{\bk \in \BZ}  \left(- {\partial f_\bk^0 \over \partial \epsilon} \right)  (v_\bk^x)^2 \quad .
\label{CSR-NI}
\ee

The Hall coefficient acquires a scattering time independent expression~\cite{Ziman},
\be
R_{\rm H}^{\rm Boltz} = {e^3 \over   c \chi_{\rm csr}^2 \cV }    \sum_{\bk}   \left( -{\partial f_\bk^0 \over \partial \epsilon} \right) \left( v_\bk^{y} \left( v_\bk^{y} {\partial \over \partial k^{x}} - v_\bk^{x} {\partial \over \partial k^{y}}\right)  v_\bk^{x}\right) \quad .
\label{RH-Boltz}  
\ee
 Eq.~(\ref{RH-Boltz}) generalizes Drude's result~(\ref{RH-Drude}) to non-spherical Fermi surfaces.

 \subsection{Example: The square lattice}
 \label{sec:SL}
The square lattice (SL) tight binding model is,
\be
H^{\rm SL}=-\sum_{\ij,s=\uparrow,\downarrow}t_{ij} (c^\dagger_{is}c^{\nd}_{js} +c^\dagger_{js}c^{\nd}_{is})   \quad ,
\label{SL}
\ee
where $c^\dagger_{is}$ creates an electron on site $i$ with spin $s$. $\ij$ are nearest neighbor bonds on the square lattice (SL), and electron occupation per site is $n_{is}= c^\dagger_{is} c^{}_{is}$.

The Hall coefficients given by Eq.~(\ref{RH-Boltz}) and related
band structure contours are depicted in Fig.~\ref{fig:SL-RH} for nearest neighbor (nn) and next nearest neighbor (nnn) model.

\begin{figure}[h!]
\begin{center}
\includegraphics[width=8cm,angle=0]{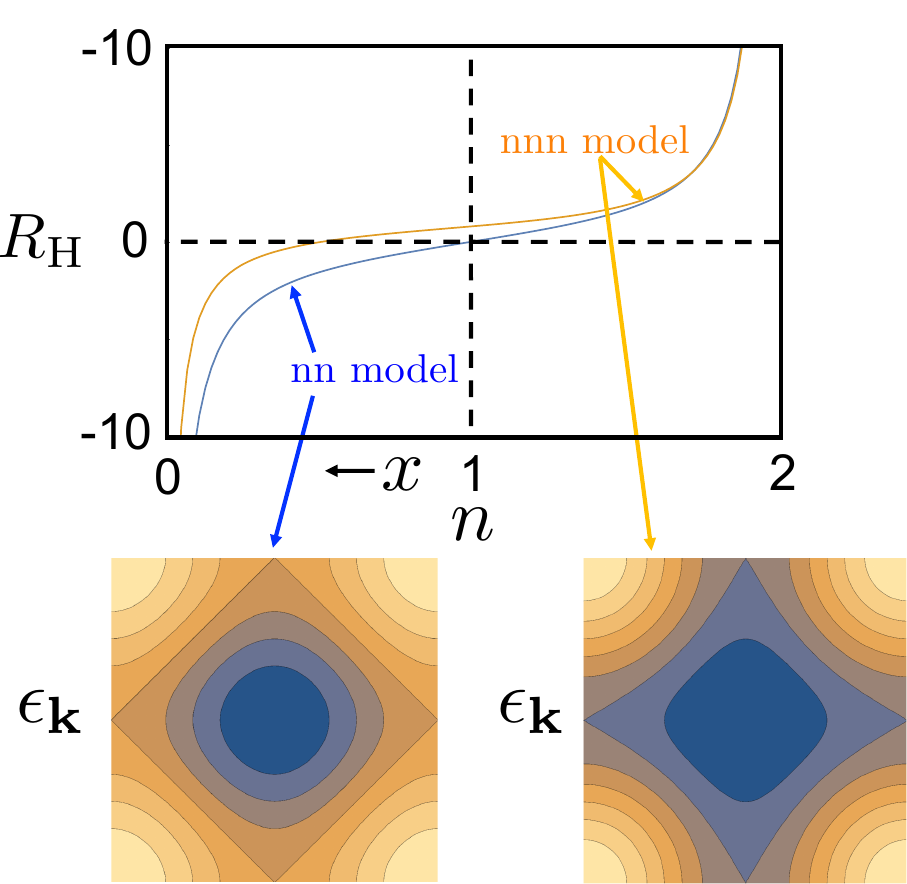}
\caption{ 
Hall coefficient $R_{\rm H}$, as given by Boltzmann equation   of the weakly scattered ($U=0$) square lattice, given by Eq.~(\ref{RH-Boltz}) at $T=0$. 
 $n$ is the electron filling, and $x$ is the hole doping
concentration customarily used in describing cuprate superconductors. The nearest neighbor (nn) model has particle-hole antisymmetry about $n=1$.  Next nearest neighbor (nnn) hopping moves the Hall sign reversal toward lower electron density. }
 \label{fig:SL-RH}
\end{center}
\end{figure}
The most important prediction of Boltzmann theory is that  the Hall coefficient is everywhere continuous and diverges only toward the full and empty band limits. 
\subsection{Conductivity relations at low temperatures}
For  $C_4$ symmetric bands, the isotropic lifetime leads to simple relations between for the longitudinal electric, thermoelectric and thermal conductivities, which can be expanded at low temperatures as
\bea
\sigma^{\rm dc}_{xx}(T) &=&  \sigma^{\rm dc}_{xx}(0)  \left(  1 + {\pi^2 \over 8}{(k_B T)^2 \over \mu^2} + ...\right) \\ \nonumber
\alpha^{\rm dc}_{xx}(T) &=&  {\sigma^{\rm dc}_{xx}(0)  \over e^2} {\pi^2 \over 2} {(k_B T)^2 \over \mu} + ... \\ \nonumber
\kappa^{\rm dc}_{xx}(T)  &=& {\sigma^{\rm dc}_{xx}(0)  \over e^2} {\pi^2 \over 3} (k_B T)^2 + ...
\eea
As $T \rightarrow 0$, one obtains the 
 Wiedemann-Franz law
 \be
{\kappa \over \sigma T} = {\pi^2 \over 3} (k_B/e)^2 = 2.45 \times 10^{-8} V^2K^{-2} \quad .
\label{WF}
\ee
\section{Memory Function Formalism}
The Memory Function (MF) formalism~\cite{Mori,Zwanzig,mem-das} is closely related to the Operator Hilbert Space (OHS) formulation of dynamical response functions as described in Section \ref{sec:OHS}. 

A primary goal of the MF was to obtain the dynamical response of a set of  ``slow operators''  which nearly commute with the Hamiltonian, by integrating out all the ``fast operators'' to obtain 
a self energy, called `memory matrix'' $M(z)$.
G\"{o}tze-W\"{o}lfle~\cite{GW}  have applied the MF to evaluate the longitudinal conductivity in a weakly scattered parabolic band of effective mass $m^*$,
\be
\sigma_{xx}(\omega) = {ne^2\over m^* } \Re {i \over z - M(z) / {ne^2\over m^* }} \quad ,
\ee
where $M(z)$ is the current-current correlation function. They consider  a white noise random potential $V$, whose ensemble average is
\be
\langle V_\bq V_{-\bq'}\rangle_{\rm dis} =w^2 \delta_{\bq,-\bq'} \quad ,
\ee
and its commutator with the current is 
\be
[j^x, V_\bq]= -i{\hbar\over m} q_x V_\bq \quad .
\ee
At low temperatures, MF recovers Drude result (\ref{Sxx-Drude}),
\be
\sigma_{xx}={  ne^2 \tau \over m^*} \quad .
\ee
The MF scattering rate agrees with   Fermi's golden rule,
\be
{\hbar \over \tau}  =  \pi {\cal N}(\epsilon_F) w^2 \quad ,
\label{FGR-tau}
\ee 
 where ${\cal N}(\epsilon_F)$ is the density of electron states at the Fermi energy.

More generally, the MF separation of timescales can facilitate obtaining the low temperature and frequency conductivities
by renormalization of the dynamical response functions onto the low energy Hilbert space.

In Part~\ref{Part:Thermo I}, the memory function $M(z)$ can be related to the first order termination function $\Delta_1^2 G^>_{11}(z)$ in the continued fraction of the longitudinal conductivity.

\section{Limits of weak scattering approaches}
BE describes transport by quasiparticles in the conduction band, with velocity $\bv_\bk$ at wavevector $\bk$  near the Fermi surface. 
For elastic scattering in dimensions one and two,  BE is invalidated by wavefunction localization~\cite{PWA,Go4,Loc-Abrahams,Eric}. The collision integral of Eq.~(\ref{BE-lin}) describes incoherent scattering processes which do not lead to localization.  

The validity of BE depends on the existence of well defined quasiparticles. Inelastic scattering broadens the quasiparticle energies
and wavevectors by $1/\tau$  and $2\pi/l$. The BE therefore  requires 
\be
\epsilon_F \tau  \gg 1 \quad ,\quad k_F l \gg 1 \quad .
\label{ROV}
\ee
The criteria (\ref{ROV}) have been related by Ioffe and Regel (IR)~\cite{IoffeRegel} to experimental values of the resistivity.
Using the simplified Drude theory of a parabolic band, $\epsilon_F = \hbar^2 k_F^2 / 2 m^*$, and $n=k_F^3 / 3\pi^2$, and by (\ref{Sxx-Drude}), one obtains an upper bound on  value of resistivity which can be explained  by Boltzmann equation,
\bea
\rho_{xx}&=&  {m^*\over n e^2 \tau} =  {3 \lambda_F\over 2}   {  h\over e^2} {1 \over k_F l}  \le \rho_{xx}^{\rm IR}\big|_{k_F l=1} \quad .\eea
For typical cubic metals with $\lambda_F \sim  1$nm,  
the IR limit is  $\rho_{xx}^{\rm IR} \simeq 40 \mu \Omega ~\mbox{\rm cm}$.

In metals, $\rho_{xx}$ usually increases with temperature
due to enhanced inelastic scattering. Many metals (see review in \cite{IR-Hussey}) exhibit resistivity saturation at the values not far above the IR limit.  

However, resistivity in certain  strongly correlated metals exceeds  this limit as temperature is raised, which has been dubbed ``bad metal'' behavior~\cite{badmetals}.   Boltzmann equation fails to account for this regime.

In semimetals and semiconductors with a small  interband gap $\Delta_{\rm ib}$, interband matrix elements of the current  must be included when
\be
\Delta_{\rm ib} \tau <1 \quad .
\ee
This leads to a multi-band BE, which involves coupled equations for intra-band and inter-band distribution functions~\cite{Bulbul,Culcer,Tserkov}. The  equations are in general unwieldy. In the presence of disorder and electron-phonon scattering, a microscopic knowledge of inter-band temperature dependent scattering rates,
and inter-band current matrix elements, is required. 
\newpage
\part{Kubo Formulas} 
\label{sec:KF}

\section{Polarizations and Currents}
Kubo's formulation~\cite{Kubo} of dynamical linear response, provides a rigorous approach to calculating conductivities of a microscopic model Hamiltonian $H_0$ which is supported on a finite $d$-dimensional volume $\cV$. 

The underlying assumption of Kubo's linear response theory is that local equilibrium is established throughout the sample. In other words, the currents' equilibration length and time scales are much shorter than those of the driving fields.

Following Luttinger~\cite{Cooper,Luttinger,Shastry-SR},  it is also assumed that the same electric and thermal currents which are driven by slowly varying {\em statistical }forces, such as $-\bnabla\mu$ and $-\bnabla T$, can be induced by {\em mechanical} electric  and ``gravitational'' forces, $\bE$ and $\bpsi$ respectively, which  couple linearly to the Hamiltonian, 
\be
H(t) = H_0 -  \Theta(t) \int {d^dq\over (2\pi)^d} \left(\bP_\bq \bE_{-\bq}(t) +  {1\over T} \bQ_\bq \bpsi_{-\bq}(t) \right) \quad .
\label{Couplingtofields}
\ee 
The polarizations which couple to the external force fields are,
\bea
\mbox{\rm Electric polarization}:&\quad &  \bP_\bq=  e \int_\cV d^d x~ e^{-i\bq\cdot \bx} ~\bx~  n(\bx) \quad ,\nonumber\\
\mbox{\rm Thermal polarization}:&\quad &  \bQ_\bq =   \int_\cV d^d x~ e^{-i\bq\cdot \bx}~ \bx ~\Big(h(\bx)-\mu n(\bx)\Big) \quad .
\label{Pol-q}
\eea
$n(\bx)$ is particle density and $h(\bx)$ is a local decomposition of the Hamiltonian which satisfies,
\be
H_0=\int d^d x~h(\bx) \quad .
\ee
The Kubo formulas for conductivities require knowledge of $H_0$ and its polarization operators of Eq.~(\ref{Pol-q}). 
The electric and thermal currents are derived by Heisenberg's equations,
\bea
\bj_\bq&=&     i\left[ H_0, \bP_\bq\right] \quad ,\nonumber\\
(\bj_Q)_\bq&=&     i\left[ H_0, \bQ_\bq\right] \quad .
\label{JP}
\eea
We note that the polarizations (\ref{Pol-q}) depend  on the choice of coordinate $\bx$.   However, any finite shift in the origin $\bx\to \bx+\ba$  drops out of the commutators in Eqs.~(\ref{JP}).
For PBC,  uniform polarizations can be defined as a $\bq\to 0$ limit after taking  volume to infinity, as  shown later in Eqs.~(\ref{Pol0-PBC}). 

\section{Examples: Hamiltonians and polarizations}
The explicit forms of polarizations are shown for five generic models of many particle systems.
\begin{enumerate}
    \item  The Hamiltonian of $N_p$ interacting Schr\"odinger particles (bosons or fermions) in first quantization notation is,
\be 
H_0^{\rm particles}=\sum_{i=1}^{N_p} {\bp^2_i\over 2m}+ \sum_i V(\bx_i) + \sum_{i<j} U(\bx_i-\bx_j) \quad ,
\label{Hparticles}
\ee
where $[\bx^\alpha_i,\bp_j^\beta]=i\hbar \delta_{ij}\delta_{\alpha\beta}$.
The  polarizations are given by,
\bea
 P^\alpha_\bq &=&      -|e|  \sum_{i=1}^{N_p} \bx^\alpha_i e^{-i\bq\cdot \bx_i} \quad ,\nonumber\\
 Q^\alpha_\bq  &=&       \sum_{i=1}^{N_p} {1\over 2}\left\{ x^\alpha_i e^{-i\bq\cdot \bx_i},{\bp^2_i\over 2m}\right\} + \sum_i x^\alpha_i e^{-i\bq\cdot \bx_i} \left(V(\bx_i) + \sum_{i<j} U(\bx_i-\bx_j)\right) \quad .
\label{Pol-particles}
\eea

\item Particles on a lattice $L$ with electric charge $e^*$, local occupation $n_i$, and hermitian two-site interaction terms:
\bea
H_0^{\rm lattice} &=&  \sum_{ij\in L}~O_{ij} \quad ,\nonumber\\
P_\bq^\alpha &=& -|e^*|\sum_{i\in L} e^{-i\bq\cdot\bx_i} x_i^\alpha n_i \quad ,\nonumber\\
Q_\bq^\alpha &=&  \sum_{i\in L}  e^{-i\bq\cdot\bx_i} x_i^\alpha \left(\left(\sum_{j\in L} O_{ij}\right) - \mu n_i\right) \quad .
\label{Hlatt}
\eea

\item General non-interacting (NI) normal Hamiltonians in second quantized form as
\be
H^{\rm NI}=\sum_{l}  \epsilon_l    a^\dagger_l a^{}_{l} \quad ,\quad [a_l,a^\dagger_{l'}]_\pm=\delta_{ll'} \quad .
\label{Hsp}
\ee
where $ a^\dagger_{l}$ creates a particle of charge $e^*$ are in single-particle eigenstate   $|l\rangle$ with energy $\epsilon_l$. $[\bullet , \bullet]_{\pm}$ denotes an anticommutator (commutator) for fermions (bosons).

The  polarizations are given by the bilinear forms,
\bea
P_\bq^\alpha &=& i e^*  {\partial\over \partial q^\alpha}\sum_{ll'} \langle l|e^{-i\bq\cdot\bx} |l'\rangle  a^\dagger_l a^{}_{l'} \quad ,\nonumber\\
Q_\bq^\alpha &=& i{\partial\over \partial q^\alpha}  \sum_{ll'}   {\epsilon_l+\epsilon_{l'}-2\mu\over 2}
\langle l|e^{-i\bq\cdot\bx} |l'\rangle  a^\dagger_l a^{}_{l'} \quad .
\label{Pol-l}
\eea

\item Non-interacting  normal Hamiltonians of bosons or fermions with $M$ basis states $\{|\bR,i\rangle\}$ per unit cell at lattice site $\bR\in L$:
\be
\langle\bR,i|\bR',i'\rangle=\delta_{\bR,\bR'}\delta_{i,i'} \quad .
\ee
In a periodic crystal (PC), the Hamiltonian is given by
\bea
H^{\rm PC} &=& \sum_{\bk\in \BZ}\sum_{i j=1}^M\left(h_{ij}(\bk)-\mu\delta_{ij} \right) a^\dagger_{\bk,i} a^{}_{\bk, j} \quad .
\label{Hperiodic}
\eea
$a^\dagger_{\bk,i}|0\rangle=|\bk,i\rangle$ creates a   fermion (boson) state  of wavevector $\bk\in \BZ$ of lattice $L$, and basis state $i$.
After taking $\cV\to \infty$,  the polarization operators of Eq.~(\ref{Pol-l}) may be represented by continuous derivatives with respect to $\bq$,
\bea
 {P}^\alpha_\bq &=& e     \sum_{\bk ij} a^\dagger_{\bk i} ~(i\bnabla_{\bq}^\alpha  a^{}_{\bk+\bq j}) \quad ,\nonumber\\
 {Q}^\alpha_\bq   &= &   \sum_{\bk ij} a^\dagger_{\bk i}\left\{ \left({h_{ij}(\bk)+h_{ij}(\bk+\bq)\over 2}-\mu\delta_{ij} \right),  i\bnabla_{\bq}^\alpha  \right\}a^{}_{\bk+\bq j} ~ \quad .
\label{Pol-periodic}
\eea

\item 
Coupled Harmonic Oscillators (CHO) Eq.~can describe collective bosonic modes such as phonons~\cite{QZS} and magnons~\cite{Nagaosa-Lee} in an insulator.  A general CHO Hamiltonian which linearly couples to an external orbital magnetic field is,
\be
H^{\rm CHO} = {1\over 2} \sum_{i} { p_i^2 \over m_i} + {1\over 2} \sum_{i,j}u_i D_{ij} u_j +  \sum_{\alpha=x,y,z} B^\alpha  \cdot \sum_{i,j } p_i M^\alpha_{ij} u_j \quad ,
\label{CHO}
\ee
where $i \to \bx_i,s(i)$ denotes both site and polarization indices, where  and   $[u_i ,p_j]=i \delta_{ij} $ are  canonically conjugated coordinates. $m_i$ and $D_{ij}$ are  local mass and force constants. $H^{\rm CHO}$ can include lattice imperfections, impurities with different masses $m_i$, and boundaries. The magnetic field $\bB$ breaks time reversal symmetry by coupling between $u_i$ and $p_j$ as represented by a magnetization matrix $\bM$.
The thermal polarization is,
\be
Q_\bq^\alpha=\sum_{i=1}^N e^{i\bq\cdot\bx_i} x^\alpha_i\left(  { p_i^2 \over 2m_i} + {1\over 2} u_i \sum_{j}
D_{ij} u_j +  {1\over 2} p_i \sum_{j}\sum_{\alpha=x,y,z} B^\alpha  \cdot \sum_{j} M^\alpha_{i,j} u_j  \right) \quad .
\label{TP-CHO}
\ee
The thermal current is obtained by Eq.~(\ref{JP}).
Using second quantized operators,
\be
a_i = {1\over \sqrt{2}}\left(u_i + i p_i\right),\quad a^\dagger_i = {1\over \sqrt{2}}\left( u_i - i p_i\right),\quad [a_i,a^\dagger_j]=\delta_{ij} \quad ,
\label{ladder}
\ee
the Hamiltonian is written in a duplicated Bogoliubov form,
\be
H^{\rm CHO} =  {1\over 2}  \sum_{ij=1}^N (a^\dagger_i,a_i)  H_{ij} \left( \begin{array}{c} a_j \\ a^\dagger_j\end{array}\right)  \equiv {1\over 2}  \sum_{ij} (a^\dagger_i,a_i)  \left( \begin{array}{cc} H^N_{ij} &  H^A_{ij} \\ (H^A_{ij})^*  &(H^N_{ij} )^* \end{array}\right)  \left( \begin{array}{c} a_j \\ a^\dagger_j\end{array}\right)  + {\rm const} \quad ,
\label{H-dup}
\ee
where the constant comes from the ordering  the operators $a,a^\dagger$. The duplicated form allows us to choose $H^N$ to be hermitian, and $H^A$ to be a symmetric matrix.

$H^{\rm CHO}$ can be diagonalized by a Bogoliubov transformation~\footnote{We thank Dan Arovas for sharing with us his impeccable notes} defined by a symplectic matrix $S_{in}$ 
of size $2N\times 2N$,
\be
  \left( \begin{array}{c} a_i \\ a^\dagger_i \end{array}\right)  =  \sum_{n=1}^N S_{in}  \left( \begin{array}{c} b_n  \\ b^\dagger_n\end{array}\right)
 =  \sum_{n=1}^N  \left( \begin{array}{cc} U_{in} &  V_{in}^* \\ V_{in}  &U_{in}^*  \end{array}\right) 
 \left( \begin{array}{c} b_n  \\b^\dagger_n\end{array}\right)
\label{S-mat} \quad ,\ee
where $[b_n,b^\dagger_{n'}]=\delta_{nn'}$, since
\be
S^\dagger J S =  S J S^\dagger =J,\quad J\equiv  \left( \begin{array}{cc} \mathbb{I}_N  &0 \\ 0  &-\mathbb{I}_N  \end{array}\right) 
\label{CCR-S} 
\ee 
and  $\mathbb{I}_N$ is a unit matrix of size $N$.
Using $S$, 
\be
H^{\rm CHO} =  {1\over 2}  \sum_{ij} (b^\dagger_n,b_n)  \left( \begin{array}{cc} \ve_n& 0 \\0  &\ve_n  \end{array}\right) \left( \begin{array}{c} b_n  \\ b^\dagger_n \end{array}\right) \quad . 
\ee
Using the relation $S^{-1} =   J S^\dagger J$,
$S$ is determined from $H^{\rm CHO}$ by the equations,
\bea
S^\dagger H S &=& \left( \begin{array}{cc} \ve & 0 \\ 0  &\ve \end{array}\right) \nonumber\\
& \Rightarrow & H S = (S^\dagger)^{-1} \left( \begin{array}{cc} \ve & 0 \\ 0  &\ve \end{array}\right) =JSJ  \left( \begin{array}{cc} \ve & 0 \\ 0  &\ve \end{array}\right)  \nonumber\\
&&J  H S = S \left( \begin{array}{cc} \ve & 0 \\ 0  &-\ve \end{array}\right) \quad .
\eea
$U,V$ are readily obtained numerically by computing the  {\em  right eigenvectors}  of the (non hermitian) matrix $JH$, and retaining those which belong to the positive spectrum $\ve_n>0$. (Existence of complex eigenvalues is possible: they reflect an instability of the harmonic Hamiltonian). The upper left block yields the eigenvalue equation, for $n=1,\ldots N$:
\bea
 && \sum_{j=1}^N H^N_{ij} U_{jn} + H^A_{ij}V_{jn} = U_{in}  \ve_n\nonumber\\
&& \sum_{j=1}^N -(H^A_{ij} )^*V_{jn} - (H^N_{ij})^* U_{jn} =  V_{in}  \ve_n \quad ,
 \eea
 where the  eigenvectors $(U_{in},V_{in})$ must be  normalized as,
\be
\forall n\in[1,\ldots N] : \quad \sum_{i=1}^N\left( |U_{in}|^2-|V_{in}|^2\right) = 1 \quad .
\ee

Using Eqs.~(\ref{ladder}) and (\ref{S-mat}),  $(p_i,u_i)$
are transformed into the eigenmode representation $b^\dagger_n,b_n$,     the thermal polarization (\ref{TP-CHO}) is given by,
\be
\bQ_\bq  = {1\over 2}  \sum_{n,n'=1}^N (b^\dagger_n,b_n)  \left( \begin{array}{cc}  \bQ^N_{ \bq}&  \bQ^A_{\bq}  \\(\bQ^A_{\bq})^*  &(\bQ^N_{\bq})^* \end{array}\right) \left( \begin{array}{c} b_{n'}  \\ b^\dagger_{n'} \end{array}\right) \quad .
\label{Q-CHO}
\ee
 
\end{enumerate}

\section{Kubo formulas in Lehmann Representation}
Henceforth we use unified notations
\bea
(\bP,\bQ)&\to& (\bP_1,\bP_2)\nonumber\\
(\bj,\bj_Q)&\to& (\bJ_1,\bJ_2) \quad .
\eea
In the Lehmann (eigenstate) representation, the currents' response functions are given by the sums,
\bea
 L^{\alpha\beta}_{ij}(\bq,\omega+i\ve) &=&- {i\over \cV} \sum_{mn} {\rho_m-\rho_n \over E_n-E_m}  ~\left({\langle n|(J^\alpha_i)_{-\bq}|m\rangle 
\langle n|(J^\beta_j)_{\bq}|m\rangle\over E_n- E_m - \omega -i\ve}\right)\nonumber\\
&=& {1\over \cV} \sum_{mn} {\rho_m-\rho_n \over E_n-E_m}  ~\Im \left({\langle n|(J^\alpha_i)_{-\bq}|m\rangle 
\langle n|(J^\beta_j)_{\bq}|m\rangle\over E_n- E_m - \omega -i\ve}\right)\nonumber\\
&&~~~-  {i\over \cV} \sum_{mn} {\rho_m-\rho_n \over E_n-E_m}  ~\Re \left({\langle n|(J^\alpha_i)_{-\bq}|m\rangle 
\langle n|(J^\beta_j)_{\bq}|m\rangle\over E_n- E_m - \omega -i\ve}\right) \quad .
\label{Kubo-Lehmann}
\eea
where $E_m$ and $|m\rangle$ are the eigenenergies and eigenstates of
$H_0$, with  $\rho_m=e^{-\beta E_m}/Z$ as Boltzmann's weights. 
$\alpha,\beta\in (x,y,z)$ are the Cartesian components of the currents.

The complex uniform ($\bq=0$) dynamical conductivities which correspond to the transport equations  (\ref{Tr0}), are calculable by the Kubo formulas,
\bea
  &\mbox{\rm electric conductivity}  : &  \sigma_{\alpha\beta}(\omega)=  \lim_{\bq\to 0}\lim_{\cV\to \infty}  ~ \Re L^{\alpha\beta}_{11}(\bq,\omega+i\ve) \quad ,\nonumber\\
 &\mbox{\rm thermoelectric conductivities} : &   \alpha_{\alpha\beta}(\omega)={1\over T}\lim_{\bq\to 0}\lim_{\cV\to \infty}    \Re  L^{\alpha\beta}_{12}(\bq,\omega+i\ve) \quad ,\nonumber\\
 &&\bar{\alpha}_{\alpha\beta}(\omega)={1\over T} \lim_{\bq\to 0}\lim_{\cV\to \infty}  \Re  L^{\bq,\alpha\beta}_{21} (\omega+i\ve) \quad ,\nonumber\\
&  \mbox{\rm thermal conductivity}  : &   \kappa_{\alpha\beta}(\omega)={1\over T} \lim_{\bq\to 0}\lim_{\cV\to \infty}  \Re L^{\alpha\beta}_{22}(\bq,\omega+i\ve) \quad .
  \label{Lsigma}
 \eea

Straightforward computation of Eqs.~(\ref{Kubo-Lehmann}) for general many-body Hamiltonians is generally a daunting task. Exact diagonalizations (ED) of $H_0$  may increase exponentially with $\cV$ even for  a single eigenstate and eigenenergy. The difficulty is compounded by the apparent necessity to compute many current matrix elements. 

\subsection{Non-interacting conductivities}
ED is of course more manageable for non-interacting bosons and fermions. For the single particle Hamiltonian  (\ref{Hsp}) and polarizations  Eq.~(\ref{Pol-l}), using Eq.~(\ref{JP}), the current matrix elements are,
\bea
\langle l | j_\bq^\alpha |l'\rangle &=& ie (\epsilon_l-\epsilon_{l'})
\langle l | x^\alpha e^{i\bq\cdot \bx} |l'\rangle \quad ,\nonumber\\
\langle l | (j_Q)_\bq^\alpha |l'\rangle &=&  i (\epsilon_l-\epsilon_{l'}) {(\epsilon_l+\epsilon_{l'}-2\mu)\over 2}\langle l | x^\alpha e^{i\bq\cdot \bx} |l'\rangle \quad .
\label{Curr-l}
\eea
The Kubo formulas reduce to  sums over single particle eigenstates,
\be
L_{ij}^{\alpha\beta}  (\bq,\omega + i\ve)  =  -i {1    \over \cV}
 \sum_{ll'} {n(\epsilon_{l'} ) - n(\epsilon_{l})\over
 \epsilon_l- \epsilon_{l'} }
~\left({ \langle l|(J_i^\alpha)_{ -\bq} |l'\rangle\langle l'|(J_j^\beta)_{ \bq} |l\rangle \over \epsilon_l- \epsilon_{l'} -\omega-i\varepsilon} \right) \quad ,
\label{KuboAC-LehmannSP1}
\ee
where
\be
n(\epsilon)={1\over e^{\beta(\epsilon-\mu)}\pm 1 } \quad ,
\ee
for fermions (bosons) with a plus (minus) sign.

\section{The Tricky DC limit}
\label{sec:DC}
DC transport coefficients require taking the orders of limits carefully.  Since for any  time independent Hamiltonian on a finite volume ($\cV<\infty$) the density matrix is in equilibrium. By definition, it cannot support any dissipative (entropy generating) steady-state transport currents.  As argued by Luttinger~\cite{Luttinger}, the DC steady state (which he called ``rapid case'') can be achieved if the driving force satisfies $|\omega +i\ve| > |\bq|^2$, as both $\omega+i\ve ,|\bq|$ are taken to zero, after we have taken  $\cV\to \infty$ to eliminate finite size gaps in the continuous thermodynamic spectrum.

For thermal transport, the statistical field $-\bnabla T$ is replaced by  Luttinger's frequency dependent ``gravitational'' field $\bpsi( \omega)$~\cite{Luttinger}. The legality of this substitution has been widely accepted~\cite{Cooper,Shastry-SR}, although its rigorous conditions are still debated~\footnote{Challenges to the equality of the current response to ``mechanical forces'' and ``statistical forces'' have been made by e.g. Ref~\cite{Simon-arg}.  }.

In taking the DC limit, one must contend with the (superfluous) effects of the static gravitational field $\bpsi_{\omega=0}$. This force field may create  an equilibrium  circulating ``magnetization current'' in any finite system, which is not part of the transport current~\cite{Cooper}.  It contributes to  $L^{\alpha\beta}_{ij}(\omega+i\ve=0)$ part of the Kubo formula which must  be subtracted out {\em before} taking the DC limit.

In summary, the proper DC order of limits is given by
\be
(L^{\alpha\beta}_{ij})^{\rm dc}= \lim_{\omega+i\ve \to i0^+} \lim_{\bq\to 0} \lim_{\cV\to \infty} ~\Re\left( L_{ij}(\bq,\omega+i\ve,\cV)-L_{ij}(\bq,0,\cV)\right) \quad .
\label{DC}
\ee
The causal decay of the real-time response function ensures that $L_{ij}(z)$ is analytic in the upper half plane ($\Im(z)>0$). Therefore the limit
$\omega+i\ve\to 0$ can be taken by a-priori setting $\omega=0$, and sending $\ve\to 0^+$  {\em after} sending $\cV\to \infty$.

On finite volume, we can distinguish between OBC and PBC.
For  OBC, the uniform limit $\lim_{\bq\to 0} (\bP_i)_\bq$ can be taken continuously, since the driving field is not required to be periodic between the boundaries.
The DC limit is then 
simplified further to, 
\be
(L^{\alpha\beta}_{ij})^{\rm dc,OBC}= \lim_{\ve \to i0^+}  \lim_{\cV\to \infty} ~\Re\left( L_{ij}(0,i\ve,\cV)-L_{ij}(0,0,\cV)\right) \quad ,
\label{DC1}
\ee
where the uniform polarizations are given by 
Eqs.~(\ref{Pol-q}) by setting $\bq=0$ on any finite volume.
In practice, as demonstrated in   
in Part~\ref{sec:DPP}, Eq.~(\ref{DC1}) can be implemented by computing $L_{ij}(i\ve(\cV),\cV)$ on an increasing sequence of volumes $\{\cV_i\}$,  by choosing 
the  $\ve(\cV_i)$ to be larger than the finite-size 
eigenenergy gaps in the Kubo formula.

On finite PBC lattices, the force fields must be continuous, and therefore their wavevectors $\bq$ are discretized. Therefore, in contrast to OBC,  the order of limits on PBC must be taken by Eq.~(\ref{DC}). Furthermore, since the uniform polarization  cannot be taken continuously on $\cV<\infty$, the limit is taken after $\cV\to \infty$, i.e. 
\bea
\bP^\alpha_{\bq=0} &=&  i e \lim_{\bq\to 0}
\lim_{\cV\to \infty}   {1\over   q^\alpha} \int_\cV d^d x~ e^{-i\bq\cdot \bx}  n(\bx) \quad ,\nonumber\\
 \bQ^\alpha_{\bq=0} &=&  i  \lim_{\bq\to 0} 
\lim_{\cV\to \infty} {1\over   q^\alpha} \int_\cV d^d x~ e^{-i\bq\cdot \bx}  ~\Big(h(\bx)-\mu n(\bx)\Big) \quad .
\label{Pol0-PBC}
\eea
An alternative to using the uniform polarization,
the uniform {\em electric} current on a finite volume PBC can be defined as the derivative of the free energy with respect to an enclosed Aharonov-Bohm flux, see Eq.~(\ref{Der-AB}). No such definition exits for  the thermal current.

\section{Onsager relations}
\label{sec:Onsager}
In the absence of spontaneous time reversal (TR) symmetry breaking, 
TR transformation of  Kubo formulas for $L_{ij}$ involves reversal of the external magnetic field and transposition of  the two currents. The conductivities  satisfy  Onsager's reciprocal relations~\cite{Onsager1,Onsager2},
\bea
 L^{\alpha\beta}_{ij}(\bB)  =   L^{\beta\alpha}_{ji}(-\bB)  \quad .
\label{Onsager-Lij}
\eea
As a consequence, the longitudinal conductivities $L^{\alpha\alpha}_{ii}$ are   even functions of magnetic field $\bB$. For  $\bB\parallel \hat{\bz}$, the  transverse conductivities can be written as,
\be
L^{xy,\pm}={1\over 2} \left(L^{xy}\pm L^{yx}\right) \quad ,
\ee
$L^{xy,-}$ are the {\em Hall conductivities}  which,
by Eq.~(\ref{Onsager-Lij}),  are antisymmetric with respect to reversal of   magnetic field:
\bea
&&\sigma_{xy,-}(\bB)=-\sigma_{xy,-}(-\bB) \quad ,\nonumber\\
&&\kappa_{xy,-}(\bB)=-\kappa_{xy,-}(-\bB) \quad .
\label{HallCond}
\eea
Eqs.~(\ref{HallCond}) allow experimenters to measure the Hall conductivities
in a four-probe bar geometry, by reversing the applied magnetic field.

For general crystals, $\alpha_{xy,-}\ne \bar{\alpha}_{xy,-}$. 
However, if the Hamiltonian has C4 rotation symmetry about the $z$-axis, and C2 symmetry about the $x$-axis, 
by Eq.~(\ref{Onsager-Lij}) it is easy to verify that 
\be
\alpha_{xy}(\bB)\underbrace{=}_{\rm C_4}\bar{\alpha}_{xy}(\bB) \quad .
\label{OnsagerC_4}
\ee
and that all symmetric transverse conductivities vanish,
\be
\sigma_{xy,+} =\alpha_{xy,+}=\bar{\alpha}_{xy,+}=\kappa_{xy,+}\underbrace{=}_{\rm C_4} 0 \quad .
\ee

\section{Hall conductivity proxies}
\label{sec:Proxies}
 For general Hamiltonians,  the Kubo formula (\ref{Kubo-Lehmann}) for the DC Hall conductivity $\sigma_{xy}^{\rm dc}$  is computationally challenging. Hence, two simpler proxy formulas for $\sigma_{xy}^{\rm dc}$ have been very popular: (i) The Chern conductivity $\sigma_{xy}^{\rm Chern}$ and (ii)  the Streda conductivity $\sigma_{xy}^{\rm Streda}$. We emphasize that these proxies are only valid  under restricted conditions, since they reverse the DC order of limits prescribed by Eq.~(\ref{DC}).

\subsection{Chern number}
\begin{figure}[h]
\begin{center}
\includegraphics[width=6cm]{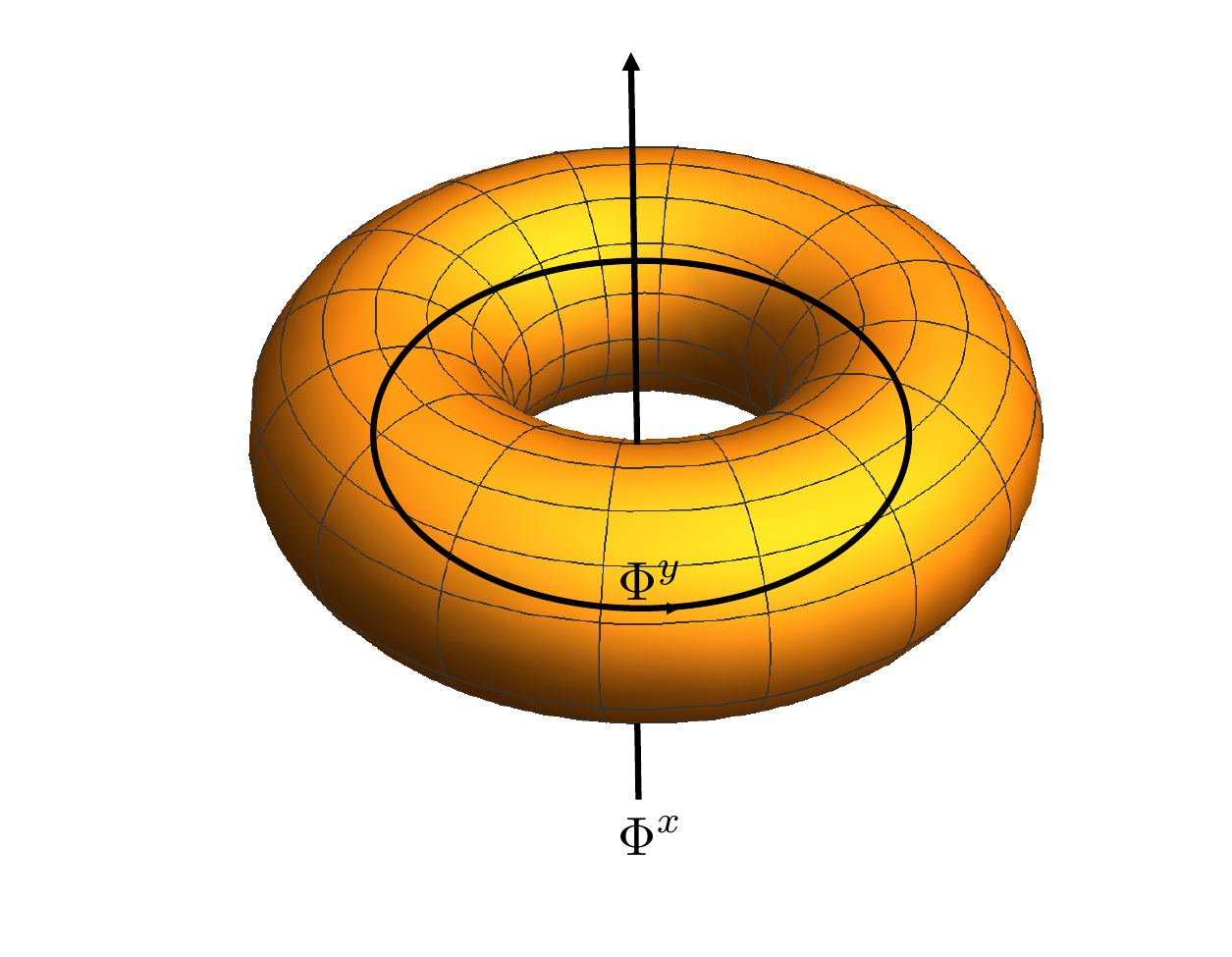}
\caption{The Gauged Torus setup for a Chern number calculation.  The charged system in placed on the surface of the torus, which is is penetrated by a uniform magnetic field.  Two Aharonov-Bohm fluxes $\Phi^x,\Phi^y$ which thread the holes of the torus serve as adiabatic parameters of the ground state wavefunction $\Psi_0$.}
\label{fig:torus}
\end{center}
\end{figure}
One considers a Hamiltonian 
\be
H =  H_0(\bB,\theta^x,\theta^y) \quad ,
\ee
which describes  particles with charge $e$ placed on the surface of a torus of area $\cV=L_x\times L_y$, with  PBC. A  uniform magnetic field $\bB$ penetrates the surface, and two Aharonov-Bohm fluxes $\Phi^\alpha,\alpha=x,y$ are threaded through the $x$ and $y$ holes of the torus, as depicted in Fig.~\ref{fig:torus}. 
The fluxes are parametrized by,
\be
\Phi^\alpha = {\Phi_0 \over 2\pi} \theta^\alpha \quad ,\quad \theta_\alpha\in [0,2\pi] \quad , 
\ee
where $\Phi_0={hc/ e}$ is the flux quantum.
 
The Chern number  of the ground state $|\Psi_0\rangle$ defined by an integral over the Berry curvature with respect to the angles $\theta_x, \theta_y$:
\bea
C_{xy}   &=&  {1\over \pi} \int \limits_{0}^{2\pi} \int \limits_{0}^{2\pi} \!  d\theta_x  d\theta_y ~ \Im  \langle { \partial\over \partial  \theta_x  }\Psi_0  | {\partial\over \partial \theta_y } \Psi_0 \rangle 
=  {1\over  2\pi } \oint \limits_{0}^{2\pi} \! d\btheta  \cdot  \Im  \langle  \Psi_0  |  \bnabla_\btheta \Psi_0 \rangle \nonumber\\
&=&     \mbox{Chern~number~(integer)} \quad .
\label{Chern}
\eea
The Chern number is a topological integer which characterizes a gapped Quantum Hall ground state $|\Psi_0\rangle$.

First order perturbation of the ground state w.r.t. the interaction
 \be
 H'= -{\hbar \over e}\sum_\alpha j^\alpha_0 {\theta^\alpha \over L_\alpha}  
 \label{Der-AB}
 \ee
yields,
 \be
{\partial\over \partial  \theta_\alpha }|\Psi_0\rangle = {\hbar \over e L_\alpha} \sum_{m\ne 0} |\Psi_m\rangle {\langle \Psi_m|j_0^\alpha|\Psi_0\rangle \over E_m-E_0 } \quad .
\ee
Thus the Chern number can be related   to the finite volume, zero frequency $\bq=0$ Kubo formula for $\sigma_{xy}$ as given by the real part of Eq.~(\ref{Kubo-Lehmann}):
\bea
\sigma^{\rm Kubo}_{xy} \Big(\bq=0,i\ve=0,\cV \Big)_{ \theta_x,\theta_y} &=&{\hbar\over \cV} 
\sum_{0m}   {   \langle \Psi_0| j_0^x |\Psi_m\rangle \langle \Psi_m|j_0^y |\Psi_0\rangle - \langle \Psi_0| j_0^y |\Psi_m\rangle \langle \Psi_m|j_0^x|\Psi_0\rangle
 \over (E_{m}-E_0)^2 } \nonumber\\
 &=& { 2 e^2 \over \hbar} \Im  \langle { \partial\over \partial  \theta_x  }\Psi_0  | {\partial\over \partial \theta_y } \Psi_0 \rangle \quad .
 \eea
 Thus, 
\be 
\sigma^{\rm Chern}_{xy} \equiv {1\over (2\pi)^2}
\int \limits_{0}^{2\pi}\! d\theta_x\!\int \limits_{0}^{2\pi} \!  d\theta_y  ~~ \sigma^{\rm Kubo}_{xy} \Big(\bq=0,i\ve=0,\cV \Big)_{ \theta_x,\theta_y}=  {e^2\over h} C_{xy} \quad .
\label{ChernMB}
\ee 
In other words, the flux-averaged Kubo formula yields an integer times $e^2/h$.
The variation of the integrand $\sigma^{\rm Kubo}_{xy}$, with flux parameters $ \theta_x,\theta_y\in[0,2\pi)$ is expected to vanish in the large volume limit. The double integral
$\int \int {d\theta_x d\theta_y\over (2\pi)^2} $  can be replaced by  $\sigma_{xy}(0,\cV)$ at $\theta_x=\theta_y=0$. 
 
Historically, the topological relation between $\sigma^{\rm Chern}$ and the Chern number was initially discovered  by Thouless, Kohmoto, Nightingale, and den Nijs, (TKNN)~\cite{TKNN},  for  filled  bands of non interacting, disorder-free electrons on a torus, penetrated by a unit-cell-commensurate magnetic field. 

The Chern proxy formula (\ref{Chern}) for interacting Hamiltonians was  independently derived by adiabatic transport theory by Niu, Thouless and Wu~\cite{NiuThoulessWu} and by Avron and Seiler~\cite{Avron}. The one plaquette  Chern number formulation by Kudo~\etal ~\cite{kudo} provides a useful simplification of  the computation of $\sigma^{\rm Chern}_{xy}$ in the large volume limit.

However, we must caution  that $\sigma^{\rm Chern}_{xy}$ {\em reverses} the proper DC order of limits prescribed by Eq.~(\ref{DC}).
Therefore the validity of this proxy is limited to the following conditions:
\begin{enumerate}
    \item For the adiabatic theorem to apply~\cite{Avron},  a finite gap $\Delta>0$ between the ground state and all excitations must exist in  the infinite volume limit.
    \item The  longitudinal conductivity must vanish $\sigma_{xx}=0$. This can be ensured only at zero temperature if there are no gapless current carrying excitations and no inelastic scattering processes.
\item If disorder is present it should be weak enough to prevent gapless current-carrying states from percolating through the bulk of the sample.   
\end{enumerate}
That said, we emphasize that  $\sigma^{\rm Chern}_{xy}$  has been instrumental in  mathematical  and experimental  characterization of Quantum Hall and topological insulator phases~\cite{Bernevig,Haldane}.

\subsection{Streda formula}

The Streda formula proxy for $\sigma_{xy}$  is obtained by reversing the DC order of limits of Eq.~\ref{DC}. In Ref.~\cite{EMT}
it is shown that,
\bea
\sigma_{xy}^{\rm Streda}&=&\lim_{\bq \to 0}\lim_{\cV\to \infty}\lim_{i\ve\to 0^+} \sigma_{xy} (\bq, i\ve) \nonumber\\
&=& c \lim_{\bq \to 0}  \lim_{\cV\to \infty}\left(\rho_{\bq} \big| M_{\bq}^z\right) =  c {\partial \rho \over \partial B}\Bigg|_{\mu,T} = c {\partial m^z \over \partial \mu}\Bigg|_{B,T} \quad ,
\label{Streda}
\eea
where $m^z$ is the $z$-magnetization and $\rho$ is the charge density.
The last equation follows from a Maxwell relation. The Streda formula is often used to characterize quantum Hall phases and topological insulators by  scanning compressibility measurements~\cite{Amir}.

In general the Streda order of limits   does not commute with the DC order of limits (\ref{DC}), and therefore  
\be
\sigma_{xy}^{\rm Streda}=\lim_{\bq\to 0} \lim_{\cV\to \infty}\lim_{\omega\to  0} \sigma_{xy} (\bq,\omega)\ne   \lim_{\omega\to  0}\lim_{\bq\to 0} \lim_{\cV\to \infty}\sigma_{xy} (\bq,\omega)=\sigma^{\rm dc}_{xy} \quad .
\label{ool}
\ee

The conditions which may permit  the
use Streda's proxy is that the ground state is bulk-incompressible, and that the Hall angle is large, i.e. $\sigma_{xy}/\sigma_{xx}\gg 1$.  Such conditions occur when there is a gap  $\Delta >0$ for excitations above the ground state, and the temperature is lower than this gap.   
A weaker condition is a pseudogap, where
the density of charge excitations
vanishes rapidly enough at low energies to allow reversal of the $\omega,\bq\to 0$ order of limits in the Kubo formula.

\section{Kubo formulas in operator Hilbert space}
\label{sec:OHS}

The Kubo formulas of Eq.~(\ref{Kubo-Lehmann}) may be formulated as matrix elements in OHS. This formalism avoids the Lehmann representation, and  proves to be  convenient for further mathematical manipulation.  Below, we shall use the OHS formalism to  derive the DPP  formulas in Part \ref{sec:DPP},  the continued fractions in Section \ref{sec:moments}, and the  Hall coefficient summation formulas in Part \ref{Part:Thermo II}.
We start with the formulation of susceptibilities as inner product in OHS.
We then proceed to formulate dynamical response functions in OHS.

\subsection{Equilibrium Susceptibilities}
Given a static Hamiltonian  and free energy,
\be
H=H_0-h_A A - h_B B,\quad F= -T \log \Tr \ e^{-\beta H} \quad .
\ee
A thermodynamic expectation value is given by,
\bea
\langle A\rangle &\equiv &  -\left. {\partial F\over \partial h_A } \right|_{h_a=h_B=0}\nonumber\\
&=& \Tr \left(\rho_0 A\right) \quad ,  
\eea
where Boltzmann weights are $\rho = {e^{-\beta H_0} / \Tr e^{-\beta H_0}}$. 
Static susceptibilities with respect to $A$ and $B$ are,
\bea
\chi_{AB} &=&  -{\partial^2 F[h_A,h_B] \over \partial h_A\partial h_B }\big|_{h_a=h_B=0}\nonumber\\
 &=& \int_0^\beta  d\tau ~ \langle A(\tau) B\rangle   ~+\beta   \langle A^{\rm diag} B^{\rm diag}\rangle  - \beta \langle A\rangle \langle B\rangle  \nonumber\\
&=& \sum_{m,n} { \rho_m  - \rho_n \over E_n- E_m}  \langle n|A|m\rangle \langle m|B|n\rangle ~+\beta   \langle A^{\rm diag} B^{\rm diag}\rangle  - \beta\langle A\rangle \langle B\rangle \quad .
\label{chiAB}
 \eea
 where $O^{\rm diag}=\sum_n |n\rangle\langle n|O|n\rangle \langle n| $.
In the second line $B(\tau)\equiv e^{\tau H_0} B e^{-\beta H_0}$.
The last line of (\ref{chiAB}) is  written in the  Lehmann representation of $H_0$. 
The vector space of linear operators $A\to |A)$ 
which act within the Schroedinger Hilbert space of the Hamiltonian $H_0$, is an OHS containing hyperstates $|A)$, and closed under linear superpositions.
Henceforth we consider general operators $A^{\rm diag}=B^{\rm diag}=0$, whose susceptibility is 
\bea
 \chi_{\rm AB} &=&  \sum_{nm} W_{nm} A^*_{nm} B_{nm} \equiv(A|B) \quad ,\nonumber\\
W_{nm} &\equiv & { \rho_m- \rho_n\over E_n- E_m}  \ge 0 \quad .
\label{BM}
\eea
$(A|B)$ defines the Bogoliubov-Mori (BM)  inner product, \cite{Mori} which depends on $H_0$ (with its boundary conditions) and 
inverse temperature $\beta$.
It obeys the Hilbert space conditions of (i) a positive norm, (ii) hermiticity and (iii) linearity,
\bea
&&({\rm i})~~(A|A)\ge 0 \quad ,\nonumber\\ 
&&({\rm ii})~~(A|B)=(B|A)^* \quad ,\nonumber\\
&& ({\rm iii})~~
(A|aB+bC)=a(A| B)+b(A|C),~\forall a,b\in \mathbb{C} \quad .
\eea
All operators which commute with $H_0$ are identified with the null hyperstate of zero norm.

\subsection{The Liouvillian and its inverse}
In OHS,  the Liouvillian hyperoperator $\cL$,
which generates the time evolution of hyperstates, is defined by,
\be
\cL | A) =  i|[H_0,A]) \quad ,
\ee
where $A$ is any operator in the  Hilbert space of $H_0$.
\begin{lemma} The Liouvillian is a hermitian hyperoperator in OHS,
\end{lemma}
\begin{proof}
\bea
\Big(A\big|\cL \big| B\Big) &=& \sum_{nm} {\rho_m-\rho_n \over E_n - E_m } A^*_{mn}(E_m-E_n) B_{mn}=
- \sum_{nm} (\rho_m-\rho_n )A^*_{nm} B_{mn}\nonumber\\
&=&  \left(\sum_{nm} (\rho_n-\rho_m ){E_m-E_n\over E_m - E_n }  B^*_{mn} A_{mn} \right)^*= \Big(B\big|\cL\big| A\Big)^*\nonumber\\
&&\Rightarrow ~~\cL = \cL^\dagger ~~~\qed
\label{hermiticity}
\eea
where the summation indices are relabelled $mn\to nm$ in the last equation.
\end{proof}
$\cL$ is therefore diagonalizable. Its  eigenoperators are $|n\rangle\langle m|$
with real eigenvalues $E_n-E_m$.
It is the generator of time evolution in the OHS, since
\be
\big|A(t)\Big) = \big|e^{i H_0 t} A e^{-iH_0 t} \Big)= e^{i\cL t} \big|A\Big) \quad .
\ee

The inverse Liouvillian $\left({1\over \cL}\right)_\ve'$  does not exist if the kernel of $\cL$ is non zero. Hence 
we define it with an imaginary infinitesimal $\ve$ prescription,
\be
 \left({1\over \cL-i\ve}\right)_\ve  \equiv 
 \left({1\over \cL}\right)'_\ve + i\left({1\over \cL}\right)''_\ve \quad ,
 \ee
 where the real part of the inverse is the hyperoperator  
 \be
 \left({1\over \cL}\right)'_\ve  \equiv 
 {\cL\over \cL^2 + \ve^2 } \quad ,
 \label{InverseLp}
 \ee
 and the imaginary part is
 \be
  \left({1\over \cL}\right)''_\ve  \equiv 
 {\ve \over \cL^2 + \ve^2 } = \pi \delta_\ve( \cL) \quad .
 \label{InverseLpp}
 \ee

 The weight of the inner product, Eq.~(\ref{BM}) can be expressed as
 \be
W_{nm}=\lim_{\ve\to 0} \sum_{nm} \Re\left( {\rho_m-\rho_n\over E_n-E_m+i\ve}\right) \quad .
\ee
Using Eq.~(\ref{InverseLp}) we can write BM inner product
in a representation independent form,
\bea
\Big(A\big|B\Big) &=& -\lim_{\ve\to 0} \Tr   \left[ \rho,\left({1\over \cL}\right)_\ve A^\dagger\right]  B  \nonumber\\
&=&  -\lim_{\ve\to 0} \Tr ~   \rho \left[ \left({1\over \cL}\right)_\ve A^\dagger, B\right] \quad .
\label{NoLehmann}
\eea
The advantage of Eq.~(\ref{NoLehmann}) over Eq.~(\ref{BM}) will be made clear during the mathematical manipulations of the Kubo formulas performed in the following Sections.

\subsection{Dynamical linear response functions}
Dynamical response functions are obtained in linear response by adding to the Hamiltonian $H_0$
a weak time-dependent field $h_B(t)$ coupled to an operator $B$. The field is turned at $t\ge 0$,
\be
H(t) = H_0 - \Theta(t)~h_B(t) ~ B \quad .
\ee
At $t>0$ the expectation value of an observable $A^\dagger$ is
\be
\langle A^\dagger \rangle(t)   =   \Tr \rho(\beta) U^\dagger (t) A^\dagger U([h_B],t] = \int \limits_{0}^{t} dt'  R_{AB}(t,t') h_B(t') + \cO(h_B^2) \quad ,
\label{A-av}
\ee
where the density matrix $\rho={1\over Z} e^{-\beta H_0}$ and the evolution operator $U$
satisfies $i \dot{U} =  H(t)U$. 
Expanding $U$ to linear order in $h_B$ yields, 
\be
R_{AB}(t-t') =   {-i}  \Tr \rho  \left[ A^\dagger(t-t'), B  \right]\Theta(t-t') \quad ,
\label{RABt}
\ee
where $A(t)=e^{{i}  H_0 t} A e^{-{i} H_0 t}$, and using $[\rho,H_0]=0$ to obtain  $R_{AB}(t,t') =  R_{AB}(t-t')$. 
The transform of (\ref{RABt}) into the upper complex plane defines the complex response function,
\bea
R_{AB}(z) &=&  \int \limits_{0}^{\infty} e^{ -izt} R_{AB}(t) =   \sum_{m n} { \rho_m - \rho_n \over E_n-E_m-z} 
\langle n|A^\dagger|m\rangle \langle m|B|n\rangle\nonumber\\
&=&   \left( A \left| {\cL\over \cL-z}\right| B\right) \quad ,\quad \Im(z)>0 \quad .
\label{RAB}
\eea

\subsection{Electric and Thermal conductivities}
The Kubo formulas of Eq.~(\ref{Kubo-Lehmann}) 
are matrix elements in OHS,
\be
L^{\alpha\beta}_{ij}(\bq,\omega+i\ve) = {1\over \cV} R_{J^\alpha_i,P^\beta_j }(\omega+i\ve)=
-{i\over \cV}\left( (J^\alpha_i)_{\bq} \left| {1\over \cL-\omega-i\ve}\right| (J^\beta_j)_{\bq}\right) \quad .
\label{Kubo-OHS}
\ee
The values of $\omega,i\ve, \bq, \cV$
are all finite. The DC limit must to be taken carefully as required by Eq.~(\ref{DC}).

\newpage

\part{DPP Hall conductivities}
\label{sec:DPP}

\section{Derivation of DPP formulas}
For simplicity the magnetic field is chosen along the $z$ axis and  $C_4$ symmetry is assumed in the $XY$ plane. According to Eq.~(\ref{DC}), the DC Hall-type conductivities are given by,
\be
\Re~(L^{xy}_{ij})^{\rm dc} =  \lim_{\varepsilon\to 0} \lim_{\bq\to 0}\lim_{\cV\to \infty} \left( \Re~L_{ij}^{xy}(\ve,\cV)- \Re~L_{ij}^{xy}(0,\cV)\right) \quad .
\label{Hall-cond}
\ee
As mentioned before,  the second term cancels the spurious magnetization currents which are created by the static component of the ``gravitational'' field $\bpsi(\omega=0)$. 
In  OHS notation,
\be   
\Re~L_{ij}^{xy}(\ve,\cV)  \equiv  {1\over \cV} \Im \left( J_i^x \left| \left({ 1\over \cL } \right)'_\ve \right| J_j^y\right)+{1\over \cV} \cancelto{0~(C_4)}{\Re \left( J_i^x \left| \left({ 1\over \cL } \right)''_\ve \right| J_j^y\right)} \quad .
\ee
where $\bq$-dependence of the currents is implicit.
Under $\pi/2$ rotation in the plane $\Re~L_{ij}^{xy}(\ve,\cV)=
-\Re~L_{ij}^{yx}(\ve,\cV)$, and since the second term is symmetric under $x\leftrightarrow y$, it vanishes.
In the Lehmann representation,
\be
\Re L_{ij}^{xy}(\ve,\cV)={1\over \cV}   \sum_{mn} 
{\rho_m-\rho_n  \over E_n-E_m }  \Im \left(  \langle n|J_i^x|m\rangle \langle m|J_j^y |n\rangle \right) ~  \left({  E_m- E_n\over (E_m- E_n)^2 +\ve^2 }  \right) \quad .
\label{Sxy-Lehmann}
\ee
At first glance, 
one might be tempted to discard $\cO(\ve^2)$ contributions in the denominator of Eq.~(\ref{Sxy-Lehmann}). On OBC, this would {\em be a gross error}! As shown below, the DC limit is dominated by eigenstates with $|E_n-E_m|\le \ve$.    

Using the identity (\ref{NoLehmann}),
\be
\Re L_{ij}^{xy}(\ve,\cV)=- {1\over \cV} \Im \Tr \rho \left[ \left({1\over \cL}\right)_{\ve'}' J_i^x, \left({1\over \cL}\right)_\ve' J_j^y \right] \quad .
\label{Sxy-eps}
\ee 
To remain consistent with $C_4$ symmetry, on finite volumes  we  identify $\ve'=\ve$, while keeping the DC order of limits in Eq.~(\ref{Hall-cond}). By Eq.~(\ref{JP}),
\be 
\left({1\over \cL}\right)_\varepsilon' J_j^\alpha = i \left( 1 - \Theta_\ve \right) P_j^\alpha  \equiv i \left( P_j^\alpha - (\tilde{P}_j^\alpha)_\ve \right) \quad ,\quad~\alpha=x,y \quad .
\label{Tve}
\ee
$\Theta_\ve$ is a Lorentzian  degeneracy projector,
\be
\Theta_\ve= {\ve^2\over \cL^2 + \ve^2} 
\label{HP}
\ee
 and  $(\tilde{P}_j^\alpha)_\ve\equiv  \Theta_\ve P_j^\alpha$ are
DPP's ({\em degeneracy projected polarizations}),  
whose matrix elements in the Lehmann representation are restricted (at small $\ve$) to connect quasi-degenerate eigenstates,
\be
\langle n| (\tilde{P}_i^\alpha)_\ve |m\rangle =\langle n|  P_i^\alpha  |m\rangle ~~{\ve^2\over |E_n-E_m|^2 + \ve^2} \quad .
\ee

Eq.~(\ref{Sxy-eps}) can be written as 
\be
\Re L_{ij}^{xy} (\ve,V) = {1\over \cV} \Im ~\Tr ~\rho ~\  \left[P_i^x-(\tilde{P}_i^x)_{\ve},  P_j^y-(\tilde{P}_j^y)_\ve \right] \quad .
\label{DPP-full1}
\ee
Expansion of the terms in (\ref{DPP-full1}) we obtain a sum of four terms,
\be
\Re L_{ij}^{xy} (\ve,V) =  \Re L_{ij}^{\rm xy}(\varepsilon=0,\cV) - L_b( \ve)-   L_c( \ve)+L_d( \ve) \quad .
\label{DPP-full}
\ee

The  first term, which is independent of $\ve$ and $\ve'$ is called a magnetization term,
\be
\Re L_{ij}^{\rm xy}(\varepsilon=0,\cV)\equiv {1\over \cV}~ \Im ~\Tr ~\rho   \left[P_i^x ,  P_j^y \right] \quad .
\label{Sa}
\ee
The physical content of  magnetization terms is discussed in Section \ref{sec:MagCT}. The magnetization term precisely cancels the second term  in Eq.~(\ref{Hall-cond}), and therefore does not need to be calculated for the Hall-type conductivity. 

The remaining three terms are related to each other,
\bea
L_b(\ve) &=&  {1\over \cV} \Im ~\Tr ~\rho  ~\left[ P_i^x, (\tilde{P}_j^y)_\ve\right] = {1\over \cV} \Im \sum_{n} \rho_n \sum_m \langle n| P_i^x|m\rangle  \langle m|  P_j^y|n\rangle  \left({\ve^2\over (E_n-E_m)^2 + \ve^2}\right) \quad ,\nonumber\\
L_c(\ve ) &=&  {1\over \cV} \Im ~\Tr ~\rho \ ~\left[ (\tilde{P}_i^x)_{\ve },  P_j^y \right] =  {1\over \cV} \Im \sum_{n} \rho_n \sum_m \langle n| P_i^x|m\rangle  \langle m|  P_j^y|n\rangle  \left({(\ve )^2\over (E_n-E_m)^2 + (\ve')^2}\right) \quad ,\nonumber\\
L_d(\ve) &=&  {1\over \cV} \Im ~\Tr ~ \rho ~\left[ (\tilde{P}_i^x)_{\ve'}, (\tilde{P}_j^y)_\ve\right] ={1\over \cV} \Im \sum_{n} \rho_n \sum_m \langle n| P_i^x|m\rangle  \langle| m P_j^y|n\rangle  \left({\ve^2\over (E_n-E_m)^2 + \ve^2}\right)^2 \quad .
\label{Sve}
\eea
$L_b$ and $L_c$ contain one DPP, and therefore  a Lorentzian factor, $\ve^2 /((E_n-E_m)^2 + \ve^2)$.
$L_d$ contains a Lorentzian square $(\ve^2 /((E_n-E_m)^2 + \ve^2))^2$.
The Lorentzian factor and its square can be effectively replaced for small $\ve$ by box projectors $\Theta(\pi \ve/2-|E_n-E_m|)$ and $\Theta(\pi \ve/4-|E_n-E_m|)$, respectively.
Since we are taking $\ve\to 0$, (after $\cV\to \infty$), 
the three terms have the same DC limit,
\be
\lim_{\ve\to 0} L_b(\ve)=\lim_{\ve\to 0} L_c(\ve)=\lim_{\ve\to 0} L_d(\ve) \quad .
\ee
 Thus, summing up Eq.~(\ref{DPP-full}) the DC Hall conductivities are given by the DPP  formulas,
\bea
\sigma_{xy}^{\rm dc} & =& - \lim_{\varepsilon\to 0}   \lim_{\bq\to 0}\lim_{\cV\to \infty} {1\over \hbar \cV}~\Im~
\Tr \rho \left[ (\tilde{P^x})_\ve, (\tilde{P^y})_\ve\right]  \quad ,\nonumber\\
\alpha_{xy}^{\rm dc} &= &- \lim_{\varepsilon\to 0}  \lim_{\bq\to 0}\lim_{\cV\to \infty}  { 1\over \hbar \cV k_B T}~ \Im
~\Tr \rho \left[ (\tilde{P^x})_\ve,  (\tilde{Q^y})_\ve\right]  \quad ,\nonumber\\
\kappa_{xy}^{\rm dc} &=& - \lim_{\varepsilon\to 0}    \lim_{\bq\to 0}\lim_{\cV\to \infty} {1\over\hbar \cV k_B T } ~\Im
~\Tr \rho \left[ (\tilde{Q^x})_\ve,  (\tilde{Q^y})_\ve\right] \quad .
\label{Hall-DPP}
\eea
Here we have restored  the dimensionful dependence on $\hbar$ and $k_B$ for practical applications.

\subsection{The Magnetization terms}
\label{sec:MagCT}
The  terms  $L_{ij}^{\rm xy}(\varepsilon=0,\cV)$ in Eq.~(\ref{Sa}) cancel against their counterterms in Eq.~(\ref{Hall-cond}). For  electric Hall conductivity $\sigma_{xy}$, the magnetization term vanishes since the two polarizations commute: $[P^x,P^y]=0$.
The  thermoelectric and thermal Hall magnetization terms do not vanish. 
For Schr\"odinger particles governed by $H^{\rm particles}$, the polarization commutators can be readily calculated to yield,
\bea
\lim_{\bq\to 0}\lim_{\cV\to \infty} L_{12}^{xy}&=&\lim_{\bq\to 0}\lim_{\cV\to \infty} 
{1\over \cV}\Im \langle {[P^x,Q^y]+  [Q^x,P^y] \over 2}\rangle  = - \lim_{\cV\to \infty}{1\over \cV} \langle\sum_{i=1}^{N_p}  \bx_i \times  \bj(\bx_i)  \cdot \hat{\bz}\rangle =  - c  \langle m^z\rangle \quad , \nonumber\\
\lim_{\bq\to 0}\lim_{\cV\to \infty}L_{22} &=& \lim_{\bq\to 0}\lim_{\cV\to \infty}{1\over \cV}\Im \langle [Q^x,Q^y] \rangle  =  -\lim_{\cV\to \infty}{1\over \cV} \langle\sum_{i=1}^{N_p}  \br_i \times  \bj_Q (\bx_i)\cdot \hat{\bz}\rangle= -2  \langle m_Q^z \rangle \quad ,
\eea   
which agrees with Ref.~\cite{Cooper,QNS}.
$m^z$ and $m_Q^z$ are the $z$-direction  electric and  thermal magnetization densities respectively.

The magnetization subtractions create  considerable headache when computing the Kubo formulas from Eq.~(\ref{Kubo-Lehmann}). Since $L_{ij}$ are weighted by ${1\over T}$ in Eqs.~(\ref{Lsigma}), any separate  approximations of the two terms in Eq.~(\ref{Hall-cond})  could result in an error which embarassingly diverges at low temperature. (Heat conductivities should actually vanish at low temperatures by the third law of thermodynamics).  Elimination of these  subtractions from the DPP formulas (\ref{Hall-DPP}) provides an essential simplification.

\section{DPP formulas for non-interacting Hamiltonians}
For a normal Hamiltonian of non-interacting fermions or bosons, Eq.~(\ref{Hsp}), the DPP formula is
\bea
L_{ij}^{xy} &=& - \lim_{\ve\to 0} \lim_{\bq\to 0}\lim_{\cV\to \infty} {1\over \cV} \Im \sum_{l} n(\epsilon_l) \langle l| \left[\tilde{P}_i^x,\tilde{P}_j^y \right]|l\rangle\nonumber\\
&=& - \lim_{\ve\to 0} \lim_{\bq\to 0}\lim_{\cV\to \infty} \sum_{ll'} \Big(n(\epsilon_l)-n(\epsilon_{l'})\Big)
~ \Im  \langle l| \tilde{P}_i^x |l' \rangle \langle l'| \tilde{P}_j^y  | l \rangle \quad .
\label{DPP-sp}
\eea
where the polarizations  $P_i^x,P_j^y$ are defined in (\ref{Pol-l}), and their DPP's are
given by
\be
\langle l|\tilde{P}_i^\alpha |l'\rangle= 
\langle l|{P}_i^\alpha |l'\rangle~{\ve^2 \over \ve^2 + (\epsilon_l-\epsilon_{l'})^2} \quad .
\ee

At low temperatures, the  factor $n(\epsilon_l)-n(\epsilon_{l'})$ ensures the sum is dominated by excitation energies at less than of order $T$ from the chemical potential.

For a bosonic Hamiltonian of coupled harmonic oscillators (\ref{CHO}), the thermal polarization is generally  an anomalous bilinear form given by Eq.~(\ref{Q-CHO}).
The anomalous matrix blocks $Q^\alpha_{A;\bq}, (Q^\alpha_{A;\bq})^* $ involve creation or annihilation
of two positive energy states, such that $\epsilon_l+\epsilon_{l'}>0$. Thus,
there are no anomalous contributions to  the DPP~\footnote{Unlike their conttribution to the magnetization terms. It is good to know
that anomalous terms are projected out, since their commutator would lead to diverging contributions to the low temperature limit of $\kappa_{xy}$.}. Thus we are left  with the normal parts of the polarizations,
\be
\langle l|\tilde{Q}_i^\alpha |l'\rangle=    \sum_{ll'}  \left( { \ve^2\over \ve^2  + (\vve^{\rm diag}_l-\vve^{\rm diag}_{l'})^2 }\right)~  (Q^\alpha_{N;\bq})_{ll'}  \quad .
\label{tQ-CHO}
\ee

\section{Berry curvature integrals}
\label{sec:Berry}
Hall-type conductivities for non-interacting  periodic crystal Hamiltonians $H^{\rm PC}$, defined in Eq.~(\ref{Hperiodic}),
can be expressed as Berry curvature
integrals over BZ, where the band Berry curvature is defined  in Eq.~(\ref{Berry-C}). 

The relations between Berry curvature integrals and Hall conductivity~\cite{Niu-Sxy}, Transverse thermoelectric conductivity~\cite{Niu-Axy}  and thermal Hall conductivity~\cite{QNS,Kxy-Murakami}, have been previously derived from semiclassical dynamics and also directly from the Kubo formula of perfectly periodic lattices~(\ref{KuboAC-LehmannSP1}).

Here we provide a somewhat simpler derivation  (in our view) starting from the single-particle DPP formulas for periodic lattices.
$U_{li}(\bk)$ is a unitary matrix which diagonalizes $h_{ij}(\bk)$ in $H^{\rm PC}$, and  defines the eigenmode   operators, 
\be
b^\dagger_{\bk l}= \sum_{i=1}^M U_{l i}(\bk)a^\dagger_{\bk i}~\quad l=1,\ldots M \quad ,
\label{eigen}
\ee
such that,
\be 
H^{\rm PC} =\sum_{\bk\in \BZ}\sum_{l=1}^M \left(\epsilon_{\bk l}-\mu \right)~  b^\dagger_{\bk l} b^{}_{\bk l} ~ \quad .
\label{HBerry}
\ee 
The  uniform  polarizations can be defined by sending $\bq\to 0$ after  $\cV\to \infty$ in Eq.~(\ref{Pol-periodic}), (see Eq.~(\ref{Pol0-PBC}): 
\bea
 {P}^\alpha &=& e     \sum_{\bk ij} a^\dagger_{\bk i} ~(i\bnabla_{\bk}^\alpha) a^{}_{\bk j} \quad ,\nonumber\\
 {Q}^\alpha   &= & {1\over 2} \sum_{\bk ij} a^\dagger_{\bk i}\left\{ \left(h_{ij}(\bk)-\mu\delta_{ij} \right), (i\bnabla_{\bk}^\alpha) \right\}a^{}_{\bk j}~ \quad .
\label{Pol-Berry}
\eea
The DPP operators are constrained to act only within the same band $l$. Using Eq.~(\ref{eigen}), the electric DPP's are given by
\be
\langle \bk l| \tilde{P}_\ve^\alpha |\bk' l\rangle ~=~  e~ {\ve^2 \over  \ve^2+( \epsilon_{\bk l}-\epsilon_{\bk', l})^2}  ~  \left\{
\overbrace{\sum_i U^\dagger_{li}(\bk)  \left(i\bnabla^\alpha_{\bk }  U_{i l} (\bk)\right)}^{A^\alpha_l(\bk )}+  i\bnabla^\alpha_{\bk'}\right\} \quad .
\label{Pol-Berry1}
\ee
The vector function $\bA_l(\bk)$ is the  {\em Berry gauge field}~\cite{Sundaram-Niu} of eigenmode $|\bk,l\rangle$,
\be
\bA_l(\bk)  \equiv  i \sum_{j=1}^M U_{lj}(\bk)    \bnabla^\alpha_{\bk} U^\dagger_{jl}
(\bk)  ~ \quad .
\ee
Similarly, the thermal DPP's are
\be
\langle \bk l| \tilde{Q}_\ve^\alpha |\bk' l \rangle ~=~   {\ve^2 \over  \ve^2+( \epsilon_{\bk l}-\epsilon_{\bk', l})^2}   ~ 
\left( (\epsilon_{\bk l}-\mu)\left(  i\bnabla_{\bk'}^\alpha + A^\alpha_l(\bk)\right) +{i\over 2} v^\alpha_{\bk l} \right) \quad ,
\label{Pol-Berry2}
\ee
where $\bv_{\bk l}= \bnabla_\bk \epsilon_{\bk l}$. Note that on a finite volume, $\bk$ are discrete with intervals of $\delta\bk = 2\pi/\cV^{1\over d} $. The DC limit is obtained by keeping  $\ve > \delta\bk$, as $\cV\to \infty$. The differentiability  of $\epsilon_{\bk l}$ and $U_{il}(\bk)$ allows us to set $\bk'\to \bk$ in Eqs.~(\ref{Pol-Berry1}) and (\ref{Pol-Berry2}), which involve  a relative correction of $\cO(\ve)$. 

By (\ref{Pol-Berry1}) the commutator  
\bea 
\langle \bk ,l|\left[\tilde{P}_\ve^x,\tilde{P}_\ve^y\right] |\bk ,l\rangle &=& i \bnabla_\bk\times \bA_{ l}(\bk)\cdot \hat{\bz} +\cO(\ve^2)\nonumber\\
&=& i ~ \Omega^z (\bk ,l) \quad ,
\label{Berry-commutator1}
\eea
where we discarded   
$[\partial_{k_x},\partial_{k_y}]=0 $ when acting on differentiable wavefunctions. $ \Omega^z(\bk, l)$ defines the Berry curvature field in the $\hat{\bz}$ direction of eigenmode $l$.

By Eq.~(\ref{Hall-DPP}) the intrinsic Hall conductivity recovers the result of Chang and Niu~\cite{Niu-Sxy},
\be
\sigma_{xy}^{\rm dc} = -  \lim_{\ve\to 0}\sum_l \int {d^d k\over (2\pi)^d }~n(\epsilon_{\bk l} )~\Omega^z_{\bk ,l} \quad .
\ee
This expression has been used to describe the anomalous Hall effect in ferromagnets~\cite{AHE-Berry}. Additional effects of impurities have been introduced semiclassically~\cite{Sinitsyn}.

The thermal Hall DPP's have  a slightly more complicated commutator,
\bea
  \langle \bk l | \left[\tilde{Q}_\ve^x , \tilde{Q}_\ve^{y}\right]|\bk l\rangle &=& i
\langle \bk l | \left( \epsilon_{\bk l}^2 \bnabla_\bk\times \bA_l(\bk) +  \epsilon_{\bk l} \bv_{\bk l}\times \bA_{\bk l} + \epsilon_{\bk l} \bv_{\bk l}\times \bnabla\right)| \bk l \rangle +{i\over 2} \epsilon_{\bk l} \nabla\times \bv_{\bk l} +\cO(\ve) \nonumber\\
&=&   i\epsilon_{\bk l}^2 \bnabla_\bk\times \bA + 2 \epsilon_{\bk l} \bv_{\bk l}\times \bA_{\bk l} +\cO(\ve) \quad .
\label{Berry-commutator2}
\eea
where $i\nabla_\bk \times \bv_{\bk l}=0$.  

For convenience we define energy resolved conductivities by,
 \bea
\tilde{\sigma}_{xy}(\epsilon) &=& - {1\over \cV} \sum_{\bk l}   \delta(\epsilon-\epsilon_{\bk l}  )  ~ \Omega^z_{\bk l} \quad , \nonumber\\
\tilde{\Sigma}_{xy}(\epsilon) &=& - {1\over \cV}\sum_{\bk l} \delta(\epsilon-\epsilon_{\bk l}  )  (\nabla_{\bk}\epsilon_{\bk l})\times \bA_{\bk l} \cdot \hat{\bz} \quad .
\eea
In three dimensions, for  fixed $k_z$,  the  $\delta$ function restricts the wavevectors to a circle on sphere of energy $\epsilon$, such that
\be
\sum_{\bk_\perp l} \delta(\epsilon-\epsilon_{\bk l}  ) ={1\over  (2\pi)^2} \oint \limits_{\bs\in \partial  \{\epsilon_{k_z,\bk_\perp,l}\le \epsilon\}} {d{\bf s} \over |\nabla_{\bk_\perp}  \epsilon_{\bk l}| } \quad .
\ee
Thus, we use  Stokes theorem to relate $\Sigma_{xy}$  to $\sigma_{xy}$ by,
\bea
\tilde{\Sigma}_{xy}(\epsilon)&=& {1\over  (2\pi)^3} \sum_l \int  dk_z \oint \limits_{\partial  \{\bs\in\epsilon_{k_z,\bk_\perp}\le \epsilon \}} d{\bf s} \cdot  \bA_l^z\nonumber\\
 &=&  \sum_l\int   {d^3 k\over (2\pi)^3}    \Theta( \epsilon-\epsilon_{\bk i}) ~\Omega^z_{l} (\bk)\nonumber\\
&=&   \int \limits_{-\infty}^{\epsilon}  d\epsilon' ~\tilde{\sigma}_{xy} (\epsilon') \quad .
\eea
Thus, we arrive at the expression previously derived by Qin, Niu and Shi~\cite{QNS},
\bea
\kappa^{\rm dc}_{xy} &=& \lim_{\ve\to 0}\lim_{\bq\to 0}\lim_{\cV\to \infty}{1\over T}    \int \limits_{-\infty}^{\infty} \! d \epsilon~ n(\epsilon) \left( \epsilon^2  \tilde{\sigma}_{xy} + 2\epsilon  \int_{0}^\epsilon  d\epsilon' \tilde{\sigma}_{xy} (\epsilon')  \right)\nonumber\\
&=&  {1\over T}    \int_{-\infty}^\infty \! d \epsilon \left(-{ \partial  n\over \partial \epsilon } \right)  \epsilon^2  \tilde{\Sigma}_{xy}(\epsilon) \quad .
\label{Kxy-Berry}
\eea

Similarly, the DPP formula for the thermoelectric conductivity of a clean $C_4$-symmetric metallic band of electrons with charge $e<0$, is given by
\bea
\alpha_{xy}^{\rm dc}&=& -\lim_{\ve\to 0}\lim_{\bq\to 0}\lim_{\cV\to \infty} {1\over 2\cV T}
\sum_{\bk l} n(\epsilon_{\bk l})\left(  \left[ \tilde{P}^x_{\bk l},\tilde{Q}^y_{\bk l}\right] +\left[ \tilde{Q}^x_{\bk l},\tilde{P}^y_{\bk l}\right]\right) \quad .
\eea
Using (\ref{Berry-commutator1}) and defining, 
\bea
\tilde{\sigma}_{xy}(\epsilon) &=& - {1\over \cV} \sum_{\bk l}   \delta(\epsilon-\epsilon_{\bk i}  ) \Omega^z_{\bk l} \nonumber\\
\tilde{\Sigma}_{xy}(\epsilon)&=& 
\int \limits_{-\infty}^{\epsilon}  d\epsilon' \tilde{\sigma}_{xy} (\epsilon') \quad ,
\eea
one can express, 
\bea 
\alpha_{xy}^{\rm dc} 
&=&  {1\over e T}    \int \limits_{-\infty}^{\infty} \! d\epsilon \left(-{ \partial  n\over \partial \epsilon } \right)  \epsilon  \tilde{\Sigma}_{xy}(\epsilon) \quad .
\label{Axy-Berry}
\eea

For low temperatures Eqs.~(\ref{Kxy-Berry}) and (\ref{Axy-Berry}) for metals reduce to the relations,
\bea
\kappa_{xy}^{\rm dc}&\simeq& {\pi^2 k_B^2 T \over 3 e}   \tilde{\sigma}_{xy}(\epsilon_F ) \quad ,\nonumber\\
\alpha_{xy}^{\rm dc} &\simeq &{\pi^2 k_B^2 T \over 3 e}  {d\tilde{\sigma}_{xy}\over d \epsilon_F } ~ \quad .
\eea
these relations extend Mott relations~\cite{Marder} between electric and thermoelectric conductivities to bandstructures with Berry curvatures~\cite{Niu-Axy}.

Introducing the effects of disorder in the anomalous Hall effect has been a major challenge in the field. The main difficulty is to extend the semiclassical analysis to include effects of short range scattering~\cite{Sinitsyn,Skew1,Skew2,Skew3}.
The general non-interacting DPP formulas~(\ref{DPP-sp}) can feasibly study effects of disorder numerically.

\section{DPP formula for confined Landau levels}
We consider an electron in a uniform magnetic field $\bB$, with vector potential $\bA$, $\bnabla\times\bA=B\hat{\bz}$
In first quantized notation, the  Landau operators
are,
\be
\Pi^\alpha\equiv  p^\alpha-{e\over c} A^\alpha \quad .
\ee
The Landau level (LL) raising and lowering operators
are respectively,
\be
a={l_B^2 \over \hbar }\left(\Pi^x + i \Pi^y \right) \quad ,\quad a^\dagger={l_B^2 \over \hbar }\left(\Pi^x - i \Pi^y \right) \quad .
\label{adag}
\ee
where $[a,a^\dagger]=1$, and the Landau length is $l_B=\sqrt{\hbar c\over eB} $.

The LL Hamiltonian is 
\be
H^{\rm LL}= {1\over 2m} \left((\Pi^x)^2 + (\Pi^y)^2 \right)=\hbar\omega_c \left(a^\dagger a^{}+{1\over 2} \right) \quad ,
\ee
where the cyclotron frequency is $\omega_c = {eB\over mc} $.
The eigenenergies are
\be
\epsilon_{k,\nu}= \hbar\omega_c \left(\nu+{1\over 2}\right) \quad ,~~~k=1,\ldots N_{\rm L}\quad .
\ee
where the integer LL index is $\nu$, and $N_{\rm L}=  {BA\over \Phi_0} $  is the number of  states per area $A$ in each LL.

The guiding center operators are,
\bea
&& R^x =x+{l_B^2\over \hbar}~\Pi^y~ \nonumber\\
&& R^y =y-  {l_B^2\over \hbar} \Pi^x \quad ,
\label{GC}
\eea
where
\be
[R^\alpha,\Pi^\beta]=0,~\quad  [R^x,R^y]=-il_B^2 \quad .
\ee

The polarization operators are
\bea
P^x &=& e x = eR^x -e{l_B^2\over \hbar}~\Pi^y \quad ,\nonumber\\
P^y &=& e y = eR^y +e{l_B^2\over \hbar}~\Pi^x \quad .
\eea
Since $\Pi^x, \Pi^y $ change the LL index
by $\pm 1$, while $R^\alpha$ connect states within each LL.
Hence under degeneracy projection, the $\Pi$-operators are projected out, and  DPP's are simply proportional to the guiding center operators, 
\bea
&& \tilde{P}^x=e R^x \quad ,\nonumber\\
&& \tilde{P}^y=e R^y \quad .
\eea
The DPP Hall conductivity formula ~(\ref{DPP-sp}) yields 
\bea
\sigma_{xy}^{\rm dc} &=& - {e^2\over \hbar \cV} \sum_{k,\nu} n(\epsilon_{k,\nu })~\Im\langle k,\nu|\left[R^x,R^y\right]|k,\nu\rangle\nonumber\\
&=& {ec\over B\cV}\sum_{k,\nu} n(\epsilon_{\nu,k})= {n_e  e c\over B} \quad ,
\label{Sxy-Galil}
\eea
where $n_e(\mu,T,B)$ is the electron density.

While the sum in Eq.~(\ref{Sxy-Galil}) includes all occupied states,  Eq.~(\ref{DPP-sp}) shows that  the Hall conductivity can be written only in terms of states near the Fermi energy $\epsilon_F$. 

Where are these current carrying states?
On finite systems with OBC, the semiclassical eigenenergies in each LL  (which apply for smooth confining potential on the scale of $l_B$),  reach the Fermi energy $\epsilon_F$ at the sample edges. Each  LL band contributes one gapless edge mode. The Hall conductivity is precisely given by Eq.~(\ref{Sxy-Galil}). 

\section{DPP Hall conductivity in a disordered metal} 
In the weak disorder and magnetic field regime proxies such as Chern numbers and Streda formulas do not work, since disorder mixes Landau levels, and the longitudinal conductivity is finite. Since $\sigma_ {xy}$ strongly depends on the scattering lifetime, it is instructive to compare Boltzmann's theory  to the Kubo formula by numerically calculating it via the DPP formalism. It also gives us a chance to demonstrate how the DC limit is taken in a metallic gapless system.

In Ref.~\cite{Bashan},  the  weakly disordered square lattice (DSL) Hamiltonian was considered,
\be
H^{\rm DSL}=-\sum_{\langle ij \rangle}  \left( e^{-iA_{ij}} c^\dagger_i c^\nd_j + {\rm h.c.} \right)+\sum_i ( w_i -\epsilon_F) c^\dagger_i c_i \quad ,
\label{DSL}
\ee
where $w_i\in[-w/2,w/2]$ is a uniformly distributed random number and $B= \sum_{\square} A_{ij}$ is the magnetic fux per plaquette.
Using the DPP formula (\ref{Hall-DPP}), we compute
the disordered averaged curves $\sigma_{xy} (\ve,L) $ 
 for a sequence of linear sizes $L$, as  shown in Fig.~\ref{fig:scaling}. 
 \begin{figure}[t!]
\begin{center}
\includegraphics[width=10cm]{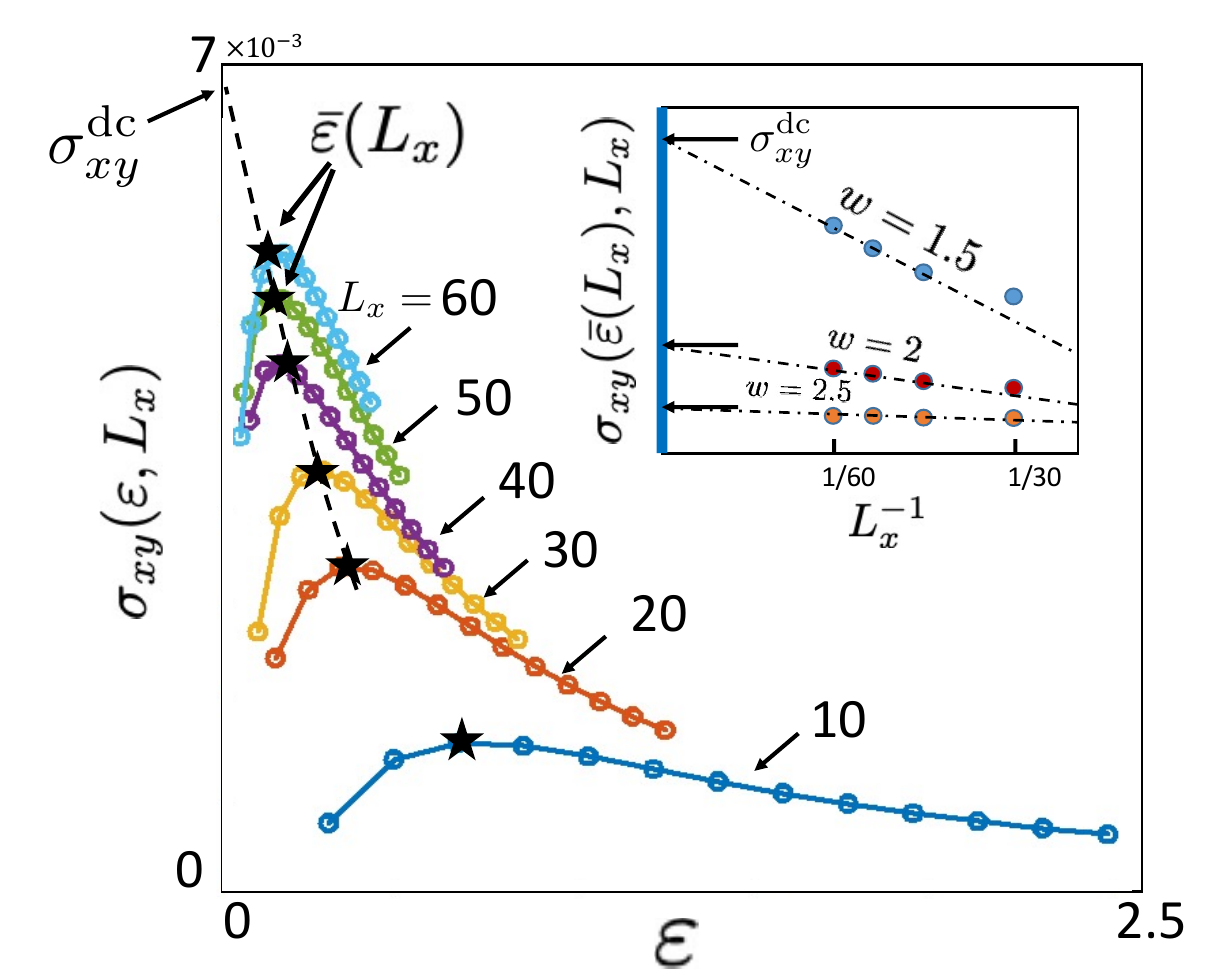}
\caption{Extrapolation of numerical Hall conductivities of the square lattice Hamiltonian, Eq.~(\ref{SL}).
Disorder averaged $\sigma_{xy}$ are plotted versus $\ve$, for a sequence of
linear dimensions $L_x$.  Stars mark the values of $\bar{\ve}$ as defined in Eq.~(\ref{scaling}).  The disorder strength is fixed at $w=3$. 
The temperature, Fermi energy, and magnetic field are $T=0.3,\epsilon_F=-1$ and $B=0.025$
respectively.
{\bf Inset:} The DC limit  $\sigma_{xy}^{\rm dc}$ (marked by black arrows)
for three values of disorder strength. }
\label{fig:scaling}
\end{center}
\end{figure}
 \begin{figure}[b!]
\begin{center}
\includegraphics[width=7cm]{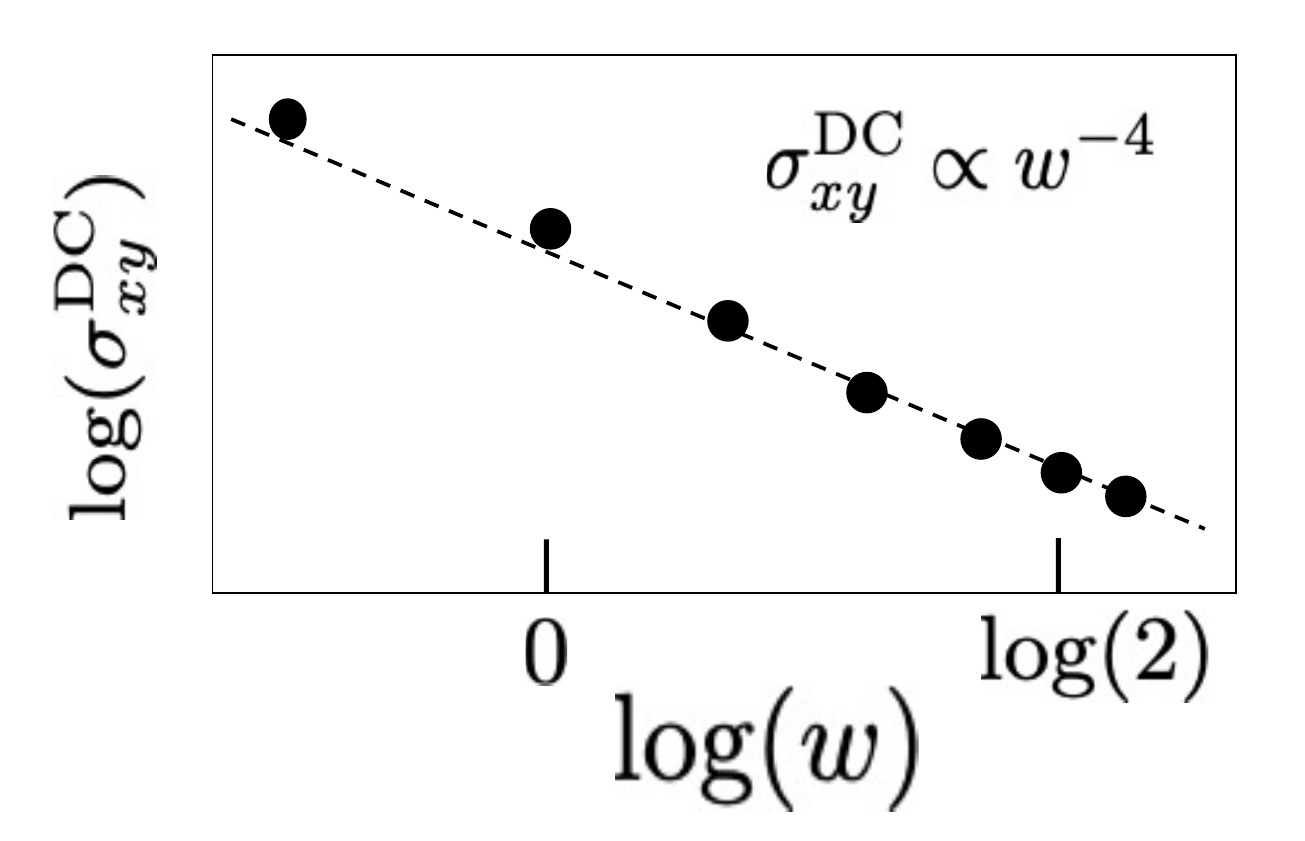}
\caption{ DC extrapolated Hall conductivity for weak magnetic field $B=0.03$
and moderate disorder strength $w$, at Fermi energy $\epsilon_F=-1$. The Hall conductivity scales as $w^{-4}$ as expected by Drude's theory $\sigma_{xy}\propto \omega_c \tau^2$.}
\label{fig:SXY-w}
\end{center}
\end{figure}

 For each curve we  determine  $\bar{\ve}(L)$ by 
 \be
 {d\sigma_{\rm xy}(\ve,L) \over d\ve }\big|_{\ve=\bar{\ve}} =0 \quad ,
 \label{scaling}
 \ee
 which  is  marked by black star in Fig.~\ref{fig:scaling}.  These values fit 
 the function  $\bar{\ve} \simeq 7 (L_x)^{-1}$.
 This scaling is consistent with level spacings  of 
one dimensional extended states. 

The DC conductivity is estimated by graphically extrapolating the sequence,
\be
\sigma_{xy}^{\rm dc} = \lim_{L\to \infty} \sigma_{xy}(\bar{\ve}(L), L) \quad .
\ee
In the inset of Figure \ref{fig:scaling}, we extrapolate the Hall conductivities at  various values of disorder strength $w$ to their their respective DC limits.

 The DPP formula calculation~\cite{Bashan}  has found $\sigma_{xy}^{\rm dc} \sim 
 B w^{-4}$, as depicted in 
Fig.~\ref{fig:SXY-w}.
These results are consistent with  DB  and BE theories~(\ref{S-Drude},\ref{Sab-Boltz})) which predict that $\sigma_{xy}\propto \omega_c \tau^2$ in the weak magnetic field regime   $\omega_c\tau\ll 1$. 
The $w^{-4} $ scaling is  consistent with Fermi's golden rule for the disorder-driven scattering rate,  Eq.~(\ref{FGR-tau}),
i.e.  $\tau^{-1}\propto w^2$. 

\section{Physical consequences of DPP formulas} 

The first implication of Eqs.~(\ref{Hall-DPP}) is that the Hall conductivities are {\em on-shell} expressions which is in apparent contrast to  the original Kubo formula (\ref{Sxy-Lehmann}). Furthermore, at low temperatures, the conductivity is
dominated by the low energy excitations $|E_n-E_0|\le \cO(T)$.
Conceptual conclusions from the DPP formulas are:
\begin{enumerate}
\item  A OBC system exhibiting a non-zero Hall-type conductivity possesses quasi-degenerate (gapless) manifolds of eigenstates which are created by the presence of the magnetic field.
\item At low temperatures, since the DPP's involve low energy eigenstates, it is easy to see that $\alpha_{xy}^{\rm dc}$ and 
$\kappa_{xy}^{\rm dc}$ vanish as $T\to 0$. It is also simple to verify for  phonons and spin-wave models, that the anomalous (i.e. $a^\dagger a^\dagger, aa$) terms are in the current operators.

\item  These manifolds are subjected to the {\em non commutative geometry} generated by the  DPP's. That is to say, since
\be
\langle \left[ \tilde{P}^\alpha ,\tilde{P}^\beta  \right] \rangle \propto  i \sigma_{\alpha\beta}^{\rm dc} \quad .
\ee
then   ${1\over \sigma_{\alpha\beta}^{\rm dc}} \tilde{P}^\beta$ acts as a conjugate momentum to $\tilde{P}^\alpha$. The DPP's are therefore similar in spirit to the non-commuting guiding center coordinates which act within a quasi-degerenate Landau level in the strong magnetic field limit. 
\item In quantum Hall and topological insulator phases, the Hall-current carrying states are be supported on the OBC sample edges. For compressible metallic phases,
the Hall current is carried on chiral extended states which  percolate through the bulk. An explicit  description of these states has been used by Chalker and Coddington model~\cite{CC} to describe the quantum-percolation transition between Hall plateaux.

\item Finally, DPP formula for $\kappa_{xy}$ proves that thermal Hall currents in insulators~\cite{Matsuda-Kitaev,Behnia,Taillefer-Kxy} are carried by extended chiral modes.
\end{enumerate}

{\em Numerical advantages.}-- The elimination of the higher eigenstates in the DPP formulas at low temperatures allows us to replace $H(B)$ by its effective Hamiltonian $H^{\rm eff}(B)$  renormalized onto the lower energy Hilbert space.  $H^{\rm eff}(B)$ may be computationally easier  to work with than $H$.

\newpage
\part{Continued fractions of longitudinal conductivities}
\label{Part:Thermo I}
\section{Moments expansion}
\label{sec:moments}
Real longitudinal conductivities (in the $x$ direction) are  given by the real part of Eq.~(\ref{Kubo-OHS}),
\be
L_{ii}^{xx}(\omega)  =    {  1 \over \cV}\Im \left(J_i^x \Big| \left({1\over \cL- \omega -i\varepsilon}\right)\Big|J_i^x\right)
=   { \pi \over \cV} \left(J^x |  \delta_\ve \left(   \cL -\omega   \right)|J^x\right) \quad ,\quad i=1,2,3 \quad .
\label{Liixx}
\ee
where $\delta_\epsilon$ is a broadened Dirac $\delta$-function of width $\ve > \cV^{1\over d}$. Henceforth we suppress $\ve,\bq,\cV$ dependence, keeping in mind the eventual correct DC order of limits of Eq.~(\ref{DC}). For the moments expansion we define a symmetrized mixed electric+thermal current as
\be
J_3^x=(J_1^x + J_2^x) /2 \quad . 
\ee

Thus, the thermoelectric conductivity can be written as a sum of auto-correlation functions, 
\be
 \alpha_{xx}(\omega)=    {2\over T} L^{xx}_{33}(\omega) -\half \sigma_{xx}(\omega)- \half \kappa_{xx }(\bq,\omega) \quad .
\ee
Henceforth we discuss $L_{ii}^{xx}(\omega)$ and suppress the currents' label $i$.

The moment  of order $2k$ of $L$ is
\bea
\mu_{2k}  &=& \int \limits_{-\infty}^{\infty} {d\omega\over \pi} ~L^{xx} (\omega )\omega^{2k}\nonumber\\
&=&   {  1 \over \cV} \int \limits_{-\infty}^{\infty}   d\omega  ~\left(J^x \big|\delta \left(  \cL - \omega   \right) \cL^{2k} \big|J^x\right)  \nonumber\\
&=& { 1 \over \cV } \left(J^x \big|  \cL^{2k} \big|J^x\right) \quad .
\label{SR}
\eea
The  $k=0$ moment is the CSR,
\be
\mu_0 ={1\over \cV} (J^x\big|J^x) =
{1\over \cV} \Im  \Tr\left( \rho \left[ P^x,J^x\right]\right)   \equiv \chi_{\rm csr} \quad .
\label{CSR-def}
\ee
which uses Eq.~(\ref{JP}) to express the susceptibility as a thermodynamic expectation value.
All odd moments  vanish by antisymmetry of $mn\to nm$,
\be
(J^x|\cL^{2k+1}|J^x) =\sum_{mn} W_{mn} (E_m-E_n)^{2k+1} ~A^*_{nm}A_{nm} =0 \quad .
\label{odd}
\ee

Moments are thermodynamic expectation values of  time-independent operators,
\bea
\mu_{2k} &=&  { 1 \over \cV  } \left( J^x\Bigg| \overbrace{[H_0,[H_0,\ldots [H_0}^{2k},J^x]]\ldots] ~\right)\nonumber\\  
&=&   { 1 \over \cV } \Im \Tr ~\left(\rho    [P^x,\overbrace{[H_0,[H_0,\ldots [H_0}^{2k-1},J^x]]\ldots]\right)  \quad ,
\label{EV}
\eea
where we used  Eq.~(\ref{NoLehmann}) and  $P^x$ is the $x$-polarization corresponding to $J^x$ by  Eq.~(\ref{JP}).

Moments are coefficients of the short time Taylor series of the real-time conductivity,
\bea
L^{xx}(t) &=& {1\over \pi} \int \limits_{-\infty}^{\infty} d\omega L^{xx}(\omega)e^{i\omega t}= {1\over \cV} \left(  J^x | e^{i\cL t }| J^x\right)\nonumber\\
&=&    \sum_{k=0,\infty} {(-1)^k\mu_{2k}\over (2k)! }~  t^{2k} \quad .
\label{Lii-t}
\eea
Unfortunately,  a finite set of moments cannot, in general, determine the long time behavior of Eq.~(\ref{Lii-t}), and the low frequency behavior of Eq.~(\ref{Liixx}).

The following constraints can be  very helpful in guiding us  to a viable extrapolation scheme to high orders.  

\begin{enumerate}
\item    Since  $~\delta_\ve \left(   \cL -\omega   \right)$ is a non-negative hermitian hyperoperator, $L^{xx}(\omega)\ge 0$ is a non-negative spectral function. 
    \item $L^{xx}(\omega)=L^{xx}(-\omega)$, which therefore  permits a finite DC conductivity at $\omega=0$.

\item All the moments $\mu_{2k}$ are squares of norms  $|| \cL^{k} A||^2$, and therefore non negative.

\item By thermodynamic arguments, the real-time conductivities of Eq.~(\ref{Lii-t}) should be continuous and differentiable for $t\in(-\infty,\infty)$. Also, they are expected to decay (relax) at long times. This ensures the analyticity of $L^{xx}(z)$ in the upper half plane $\Im(z)>0$. \footnote{The only notable exception is in a superconductor, where persistent currents do not decay. For superconductors, the zero frequency conductivity is excluded from $L^{xx}(\omega)$.}
Thus, by (\ref{Lii-t}),  the asymptotic growth of the high order moments is bounded by,
\be
 \lim_{k\to \infty}{\mu_{2k}\over (2k)!}<\infty \quad.
 \label{muMAX}
 \ee
 \end{enumerate}

\subsection{Krylov bases}
An orthonormal Krylov basis of hyperstates can be constructed from the current operator. We start with the normalized root (zeroth order) state,
\be
\kket{0}  \equiv   {J\over \sqrt{(J|J)} } \quad .
\label{Krylov0}
\ee
($\kket{\bullet}$ denotes normalized hyperstates, in contrast to non normalized hyperstates such as $\big|\bullet\big)$.)

Higher order Krylov hyperstates are inductively generated by the equations
\be
\kket{n+1}   = {1\over N_{n+1}} \left( \cL \kket{n}  -\Delta_n \kket{n-1}\right) \quad ,~~~~N_{n+1}\equiv \left( \kexpect{n}{\cL^2}{n} - \Delta_n^2 \right)^{1\over 2} \quad ,\quad n=1,2,\ldots \infty
\label{Krylov-n}
\ee
where $\Delta_n$ is the {\em recurrent} of order  $n$, which is equal to the  
matrix elements of the Liouvillian,
\be
\Delta_n = \kexpect{n}{\cL}{n-1} = \kexpect{n-1}{\cL}{n} \quad .
\label{Delta-n1}
\ee
In order to evaluate $\Delta_n$  the following statements are verified:

\begin{enumerate}
\item The Liouvillian expectation values vanish in the Krylov basis.  By construction of Eq.~(\ref{Krylov-n}), even (odd) order Krylov hyperstates involve states with  even (odd) powers $\cL^n|J)$. Hence we can expand any Krylov hyperstate as
\be
\kket{n}  = \sum_{k=0}^{^{  {\rm Int}(n/2)} } a_{n-2k}  \cL^{n-2k}\kket{0} \quad .
\ee
The expectation values of $\cL$ in Krylov hyperstates are,
\bea
\kexpect{n}{\cL}{n} &=&\sum_{k,k'=0}^{  {\rm Int}(n/2)}   a_{n-2k}a_{n-2k'} \kexpect{0}{\cL^{n-2k} \cL \cL^{n-2k'}}{0}  \nonumber\\
&=&   
\sum_{k,k'=0,n}   a_k b_{k'} \kexpect{0}{\cL^{2(n - k- k')+1}}{0} =0 \quad ,
\eea
which follows from Eq.~(\ref{odd}).

\item The Krylov basis is orthonormal,
\be 
\kbraket{m}{n} = \kbraket{n}{m} = \delta_{nm} \quad , \quad n,m=0,1,2,\ldots \infty
\label{ortho}
\ee
which is proven by induction using Eq.~(\ref{Krylov-n}) starting from $n,m=0$. 
\item The Liouvillian representation in the Krylov basis is, 
\bea
L_{nm}&=&  \kexpect{n}{\cL}{m} = \delta_{m,n-1} \Delta_n +\delta_{n,m-1} \Delta_{n+1} \quad , \nonumber\\
\\
&=&   \left( \begin{array} {cccccc}
0&\Delta_1&0&0&0&\ldots \\
\Delta_1&0&\Delta_2&0&0&\ldots \\
0&\Delta_2&0&\Delta_3&0&\ldots \\
0&0&\Delta_3&0&\Delta_4&\ldots \\
0&0&0&\Delta_4&0\ldots \\
0&\vdots &\vdots&\vdots&\vdots&\ldots \\
\end{array}\right)_{nm}\quad  n,m=0,1,2,\ldots\infty \quad .
\label{Lmn}
\eea

The matrix $L_{nm}$  can be regarded as a tight binding hopping Hamiltonian on a half chain, with hopping parameters  $\Delta_n$.
\item 
The phases of $|n\rangle$ can be chosen such that we can gauge all the recurrents $\Delta_n$ to be  real and positive. 
\end{enumerate}

$G_{nm}(z)$ is the ``Liouvillian Green  function'',
\be
G_{nm} (z) \equiv  \kexpect{n}{\left( z-\cL \right)^{-1}}{m} \quad .
\label{Gnm}
\ee
Since $L$ is symmetric, so is $G_{mn}(z)$, which is in general complex.
On the real axis 
$z\to \omega+i0^+$, $-\Im G_{nn}(\omega+i\ve)=-G_{nn}''(\omega)$ is a  {\em spectral function}. It is
non-negative and symmetric in  $\omega\to-\omega$. 
The real part of the complex function $G_{nn}(z)$ can be obtained from the spectral function by the Kramers-Kronig (KK) transform,
\be
G_{nn}'(\omega)={1\over \pi} \int \limits_{-\infty}^{\infty} d\omega' ~{G_{nn}'' (\omega')\over \omega-\omega'} \quad .
\label{KK}
\ee
$G''_{nn}(\omega)=G''_{nn}(-\omega)$, and by Eq.~(\ref{KK}), $G'_{nn}(0)=0$. Non diagonal Green functions $G_{nm}(0+i\ve)$ can be obtained  from the inversion equation $(G L)_{nm} =\delta_{nm}$,
\be
\sum_{m\ge 0} G_{nm}(0^+)~ L_{mn'}  = \delta_{nn'} \quad .
\ee
For example,
\bea
\sum_{m} G_{0 m}(0^+) L_{m 1}&=& G_{00} \Delta_1 + G_{02}\Delta_2=0, ~\Rightarrow ~ G_{02}(0^+)= i G''_{00} \left(-{\Delta_1\over \Delta_2}\right) \quad ,\nonumber\\
\sum_{m} G_{0 m}(0^+)L_{m 4}&=& G_{02}(0^+) \Delta_3+ G_{04}(0^+)\Delta_4=0+  ~\Rightarrow ~  G_{04}(0^+)= i G''_{00} \left(-{\Delta_1\over \Delta_2}\right)\left(-{\Delta_3\over \Delta_3}\right) \quad ,\nonumber\\
\vdots && \vdots,
\eea
and generally,
\bea
G_{0,2k} (0^+) &=& G_{2k,0}(0^+)  = i G''_{00}  R_k \quad , \quad k=0,1,2\ldots\nonumber\\
R_k &=& \prod_{k'=1}^k \left(-{\Delta_{2k'-1}\over \Delta_{2k'}}\right) \quad .
\label{G02k}
\eea
Another important relation is,
\be
 G_{1, n}(0^+)  ={\delta_{n,0}\over \Delta_1} \quad ,
\label{G1n}
\ee
which is purely real.
As a consequence, the matrix elements of the imaginary inverse Liouvillian do not connect to the first Krylov state: 
\be
\kexpect{n}{\left({1\over \cL }\right)''}{1} = \kexpect{1}{\left({1\over \cL }\right)''}{n} =  0 \quad .
\label{G1n-useful}
\ee
This result will prove to be useful in  Section \ref{sec:dSdB}.

\subsection{The continued fraction representation}
\label{sec:LI}
Since $z-L$ is a tridiagonal matrix, an iterative inversion can be used to invert $G_{00}$ in Eq.~(\ref{Gnm}), and express it as
\bea
G_{00} (z) &=&    
 {1\over    z-\Delta_1^2  G^>_{11}(z) }~ \quad ,\nonumber\\
G^>_{11} (z) &=&   {1\over   z-\Delta_2^2  G^>_{22}(z)  }~ \quad ,\nonumber\\
\vdots &=& \quad  \vdots \quad,\nonumber\\
\label{CF1}
\eea

where $G^>_{nn}$ is the {\em termination Green function} on the half-chain with  the sites $m\ge n$.
The sequence of equations (\ref{CF1}) comprises an infinite continued fraction (CF),
 \begin{equation}
G_{00}(z)\simeq {1\over z - {\vvv\textstyle{\Delta_1^2} \over \textstyle{z \, - \, } {\vvv\textstyle \Delta_2^2\over 
\textstyle{z \, - \, } {\vvv\textstyle \Delta_3^2 \over
{ \textstyle {\vdots} } } } }} \quad .
\label{CF2}
\end{equation} 

\section{From moments to recurrents}
\label{sec:moments-rec}
The first question that comes to mind is: How do we obtain the recurrents $\Delta_1,\Delta_2\ldots\Delta_k$, from a given sequence of moments  $\mu_0,\mu_2\ldots \mu_{2k}$ ? 

Eq.~(\ref{SR}) and  (\ref{Lmn})  generate a sequence of identities:
\bea
\mu_{2k} &=&\mu_0~ \langle0| (L[\Delta])^{2k}|0\rangle \quad .\nonumber\\
{\mu_2\over 
\mu_0}&=& \Delta_1^2 \quad ,\nonumber\\
{\mu_4\over 
\mu_0}&=& \Delta_1^4 + \Delta^2_1 \Delta^2_2 \quad ,\nonumber\\
{\mu_6\over 
\mu_0 }&=& \Delta_1^6 + 2 \Delta_1^4 \Delta_2^2 + \Delta_1^2 \Delta_2^4 + \Delta_1^2 \Delta_2^2 \Delta_3^2 \quad ,\nonumber\\
 {\mu_8\over 
\mu_0}&=&  \Delta_1^8 + 3 \Delta_1^6 \Delta_2^2 + 3 \Delta_1^4 \Delta_2^4 + \Delta_1^2 \Delta_2^6 + 2 \Delta_1^4 \Delta_2^2 \Delta_3^2 + 
 2 \Delta_1^2 \Delta_2^4 \Delta_3^2 + \Delta_1^2 \Delta_2^2 \Delta_3^4 + \Delta_1^2 \Delta_2^2 \Delta_3^2 \Delta_4^2 \quad ,\nonumber\\
\vdots && \quad\quad \vdots
\label{mu-Delta1}
\eea
Importantly,  $\mu_{2k}/\mu_0$ depend  only on the  recurrents $\Delta_1,\Delta_2\ldots \Delta_k$.
This allows us to solve for the recurrents iteratively, starting at $k=1$. As a result we obtain
the algebraic equations, 
\bea
\Delta_1^2&=& \bmu_2 \quad ,\nonumber\\
\Delta^2_2 &=& {\bmu_4\over \bmu_2} - \bmu_2 \quad ,\nonumber\\
\Delta_3^2 &=& {\bmu_4^2  -\bmu_2\bmu_6\over\bmu_2^3 - \bmu_2\bmu_4 } \quad ,\nonumber\\
\Delta_4^2 &=&  \bmu_2{(\bmu_4^3 + \bmu_6^2 + \bmu_2^2\bmu_8 - \bmu_4 (2\bmu_2 \bmu_6 + \bmu_8))\over
(\bmu_2^2 - \bmu_4 ) (\bmu_2 \bmu_6 - \bmu_4^2 )} \quad ,\nonumber\\
\vdots &=& \vdots
\label{Delta-mu}
\eea
In general $\mu_{2k}$ increase very rapidly (of order $(2k)!$) with $k$, while $\Delta_k$ are much smaller numbers which increase at a much slower rate. The presence of subtractions in Eq.~(\ref{Delta-mu}) imply that small relative errors in $\mu_{2k'}$ may result in  large relative errors in $\Delta_{k\ge k'}$~\footnote{We thank Snir Gazit for alerting us to this numerical challenge.}

\begin{figure}[!ht]
    \centering
  \includegraphics[width=0.6\columnwidth]{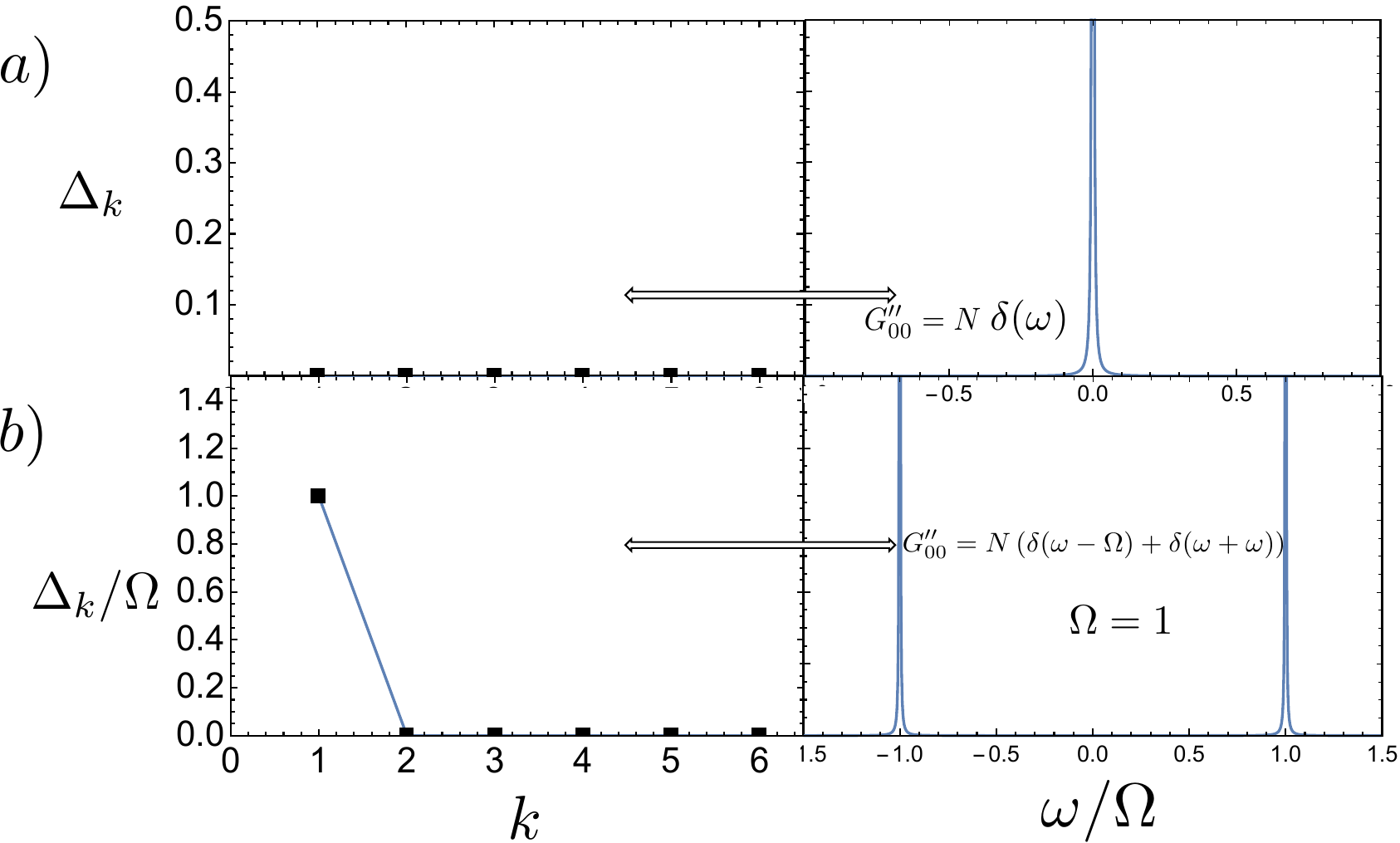}
    \caption{Single mode correlations.}   \label{fig:SMA}   
 \end{figure}

\subsection{Single Mode Spectra} 
Let us start with a single sharp mode at zero frequency as shown in Fig.~\ref{fig:SMA}(a),
\be
-G_{00}''(\omega) =  \pi \delta(\omega) \quad . 
\ee
All the recurrents vanish identically, since $[H,A]=0$.
\be
\Delta_k^2=0 \quad ,\quad \forall n \quad .
\ee
A single mode at finite frequency $|\Omega|>0$ is shown in Fig.~\ref{fig:SMA}(b). It is described by 
\be
-G_{00}''(\omega) =  \pi\left( \delta(\omega-\Omega)+\delta(\omega+\Omega)\right)/2 \quad . 
\ee
The only non-zero recurrent is $\Delta_1=\Omega$. Physically this occurs when  $J^x$ is an eigen-operator of the Liouvillian
\be
\cL J^x = \Omega J^x \quad .
\ee
and therefore applying $J^x$ to the Hamiltonian's ground state $|\Psi_0\rangle$ creates an   eigenstate $J^x|\psi_0\rangle$  with excitation energy  $\Omega$. 

An approximation  which {\em assumes} a spectral function which approximately is described by a $\delta$-function, is the single mode approximation. This approximation was used by Feynman~\cite{Feynman} to describe the roton minima in the spectrum of superfluid helium~\cite{Feynman}, by  Girvin, Macdonald and Platzman~\cite{GMP} to describe the magneto-roton excitation of the Fractional Quantum Hall phase, in Refs.~\cite{AAH,IEQM} to approximate the Haldane gap in a spin-one chain model.

\section{From recurrents to conductivities}
\label{sec:Recurrents-Cond}

The  longitudinal conductivity of Eq.~(\ref{Liixx}) is determined by the zeroth order spectral function,
\bea
L^{xx}(\omega) &=&-\mu_0 \lim_{z\to \omega+i0^+} G_{00}''(\omega +i0^+) \nonumber\\
&=& \pi   \mu_0 \kexpect{0}{\delta(\cL-\omega)}{0} \quad , \eea
where $L[\Delta_n]$ is the infinite tridiagonal matrix of Eq.~(\ref{Lmn}), and $G_{00}(z)$
is defined as a CF in Eq.~(\ref{CF2}).

In practice, the calculation of $\Delta_n$ is limited to a finite sequence. However, truncating the CF at any finite order does not lead to a continuous function, but to a sequence of $\delta$-functions. In effect, determining $G_{00}''$ amounts to inversion of the moments series, which is in general and unsolved problem. 

CF however allow us to
design extrapolation schemes which may be suitable for certain class of physical problems, which provide additional information about the desired  $G_{00}''$.  For example, $-G''_{00}(\omega)$ is a positive and symmetric function, whose all its recurrents are finite, non-negative numbers. 

Since extrapolation is a tricky art, it is useful to first learn about rigorous results  relating  high frequency asymptotics of $G_{00}''$ and high order recurrents. 

\subsection{Freud's high order asymptotics}
For a large class of smooth spectral function $-G_{00}''(\omega)$ with support on the whole  frequency axis~\footnote{We thank Ari Turner for explaining to us the mathematical literature reviewed in this section}, Freud~\cite{Freud} has conjectured an asymptotic relation between the high frequency asymptotic fall-off of the spectral function,  and the   asymptotic behavior of the high order  recurrents $\Delta_k$ as  $k\to \infty$. 

This conjecture was proven by Lubinsky, Mhaskar and Saff (LMS)~\cite{Lubinsky}, for spectral functions described by
\be
-G_{00}''(\omega)=\exp\left(-Q(\omega) \right) \quad .
\label{GQ}
\ee
The fall-off exponent $Q$ is assumed to satisfy the following conditions:

\begin{enumerate}
    \item  $Q(\omega)=Q(-\omega)$.    
    \item  $Q'(\omega)$ exists for $\omega\ne 0$, and $\omega Q'(\omega)$ is bounded near the origin as $\omega\to 0^+$. Furthermore, $Q''$ exists for large enough $\omega$. 
    
    \item For finite values of $C>0, \alpha>0$  at large enough $\omega$,
\bea
&&Q'(\omega)>0 \quad ,\nonumber\\
&&\omega^2 |Q''(\omega)|/Q'(\omega)\le C \quad ,\nonumber\\
&&
\lim_{\omega\to \infty} \left( 1+ \omega   Q''(\omega) /Q'(\omega) \right)=\alpha>0 \quad .
\eea
\end{enumerate}

\begin{figure}[!ht]
    \centering
  \includegraphics[width=0.6\columnwidth]{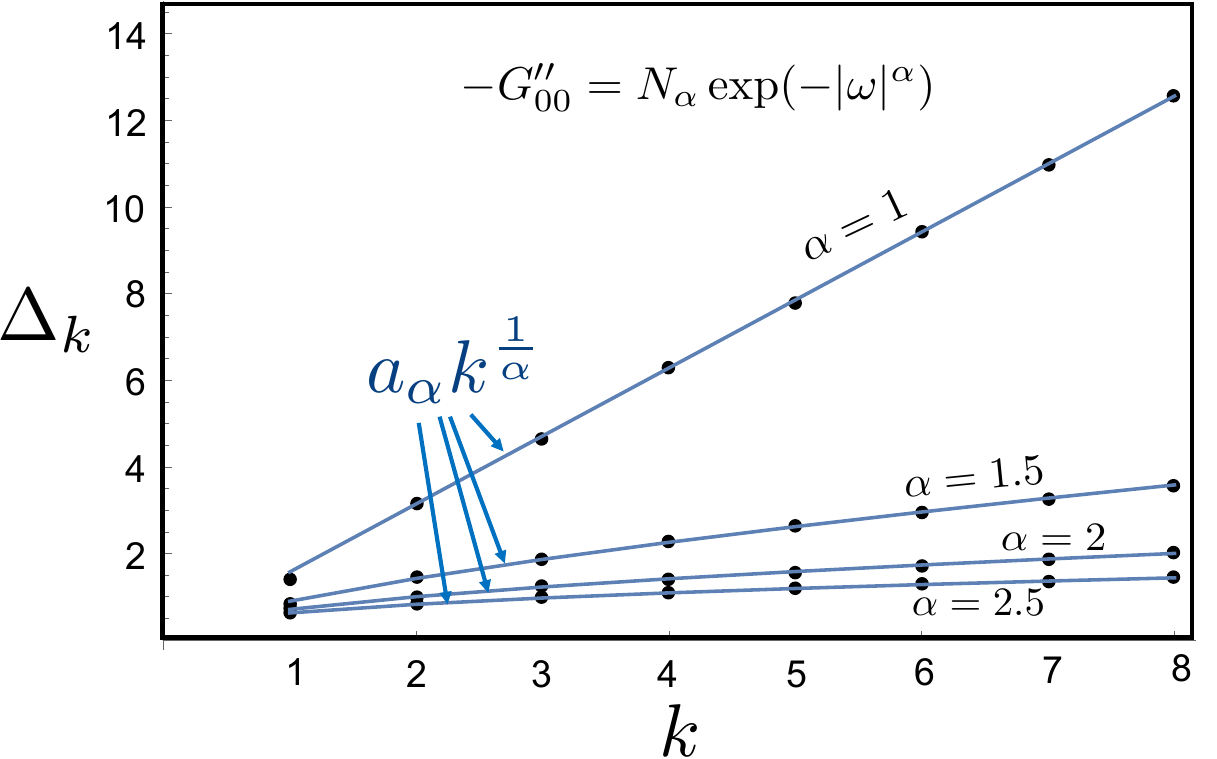}
    \caption{The recurrents of the pure stretched exponentials (black circles), and Freud's asymptotic expression (blue  lines) 
    as governed by Eq.~(\ref{Freud}). The differences between the values are plotted in Fig.~\ref{fig:FreudErrors}.
 \label{fig:Freud}}   
 \end{figure}

LMS proved that the corresponding recurrents $\Delta_n$ for Eq.~(\ref{GQ})  exhibit the asymptotic behavior,
\bea
\lim_{k\to \infty}  \Delta_k  &\sim &a_\alpha k^{1\over \alpha}\nonumber\\
a_\alpha&=& {1\over 2} \left({\sqrt{\pi}\Gamma\left({\alpha\over 2}\right) \over     \Gamma\left({\alpha+1\over 2}\right) } \right)^{1\over \alpha} \quad .
\label{Freud}
\eea

Interestingly, for pure stretched exponential
\be
Q(\omega,\alpha)= |\omega|^\alpha \quad , 
\ee
as shown in Figs.~\ref{fig:Freud} and \ref{fig:FreudErrors}, 
the exact  low order recurrents  are close but slightly different from the analytic asymptotic values of Eq.~(\ref{Freud}), except for the pure Gaussian  $\alpha=2$, for which Eq.~(\ref{Freud}) is exact.     

\begin{figure}[!ht]
    \centering
  \includegraphics[width=0.6\columnwidth]{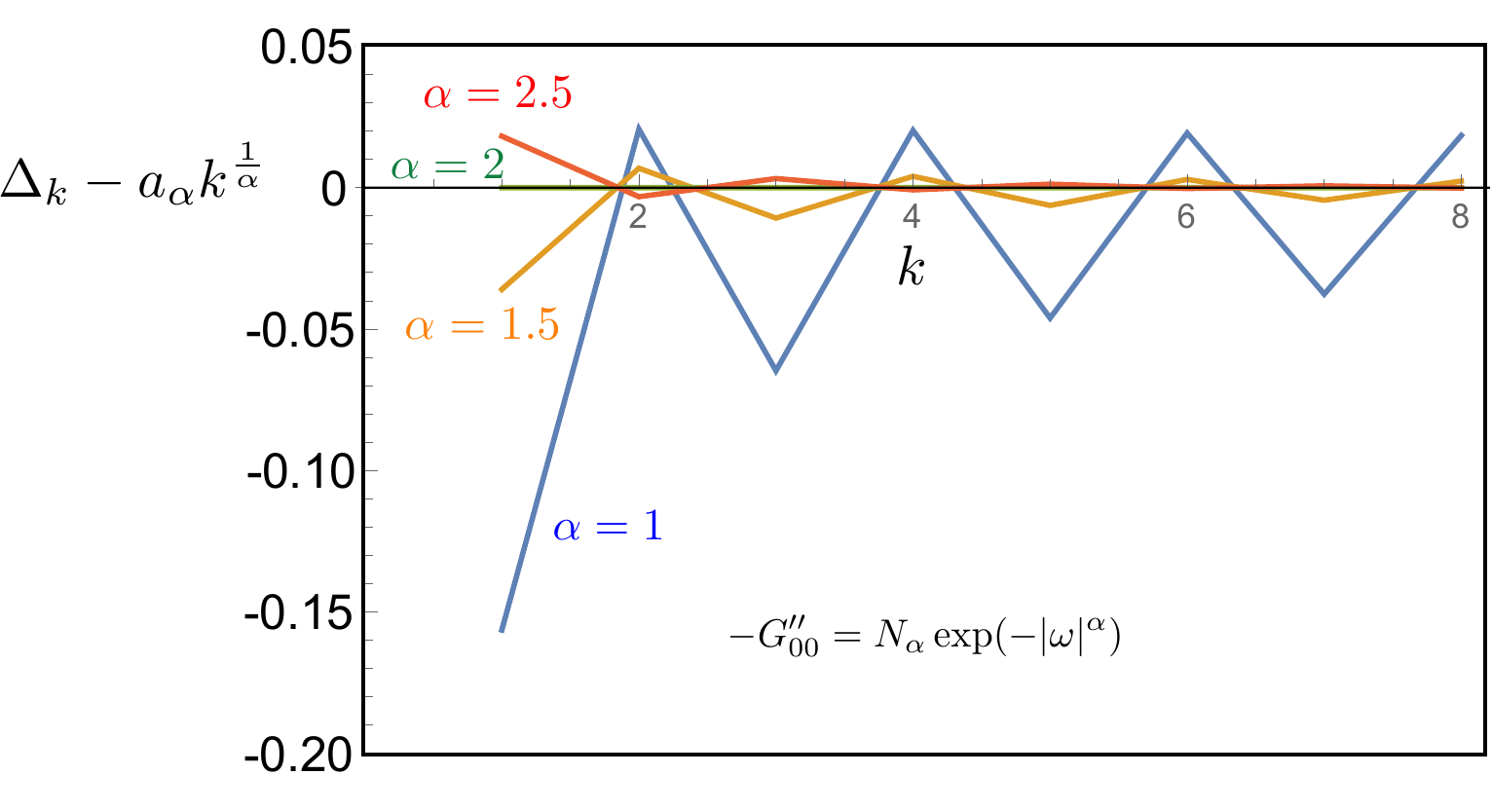}
    \caption{The difference between the true recurrents of the stretched exponential and asymptotic values of Eq.~(\ref{Freud}), which decrease with the order $k$. Note that the asymptotic values are exact only for the Gaussian $\alpha=2$.
 \label{fig:FreudErrors}}   
 \end{figure}

In general, since the moments and recurrents are non-negative, by Eqs.~(\ref{mu-Delta1}),
the moments grow at least as fast as 
\be
\mu_{2k}\ge (k!)^{2 / \alpha} \quad .
\ee
Due to assumed continuity of $G_{00}(t)$ on the real-time axis,  by Eq.~(\ref{muMAX}), the moments cannot grow faster than $\mu_{2n}\le (2n)! $.

This implies that for continuous and bounded response functions,  the asymptotic power $\alpha$ of $Q(\omega)$   
is limited to
\be
1\le \alpha\le \infty \quad .
\ee

\subsection{Termination functions}
Assume that we  have calculated a finite set of recurrents
\be
\{\Delta \}_1^{k_{\rm max}}= (\Delta_1,\ldots \Delta_{k_{\rm max}}) \quad .
\label{Delta-n2}
\ee
This finite set is not sufficient to describe the continuous function $G_{00}''(\omega)$.  Extrapolation of recurrents is tantamount to finding an accurate termination function  $\bar{G}^>_{k_{\rm max}}(z)$ such that,
\begin{equation}
G_{00}(z)\simeq {1\over z - {\vvv\textstyle{\Delta_1^2} \over \textstyle{z \, - \, } {\vvv\textstyle \Delta_2^2\over { \textstyle {\quad\ddots} \atop \qquad\qquad 
{\textstyle{\Delta_{k_{\rm max}-1}^2} \over \vvv\textstyle{z - \Delta_{k_{\rm max}}^2 \bar{G}^>_{k_{\rm max}}(z;[\alpha])}}} } } }\quad .
\label{CF-term}
\end{equation} 
$\bar{G}^>_{k_{\rm max}}(z)$ is the normalized CF which contains only the higher order  recurrents $\bar{\Delta}_k,k\ge k_{\rm max}+1$,
\begin{equation}
\bar{G}^>_{k_{\rm max}}(z)= {1\over z - {\vvv\textstyle{\bar{\Delta}_{k_{\rm max}+1}^2} \over \textstyle{z \, - \, } {\vvv\textstyle \bar{\Delta}_{k_{\rm max}+2}^2\over 
\textstyle{z \, - \, } {\vvv\textstyle \bar{\Delta}_{k_{\rm max}+3}^2 \over
{ \textstyle {\vdots} } } } }}\quad .
\label{CF-termFunc}
\end{equation}

If $\{\bDelta\}_{k},~k=0,\ldots\infty$ are the known recurrents of a complex variational spectral function $\bar{G}_{00}(z)$.
The variational reurrents  are used to produce the termination function by iteratively inverting the CF of the complex function $\bar{G}_{00}(z)$,
\bea
\bar{G}^>_{11}(z) &=&  {1\over \bDelta_1^2} \left( z-{1\over \bar{G}_{00}(z)}\right)\nonumber\\
\vdots &=& \vdots\nonumber\\
\bar{G}^>_{k_{\rm max}} &=&    {1 \over \bDelta_{k_{\rm max}-1}^2 } \left(  z-{1\over \bar{G}^>_{k_{\rm max}-1 }(z)}\right) \quad.
\label{TF}
\eea
We emphasize that or the inversion, both real and imaginary parts of $\bar{G}_{00}(\omega+i\ve)$ are required, where the real part is obtained by a Kramers-Kronig transformation,
Eq.~(\ref{KK}), of the imaginary part.

Termination functions can be used to study  the effects of  low order recurrents on the frequency dependence of $G_{00}''(\omega)$. 

\subsection{Low frequency behavior}  
By  LMS theorem, the high order recurrents can
determine  high frequency decay of  a large class of spectral functions. Conversely, low frequency behavior can be deduced in certain cases from the low order recurrents. We first demonstrate the low frequency effects of varying  the first recurrent, and then we  show the low frequency effects of alternating even-odd deviations of recurrents
from their high order asymptotic behavior.

The recurrents of the semicircle spectral function function,
\be
-G_{00}''(\omega) = {1\over 2\Delta^2}\sqrt{2\Delta^2-\omega^2} \Theta(2\Delta^2-\omega^2)
\label{SC}
\ee
are,
\be
\bar{\Delta}^{\rm sc}_k=\Delta \quad ,\quad k=1,2,\infty \quad .
\label{SCT}
\ee

The recurrents of a Gaussian spectral function are,
\be
-G_{00}''(\omega) = \sqrt{\pi} e^{-\omega^2}
\ee
are,
\be
\bar{\Delta}^{\rm Gauss}_k=\sqrt{k\over 2} \quad ,\quad k=1,2,\infty \quad .
\label{GT}
\ee

\begin{figure}[!ht]
    \centering
  \includegraphics[width=0.7\columnwidth]{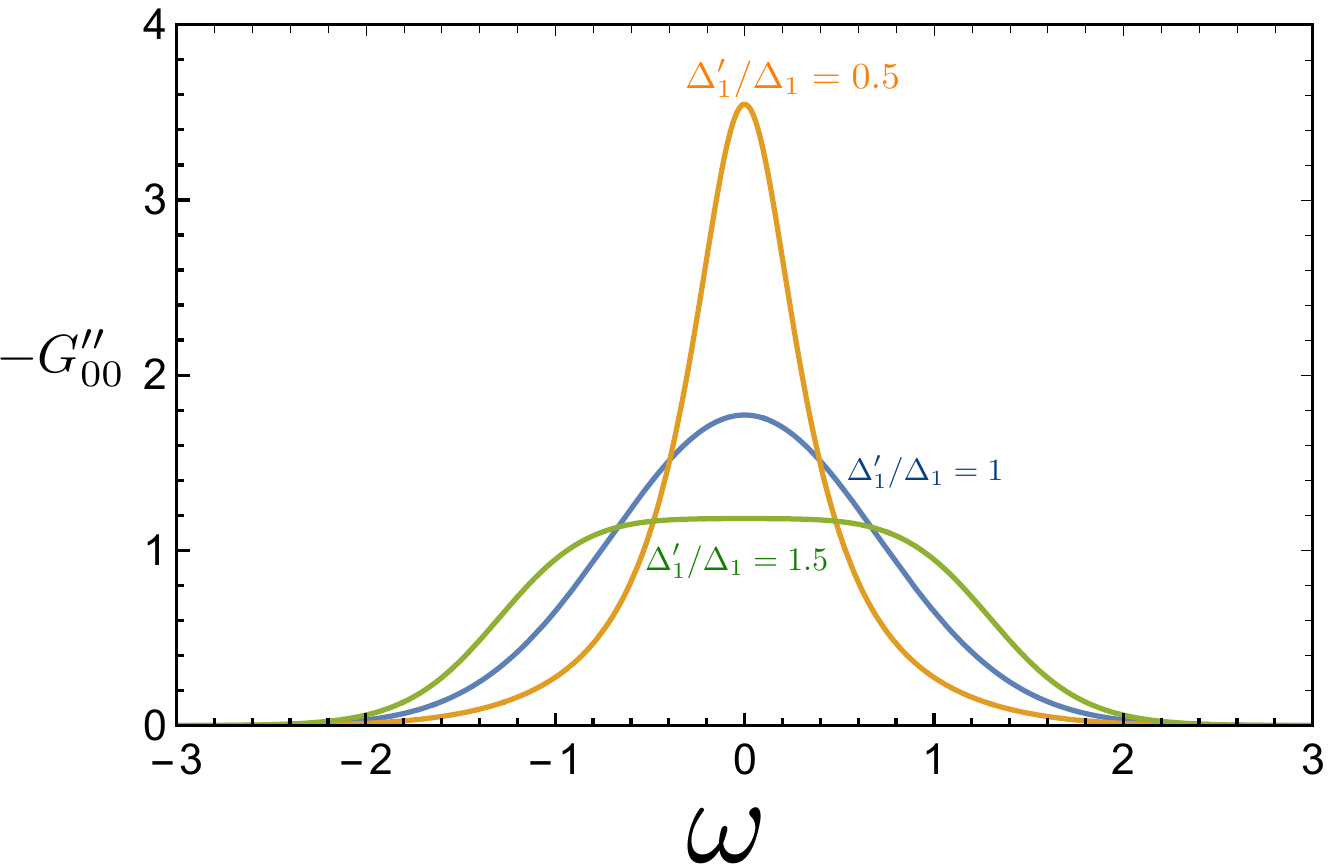}
    \caption{Changing just the first recurrent $\Delta_1'$ in the Gaussian spectral function.}
    \label{fig:GaussD1}
\end{figure}

\begin{figure}[!ht]
    \centering
  \includegraphics[width=0.7\columnwidth]{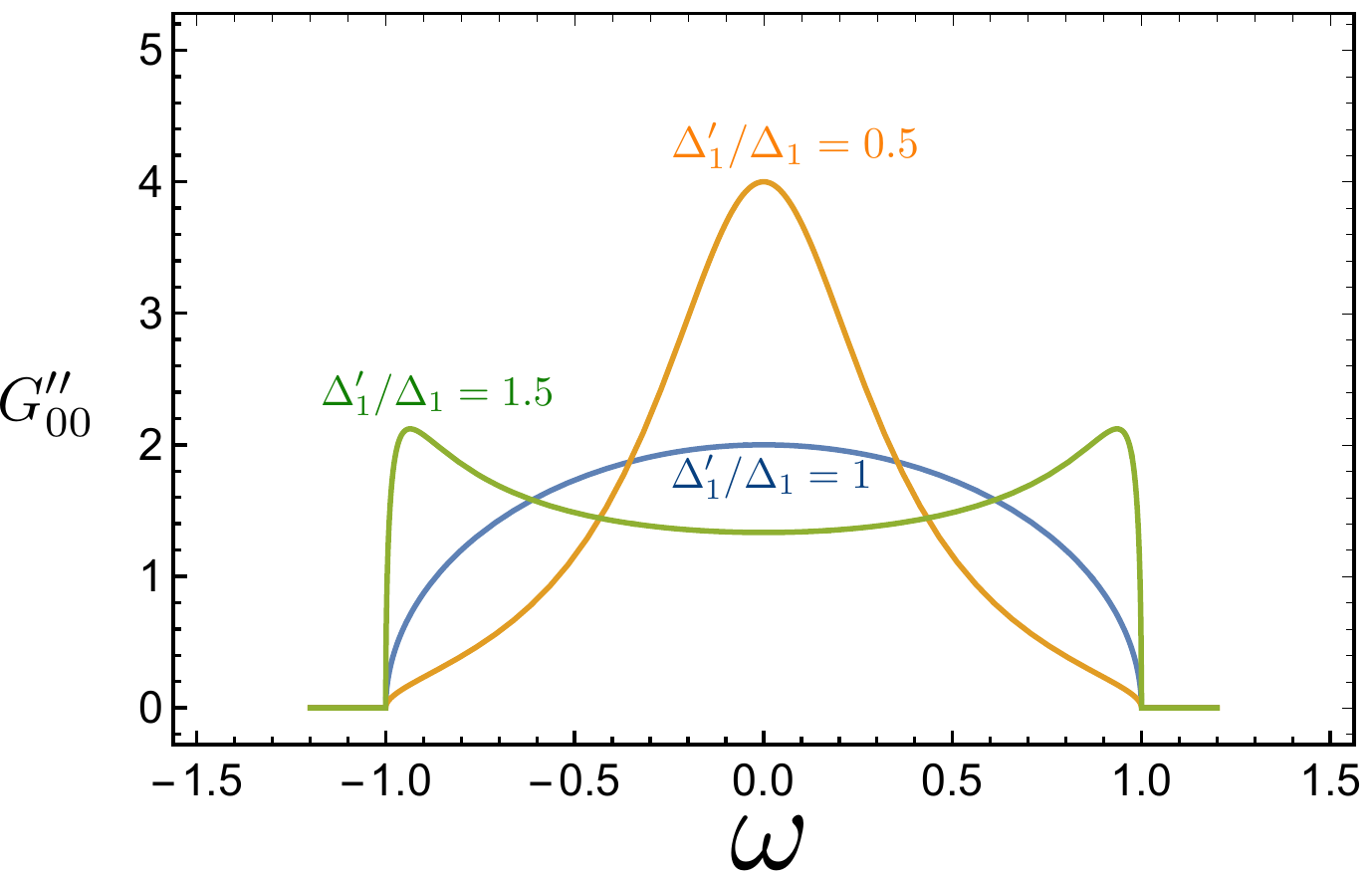}
    \caption{Changing just the first recurrent $\Delta_1'$ in the semicircle spectral function.
    }
 \label{fig:SCD1}  
 \end{figure}
As noted by the memory function approach, in the weak scattering limit, the first order recurrent plays an important role in determining the low frequency conductivity. 

The first recurrent $\Delta_1$ 
dominates the DC conductivity.
 This is demonstrated by varying $\Delta_1'$ in  Fig. \ref{fig:GaussD1}, keeping the higher order recurrents the same.  For both the Gaussian in Fig.~\ref{fig:GaussD1}, and the semicircle in  Fig.~\ref{fig:SCD1}, the DC limit of the resulting normalized modified functions vary with $\Delta_1'/\Delta_1$ as
\be
{ G_{00}''(0)\over \bar{G}_{00}''(0)} = {(\Delta_1')^2  \over \Delta_1^2 } \quad .
\label{RatioSxx}
\ee

 \begin{figure}[!ht]
    \centering
  \includegraphics[width=0.7\columnwidth]{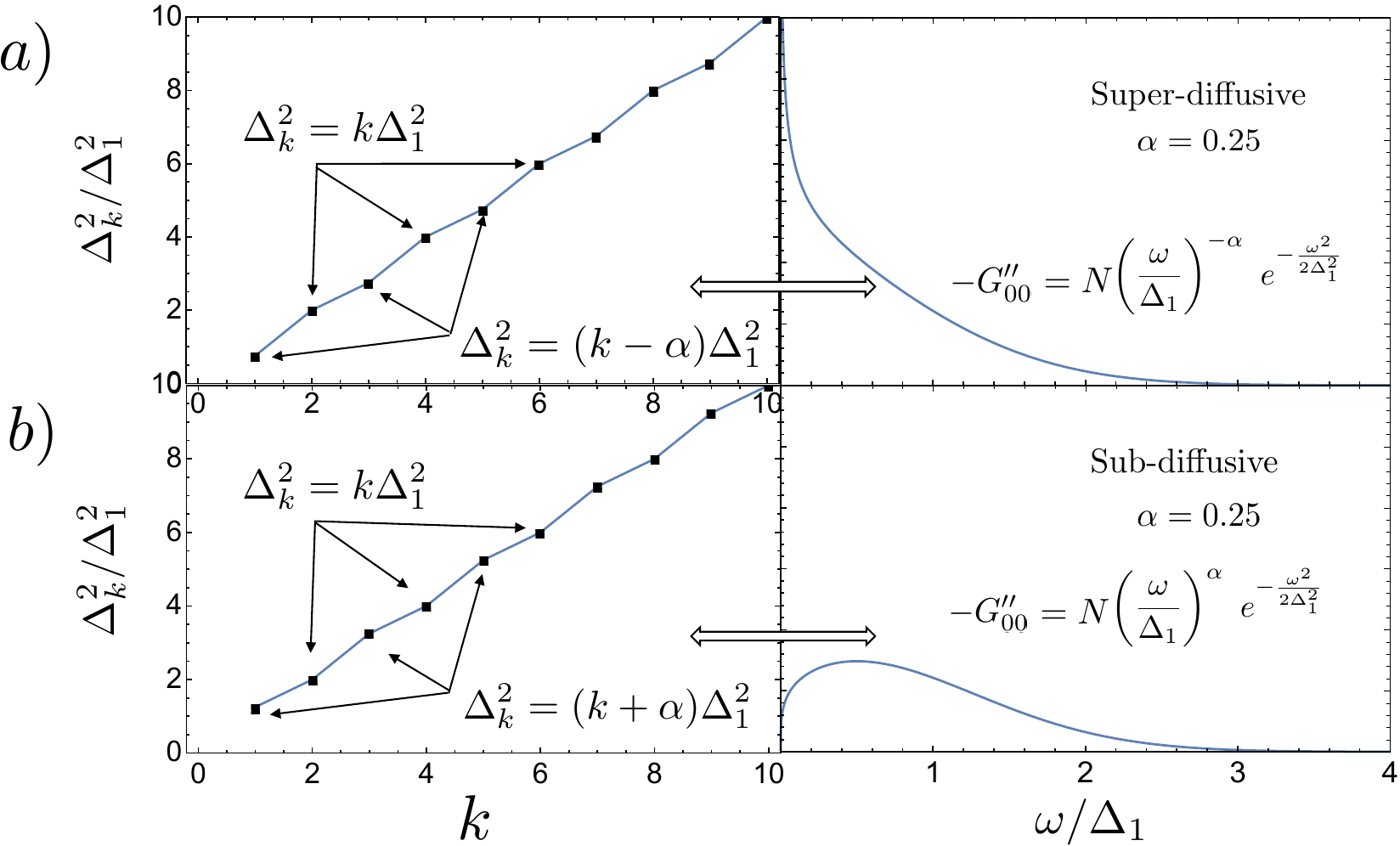}
    \caption{Continued fraction recurrents of PG  functions. High frequency Gaussian tail is reflected in the average linear increase of $\Delta_k^2$. Low frequency divergence (vanishing) in the super-diffusive (sub-diffusive) functions is signalled in the the oscillation of even and odd recurrents   starting at the lowest orders. 
 \label{fig:PGauss}}      
\end{figure}

Stronger modifications of the  low frequency behavior is produced by  even-odd deviations of the recurrents from the asymptotic form.  This is demonstrated by the power-law times Gaussian function (PG), as discussed in detail in Ref.~\cite{Viswanath},
\be
 -  G_{\rm PG}''(\omega) =    {\sqrt{2}\pi\over \Omega\Gamma\left({\beta + 1\over2} \right) } ~\left| {\omega\over  \Omega }\right|^\beta \exp\left( -{\omega^2 \over  \Omega^2 }\right) \quad .
   \label{PGauss}
    \ee  
The recurrents of the PG have been evaluated  analytically,
    \be 
  \left( \Delta^{\rm PG}_k\right)^2 =  {1\over 2}   (k+\beta) ~ \delta_{k,{\rm odd}} ~ + ~{1\over 2} k ~\delta_{k,{\rm even}} 
   \ee
and are plotted versus their low order recurrents in Fig.~(\ref{fig:PGauss}).
The zero frequency conductivity vanishes or diverges depending on the sign of $\beta$. While the relative deviations from the asymptotic $k^{1/2}$  behavior at large $k$ decreases, the effects of the low order recurrents are dramatic at low frequencies. A similar effect is seen numerically for power law functions times stretched exponentials at $\alpha\ne 2$~\cite{Viswanath}.

\subsection{Addition of spectral functions with 
different frequency scales}
   
\begin{figure}[!ht]
    \centering
  \includegraphics[width=0.7\columnwidth]{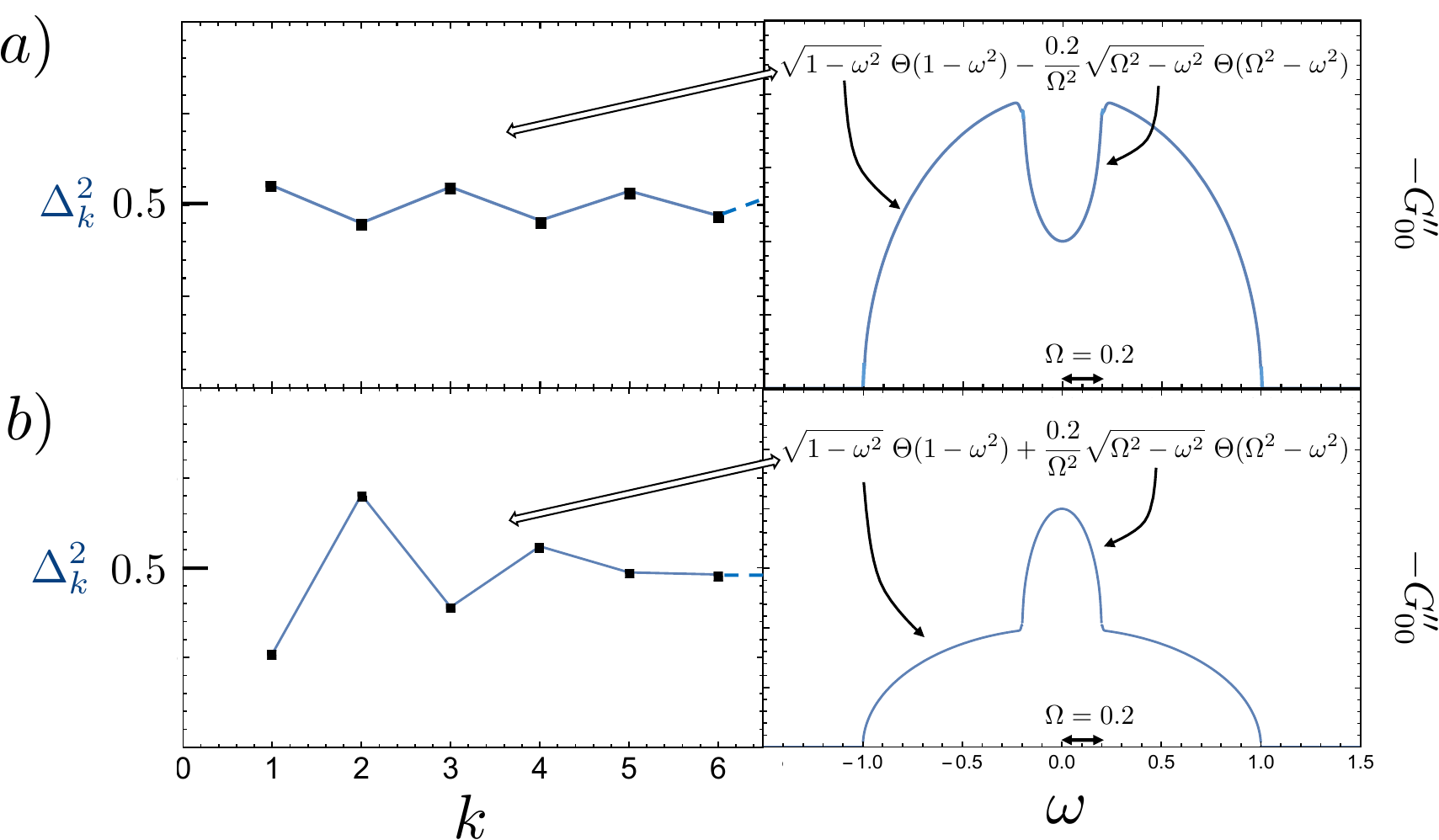}
    \caption{General example of low frequency effects. Top Frames:  low frequency suppression of a semicircle function, corresponds to odd order recurrents larger than even orders. Bottom frame: low frequency enhancement  corresponds to even order recurrents higher than odd orders. This is qualitatively similar to the effects seen in the recurrents of PGauss functions in Fig.~\ref{fig:PGauss} }
    \label{fig:SCpm}
\end{figure}

The low frequency regime of the response function
affects the behavior of the recurrents in a more  complicated way  than the high frequency asymptotics.   For example we consider
a sum of two semicircles with different energy scales $\Omega_2\ll \Omega_1$ 
\be
G''_{\rm 2sc}(\omega)=   G_{\rm sc}''(\omega;\Omega_1) \pm  a_2 G''_{\rm sc}(\omega;\Omega_2) \quad ,
\ee
where $G_{\rm sc}''(\omega;\Omega)$ is given by Eq.~(\ref{SC}). 
$a_2$ may be positive (for enhanced DC conductivity), or negative, 
for suppressed DC conductivity. 
The moments are additive,  
\be
\mu^{\rm 2sc}_{2k} = \mu_{2k}^{\rm sc1} \pm a_2 \mu_{2k}^{\rm sc1} \quad . 
\ee
(where we must use $a_2<1$, since all moments must be non-negative).
After hashing the algebraic relations (\ref{Delta-mu}), the combined recurrents are highly entangled. The effects of summing two spectral functions are shown in
Fig.~\ref{fig:SCpm}. The addition of two spectral functions results in  even-odd oscillations of the recurrents which persist to high orders, although its signature is already observed at low orders.

\subsection{Variational Extrapolation of Recurrents}
\label{sec:VER}
Since nested  commutators of the Hamiltonian with the current create a factorial growth of number of operators,
the computational cost of high order moments and recurrents in many cases increases faster than exponentially. 

Here we describe the Variational Extrapolation of Recurrents as a possible scheme which allows for a numerical test of its convergence.
 
Having computed $\Delta_1,\ldots, \Delta_{n_{\rm max}}$, may lead to a reasonable guess for
the possible asymptotic behavior at large and low frequencies. We choose a family of analytic variational spectral functions $\bar{G}_{00}''(\omega;\{\alpha \})$,
which can be parametrized by variational parameters
$\alpha_1,\ldots \alpha_{m_{\rm var}}$.

Using Kramers-Kronig relation (\ref{KK}) find the real part of $\bar{G}_{00}(\omega+i0^+)$, 
\be
\bar{G}_{00}'(\omega;\{\alpha\})= {1\over \pi} \int \limits_{-\infty}^{\infty} d\omega' \bar{G}_{00}''(\omega;\{\alpha\}) \quad ,
\label{bG00}
\ee
and calculate the variational recurrents $(\bDelta_1\ldots \bDelta_{n_{\rm max}})$ from the lowest $2n_{\rm max}$ moments of $\bar{G}''(\omega)$.

The variational parameters are obtained by minimizing the least squares function with respect to $\{\alpha\}$,
 \be
\chi^2= \underset{\{\alpha_1,\ldots\alpha_{m_{\rm max}}\}}{\rm min} {1\over n_{\rm max}-n_{\rm min}  } \sum\limits_{n=n_{\rm min}+1}^{n_{\rm max}} \left( {\Delta_n -  {\bDelta}_n(\{\alpha\})\over \Delta_n}\right)^2 \quad .
\label{eq:LS}
\ee

The lower cut-off $n_{\rm min} \ge 1 $ is  chosen to preferentially fit the recurrents of $\bar{G}_{00}$ to the higher orders of the calculated recurrents. The total number of variational parameters $m_{\rm var}$ must be smaller than $n_{\rm max}-n_{\rm min}$ to avoid overfitting.

Inverting Eq.~(\ref{CF-term}) provides  $\bar{G}_{n_{\rm max},n_{\rm max}}^>(z)$ which directly provides the Variational Extrapolation of Recurrents (VER) approximation for $G_{00}$ (see Refs.~\cite{LA,Ilia}
\begin{equation}
G^{\rm ver}_{00}(z)\simeq {1\over z - {\vvv\textstyle{\Delta_1^2} \over \textstyle{z \, - \, } {\vvv\textstyle \Delta_2^2\over { \textstyle {\quad\ddots} \atop \qquad\qquad 
{\textstyle{\Delta_{k_{\rm max}-1}^2} \over \vvv\textstyle{z - \Delta_{k_{\rm max}}^2 G^>_{k_{\rm max}}(z;[\alpha])}}} } } }\quad .
\label{VER}
\end{equation} 
The spectral function is given by,
\begin{equation}
\Im ~G^{\rm ver}_{00}(\omega+i0^+)\simeq \Im {1\over \omega - {\vvv\textstyle{\Delta_1^2} \over \textstyle{\omega \, - \, } {\vvv\textstyle \Delta_2^2\over { \textstyle {\quad\ddots} \atop \qquad\qquad 
{\textstyle{\Delta_{k_{\rm max}-1}^2} \over \vvv\textstyle{\omega - \Delta_{k_{\rm max}}^2 \bar{G}^>_{k_{\rm max}}(\omega+i0^+;[\alpha])}}} } } }\quad .
\label{VERpp}
\end{equation} 
Eq.~(\ref{VERpp}) shows that the imaginary part of $G_{00}$ is due to the complex termination function with $\Im ~\bar{G}_{k_{\rm max}}\ne 0$.
The quality of the VER can be tested by two criteria: (i) The  $\chi^2$ in (\ref{eq:LS}) should be much less than unity.   
(ii) Computing some additional recurrents $\Delta_{k> k_{\rm max}}$ and comparing them to the extrapolated $\tilde{\Delta}_k$.

\newpage
\part{Thermodynamic formulas of Hall coefficients}
\label{Part:Thermo II}
For simplicity, the magnetic field is chosen along the $z$ axis and $C_4$ symmetry is assumed in the $XY$ plane. The Hall-type coefficients are the
electric Hall coefficient $R_{\rm H}$, the  Nernst 
Coefficient $N$, and the thermal Hall coefficient $R_{\rm TH}$:
\be
\begin{split}
R_{\rm H} &= \left. {1\over (\sigma^{\rm dc}_{xx})^{2}} {d\sigma^{\rm dc}_{xy} \over dB} \right|_{B=0} \quad ,\nonumber\\
N  &= \left.  {1\over \sigma^{\rm dc}_{xx}} {d\alpha^{\rm dc}_{xy} \over dB} \right|_{B=0}- R_{\rm H} ~\alpha^{\rm dc}_{xx} \quad ,\nonumber\\
R_{\rm TH} &= \left.  {1\over (\kappa^{\rm dc}_{xx})^2  } {d\kappa^{\rm dc}_{xy} \over dB} \right|_{B=0} \quad .
\end{split}
\label{HallCoeff}
\ee
The coefficients in Eq.~(\ref{HallCoeff})  exist if we assume that 
\be
\lim_{B\to 0}\sigma_{xx} \quad , \quad \lim_{B\to 0} \kappa_{xx} >0 \quad .
\ee
That is to say, the expressions apply to gapless dissipative phases, and do not apply to superconductor, insulator, or  quantum Hall phases.  
In addition we also assume no spontaneous magnetization at zero magnetic field, which implies no anomalous Hall effect.

Straightforward Kubo formula calculations of Eq.~(\ref{HallCoeff}) demand  determination of both longitudinal and Hall conductivities. The summation formulas~\cite{EMT,EMTPRL}
derived below, replaces the difficulties of DC Kubo formulas by calculations of thermodynamic coefficients.

 \section{Magnetic Field Expansion of Hall-type Conductivity}
\label{sec:dSdB}

The derivative  with respect to the magnetic field is very difficult to perform using the Lehmann representation, which requires taking derivatives of current operators, wavefunctions, and eigenenergies.

It is therefore  much more convenient to differentiate by parts the DC Hall conductivities as written in Eq.~(\ref{DPP-full1}), 
\bea
\left({dL_{ij}^{xy}(\ve,V)\over dB}\right)_{\bB=0} &=&
{1\over \cV} \Im ~{d\over dB} \Tr ~\rho ~\  \left[P_i^x-(\tilde{P}_i^x)_{\ve},  P_j^y-(\tilde{P}_j^y)_\ve \right]  \nonumber\\
&=& \cancelto{0}{\Xi_{\rm TR}} + \Xi_{\cL} \quad . 
\eea
We assume OBC in order to define uniform polarizations. OBC are also required to continuously vary the magnetic field $B$ and avoid Dirac's quantization~\cite{RXY-PRB}.

$\Xi_{\rm TR}$ vanishes under the trace by even time reversal symmetry of $H_0,\rho, P^x_i,P^y_j$ at $B=0$,  and  odd symmetry of ${d\rho\over dB}, {dP^x_i\over dB},{dP^x_i\over dB}$:
\bea
\Xi_{\rm TR} &=& {1\over \cV} \Im \cancelto{0}{\Tr    {d\rho\over dB}~\left[  P_i^x-\tilde{P}_i^x,P_j^y-\tilde{P}_j^y \right]} + {1\over \cV}\Im \cancelto{0}{\Tr \rho \left[  (1-\Theta_\ve){dP_i^x\over dB} ,(1-\Theta_\ve)  P_j^y  \right]}\nonumber\\
&&  + {1\over \cV}\Im \cancelto{0}{ \Tr \rho 
\left[  (1-\Theta_\ve) P_i^x  ,(1-\Theta_\ve) {dP_j^y\over dB}  \right]}   =0 \quad .
\eea
The remaining term $\Xi_\cL$ is calculated by differentiating the hyper-projector (\ref{HP}) with respect to $B$,
\bea
 {d \Theta_\ve(B)\over dB}   &=&  
{\ve \over \cL^2 + \ve^2} \left( \cM \cL + \cL \cM \right) {\ve \over \cL^2 + \ve^2}  \nonumber\\
&=& (-i) {\ve \over \cL^2 + \ve^2} \cM  {\ve \over \cL^2 + \ve^2} J_i + 
{\ve   \over \cL^2 + \ve^2}\cL \cM  {\ve \over \cL^2 + \ve^2} \quad .
\label{Der2}
\eea
The hypermagnetization is defined by
\be
\cM \equiv  -{\partial \cL\over dB} = [M^z, \bullet] \quad ,
\ee
where $M^z = - {dH_0\over dB}$.
\footnote{Note: Under TR, $M^z\to -M^z$,  and  the commutators in the response functions reverse their order $[A,B] \to  [{\rm TR}(B),{\rm TR}(A)]$. Therefore the hypermagnetization $\cM$ is even under TR. }

Thus we obtain,
\bea
\Xi_{\cL} &=& {1\over \cV}  \Im  \Tr  \rho \left( \left[P_i^x-   (\tilde{P}_i^x)_\ve, -{d\Theta_\ve\over dB}P_j^y,\right]  -
\left[P_j^y-   (\tilde{P}_j^y)_\ve, -{d\Theta_\ve\over dB} P_i^x \right] \right) \nonumber\\
&=&  {1\over \cV}  \Im  \Tr  \rho \left( \left[\left({1\over \cL}\right)_\ve' J_i^x,    {\ve \over \cL^2 + \ve^2} \cM \cL {\ve \over \cL^2 + \ve^2}  P^y_i  + i {\ve \over \cL^2 + \ve^2} \cL \cM  {\ve \over \cL^2 + \ve^2} P^y_i\right]  + \left(J_i^x \leftrightarrow  J_j^y\right) \right)\nonumber\\
&=& - {1\over \cV}  \Im \left( J_i^x \left| \left( {1\over \cL }\right)_\ve'' \cM  \left({1\over \cL }\right)_\ve'' \right| J_j^y\right)+ \cancelto{0}{\Xi_0(J_i^x,J_j^y)} -   \left(J_i^x \leftrightarrow  J_j^y\right) \quad .
\label{XiL}
\eea
In the last row of (\ref{XiL}), we applied Eq.~(\ref{NoLehmann}) to reconstruct  the
OHS current matrix element. The vanishing of $\Xi_0$ is proven as follows. Using the
hermiticity of $\cL$ (\ref{hermiticity}), we can write:
\bea
\Xi_0 &=&  {i\over V}  \Im \left( J_i^x \left|     \left( {1\over \cL }\right)_\ve'' ~\cL  \cM ~ \left({1\over \cL }\right)_\ve''\right| P_j^y \right)    \nonumber\\
&=&  {i\over V}  \Im \left( ~    \cL  \left( {1\over \cL }\right)_\ve''   ~J_i^x   \left| ~  \cM  \left({1\over \cL }\right)_\ve''~\right. P_j^y\right) \quad .
\eea
Next, we will prove that,
 \be
  \left({1\over \cL }\right)_\ve'' \cL  ~\big|J^\alpha_i\Big)  =0+\cO(\ve) \quad .
 \label{xDx}
 \ee
The inner product of the hyperstate given by Eq.~(\ref{xDx}) with any Krylov basis hyperstate $\kbra{ n; J_i^\alpha}$ belonging to the same root current is, 
 \bea
\kexpect{n; J_i^\alpha} {\left({1\over \cL }\right)_\ve'' \cL}  {J^\alpha_i} &=&  \sqrt{\chi_{\rm csr}}~  
 \kexpect{n; J_i^\alpha} {\left({1\over \cL }\right)_\ve'' \cL}  {0;J_i^\alpha} \nonumber\\
 &=&\sqrt{\chi_{\rm csr}}~ \Delta_1~\Im G_{n1}=0 +\cO(\ve) \quad ,
  \eea
since $G_{n1}$ is purely real, as proven in Eq.~(\ref{G1n-useful}).
 
 In contrast to Eq.~(\ref{xDx}), the surviving terms in (\ref{XiL}) do not vanish since, 
 \bea
 \kbra{n; J_i^\alpha} \cL\left({1\over \cL }\right)_\ve'' \ |P^\alpha_i) &=& \kbra{n; J_i^\alpha} \left({1\over \cL }\right)_\ve'' \cL \ |P^\alpha_i) 
 \nonumber\\
 &=&-i \sqrt{\chi_{\rm csr}} \ \kexpect{n; J_i^\alpha} {\left({1\over \cL } J_i^\alpha \right)_\ve''} {1;J_i^\alpha} \nonumber\\
 &=& -i \sqrt{\chi_{\rm csr}} ~\Im G_{n,0} = -i \delta_{n,2k} \chi_{\rm csr} G_{00}'' R_{k}~  \ne 0 \quad .
 \eea
where the last identity uses  Eq.~(\ref{G02k}). 

Thus we obtain
\be
 \left({dL_{ij}^{xy}(\ve,V)\over dB}\right)_{\bB=0} = -   {2\over \cV}  \Im \left( J_i^x \big| \left( {1\over \cL }\right)_\ve'' \cM  \left({1\over \cL }\right)_\ve'' \big| J_j^y\right)  +\cO(\ve) \quad .
\label{dSdB2}
\ee

It is possible to evaluate Eq.~(\ref{dSdB2}) by inserting two Krylov bases resolutions of identity,
 \be
 \sum_{n=0}^\infty \kket{n;J_i^\alpha} \kbra{n; J_i^\alpha} = \mathbb{I}_{i\alpha} \quad , 
 \label{ROI}
 \ee
 where $\mathbb{I}_\alpha$ is the projector onto  OHS subspace spanned by the hyperstates $\cL^k | J_i^\alpha)$ for $  0\le k\le \infty$.
Using the $C_4$ symmetry, we obtain
 \bea
   \left({dL_{ij}^{xy}(\ve,\cV)\over dB}\right)_{\bB=0} &=&  -   {2\chi_{\rm csr}\over \cV}  \Im \sum_{nm} \kexpect{0;J_i^x} {\left( {1\over \cL }\right)_\ve''} {n;J_i^x} ~\cM^{ij}_{nm}~ \kexpect{m; J_j^y} {\left({1\over \cL }\right)_\ve''} {0;J_j^y} \nonumber\\
&=& -   {2\chi_{\rm csr}\over \cV}  \Im \sum_{nm} G_{0,n}'' ~\cM^{ij}_{nm}~ G_{m,0}'' \quad .
\label{dSdBnm}
\eea
where
   $\cM^{ij}_{nm}$ are the $z$-hypermagnetization normalized matrix elements between Krylov hyperstates,
   \be
   \cM^{ij}_{nm} = \Im \ \kexpect{n; J_i^x} {\cM} {m,J_j^y} \quad .
\ee
By Eq.~(\ref{G02k}), 
\be
\Im G^i_{0,n}= \delta_{n,2k} ~ R^i_k ~\Im G^i_{00} ,\quad    R^i_k =\prod_{k'=1}^k \left(-{\Delta^i_{2k'-1}\over \Delta^i_{2k'}}\right) \quad ,
\ee
and by Eq.~(\ref{Lsigma}),
\be
\Re L_{ii}^{xx} =\Re L_{ii}^{yy}= -\chi^i_{\rm csr} ~\Im G^i_{00} \quad .
\ee 
Thus the sum over $n,m$ in Eq.~(\ref{dSdBnm}) includes only even integers and results in the  summation formula,
\be 
 \left({L_{ij}^{xy}(\ve,\cV)\over dB}\right)_{\bB=0}
 = -2 {L_{ii}^{yy}(\ve,\cV) ~L_{jj}^{xx}(\ve,\cV)\over \sqrt{ \chi_{\rm csr}^{i}~\chi_{\rm csr}^{j} }}\sum_{kl} R_k^{(i)}R_{l}^{(j)}  \cM^{ij}_{2k,2l}+\cO(\ve) \quad .
 \label{dSdB3}
 \ee 

The longitudinal conductivities $L_{ii}^{xx}(\ve)$ and $L_{jj}^{yy}(\ve)$  factor out of the sum in Eq.~(\ref{dSdB3}). They produce the non-commuting DC limit of $\ve,1/\cV\to 0$.

Thus,  (unless longitudinal conductivities vanish, or the currents are completely separable as explained in subsection \ref{sec:separable}), 
the longitudinal conductivities may be divided out
and the Hall-type coefficients in Eqs.~(\ref{HallCoeff}) can be expressed as a 
$\ve$-independent thermodynamic coefficients.

\section{Hall Coefficient}
In the double summation over the Krylov bases we separate out the $n=0,m=0$ term  write the electric Hall coefficient as a sum of two terms,
\be
R_{\rm H} =  R_{\rm H}^{(0)} + R_{\rm H}^{\rm corr} \quad , 
\ee
where
\bea
R^{(0)}_{\rm H}& =& {\chi_{\rm cmc}\over \chi_{\rm csr}^2 } \quad , \nonumber\\
\chi_{\rm csr}&=&\lim_{\bq\to 0}\lim_{\cV\to \infty}~{\hbar\over \cV}  (j^x|j^x) \quad ,\nonumber\\
\chi_{\rm cmc} &=& -2 \lim_{\bq\to 0}\lim_{\cV\to \infty}~{\hbar\over \cV}  \Im (j^x|\cM|j^y) 
\quad ,\nonumber\\
R_{\rm H}^{\rm corr}&=& -{2 \over\chi_{\rm csr}} \sum_{kl} R_k R_l  \cM_{2k,2l} \left(1-\delta_{k,0}\delta_{l,0}\right) \quad .
\label{RH}
\eea 
$\chi_{\rm csr}$ is the zeroth moment of the longitudinal conductivity, which was defined in Section \ref{sec:moments}. It represents the kinetic energy of the constituent charge carriers. 
The current-magnetization-current (CMC) susceptibility,  $\chi_{\rm cmc}$,  measures the effect of the Lorentz force on the currents, as shown below. $R_{\rm H}^{(0)}$ reproduces Boltzmann's equation result for energy dependent scattering time~\cite{Abhisek}. 

The correction term $R^{\rm corr}_{\rm H}$ involves higher order recurrents and hypermagnetization matrix elements,  $\cM^{ij}_{nm}$ defined in (\ref{dSdBnm}).
The recurrents and Krylov operators involve  current non-conservation, caused by  disorder, hard core interactions and lattice Umklapp scattering. 
These terms are increasingly difficult to compute.
In order to be allowed to neglect them,  they must be estimated to be smaller in magnitude than $R_{\rm H}^{\rm corr}$.
In Section \ref{sec:tJM} and Section \ref{sec:HCB}, such estimates are obtained by calculation of the lowest order terms for certain lattice models of strongly interacting electrons and hard core bosons. 
 
The current-magnetization-current (CMC) susceptibility,  $\chi_{\rm cmc}$,  measures the effect of the Lorentz force on the currents, as shown below. $R_{\rm H}^{(0)}$ reproduces Boltzmann's equation result for energy dependent scattering time~\cite{Abhisek}. 
$R^{\rm corr}_{\rm H}$  includes the higher order corrections due to disorder, hard core interactions and lattice Umklapp scattering. For Hard Core Bosons, $R^{\rm corr}_{\rm H}$ will be partially evaluated in subsection \ref{sec:Rcorr-HCB}.  
\subsection{Weak scattering limit}
\label{sec:RH-WS}
For non-interacting band electrons with dispersion
$\epsilon_{\bk}$, the
CSR is given by Eq.~(\ref{CSR-NI}), and the CMC is
\be
\chi_{\rm cmc}=  {e^3  \over  c}    \sum_{\bk}   \left( -{\partial f_\bk^0 \over \partial \epsilon} \right) \left( v_\bk^{y} \left( v_\bk^{y} {\partial \over \partial k^{x}} - v_\bk^{x} {\partial \over \partial k^{y}}\right)  v_\bk^{x}\right) \quad .
\label{CMC-NI}
\ee
Thus, $R_{\rm H}^{(0)}= \chi_{\rm cmc}/\chi_{\rm csr}^2$ recovers Boltzmann equation result (\ref{RH-Boltz}) for isotropic lifetime.  $R^{\rm corr}\propto \Delta_1$, which depends on $[H,j^x]\ne 0$. Therefore the correction is relatively suppressed at low disorder.

\section{Modified Nernst Coefficient}
The Modified Nernst Coefficient $\bar{N}$ for $C_4$ symmetric Hamiltonians
is
\be
\bar{N} =   {1\over \sigma^{\rm dc}_{xx}\kappa^{\rm dc}_{xx}}{d\alpha_{xy}\over dB} = {1\over \kappa_{xx}^{\rm dc}}\left( N+R_{\rm H} \alpha_{xx}\right) \quad .
\ee
$\bar{N}$ can be expressed by the summation formula,
\be 
\bar{N} = \bar{N}^{(0)} + \bar{N}^{\rm corr} \quad ,
\ee
where 
\bea
\bar{N}^{(0)} &=&  {\chi^{th-el}_{\rm cmc}\over \chi_{\rm csr} \chi_{csr}^{\rm th} } \quad ,\nonumber\\
\chi^{\rm th}_{\rm csr}&=&\lim_{V\to \infty}~{\hbar\over \cV}  (j_Q^x|j_Q^x) \quad ,\nonumber\\
\nonumber\\
\chi^{\rm th-el}_{\rm cmc} &=& - {2 \over V} ~\Im (j_Q^x|\cM|j^y) \quad ,\nonumber\\
\bar{N}_{\rm H}^{\rm corr}&=&  -{2\over\sqrt{ \chi^{\rm th}_{\rm csr}\chi_{\rm csr} }} \sum_{kl} R^{\rm th}_k R_l    \cM^{\rm th-el}_{2k,2l} \left(1-\delta_{k,0}\delta_{l,0}\right) \quad ,\nonumber\\
\cM^{\rm th-el}_{nm} &=& \Im \ \kexpect{n;j_Q^x} {\cM} {m;j^y} \quad .
\label{barN}
\eea 
$\chi_{\rm csr}^{\rm th}$ and $\cM^{\rm th-el}_{nm}$ are
the thermal CSR,  
and  thermoelectric hypermagnetization matrix elements respectively.
$\bar{N}$ is therefore expressed as a thermodynamic coefficient similar to $R_{\rm H}$, and $R_{\rm TH}$ which follows below.

\section{Thermal Hall coefficient}
The thermal Hall coefficient is, 
\be 
R_{\rm TH} =R^{(0)}_{\rm TH}  + R_{\rm TH}^{\rm corr} \quad ,
\ee
where 
\bea
R^{(0)}_{\rm TH}& =& T { \chi^{th}_{\rm cmc}\over (\chi^{\rm th}_{\rm csr})^2 } \nonumber\\
 R_{\rm TH}^{\rm corr}&=&  -{2 T \over \chi^{\rm th}_{\rm csr} }\sum_{kl} R^{\rm th}_k R^{\rm th}_l   \cM^{\rm th-th}_{2k,2l} \left(1-\delta_{k,0}\delta_{l,0}\right)\nonumber\\
\cM^{\rm th-th}_{nm} &=& \Im \ \kexpect{n;j_Q^x} {\cM} {m;j_Q^y} \quad .
\label{RTH}
\eea 

\section{Calculating the correction terms}
\label{sec:Rcorr}
The correction  terms  in Eqs.~(\ref{RH},\ref{barN},\ref{RTH}) depend on
recurrents $\Delta_n, n=1,2,\ldots$, which demand calculations of moments $\chi_{\rm csr}, \mu_2, \mu_4 ,\ldots $, and  hypermagnetization matrix elements $\cM_{2k,2l}$.
The latter are most conveniently derived from the {\em non-normalized} magnetization matrix elements
\be
\tilde{\cM}_{nm}\equiv \Im \Big(\cL^n J_i^y \big| \cM \big| \cL^m J_j^x\Big)=
- \Im \Tr \rho \left[\cL^{n-1}J_i^y,  \cM \cL^{m}J_j^x\right] \quad ,
\label{HyperM}
\ee
which are thermodynamic expectation values of nested commutators. For short range Hamiltonians, the commutators include sums over connected clusters of operators which are easier to trace over than   calculating two-operator susceptibilities.
These clusters can be generated and traced over by symbolic manipulation, as demonstrated  in Sections \ref{sec:tJM} and \ref{sec:HCB}.

$\tilde{\cM}_{nm}$ are related to $ {\cM}_{nm}$ of the correction terms by the Gramm-Schmidt matrix $K$. This matrix is determined as follows.  Applying the resolution of identity with the Krylov basis,  
\bea
\cL^{k} \kket{0} &=&   
\sum_{k'\le k} \kket{k'}\kbra{k'} \cL^{k} \kket{0} \nonumber\\
&\equiv&  \sum_{k'} K_{k, k'} \kket{k'} \quad .
\eea
The matrix $K$ is obtained by powers of the tridiagonal Liouvillian matrix (\ref{Lmn}), which results in  polynomials of  the recurrents $\Delta_{k'\le k}$,
\be
 K_{k',k}[\Delta_1,\ldots\Delta_k] =  (L^k)_{k',0} ,\quad k'\le k \quad .
\ee
Thus we obtain,
\be
\cM_{n,m}  = \sum_{n'\le n}\sum_{m'\le m}~ K^{-1}_{n,n'} [\Delta]~   K^{-1}_{m,m'}[\Delta] ~~ \tilde{\cM}''_{n',m'} \quad .
\ee

\section{Optimization of thermodynamic approaches}

\subsection{The Separability problem}
\label{sec:separable}
The thermodynamic approaches have their limitations.
The continued fraction extrapolation can only be performed efficiently if a clear asymptotic power law can be deduced from the dependence of the calculated  recurrents on their order.  In practice, the fluctuations of the  calculated recurrents from an asymptotic power law behavior should decrease monotonically if a reliable termination function is to be found.

Nearly separable currents can significantly slow the  convergence of continued fractions and Hall coefficients corrections.
For example, the recurrents of  a sum of two spectral functions with very different characteristic frequency scales is shown in Fig.~\ref{fig:SCpm}. We note that the fluctuations about average recurrents do not converge rapidly.  

Physically, sums of spectral functions with vastly different frequency scales can be expected  for separable currents with different relaxation times.
That is to say, if the Hamiltonian and its currents 
can be written as sum of commuting channels,
\be
H= \sum_{r=1 }^{N_r}  H_r \quad ,\quad \bj= \sum_{r }  \bj_r \quad ,\quad [\bj_r,H_{r'}]\propto \delta_{rr'} \quad ,
\ee
the resulting conductivities will be sums,
\be
\sigma_{\alpha\beta}(\omega) =\sum_r \sigma_{\alpha\beta}^r(\omega) \quad . 
\ee
While their moments are additive,
\be
\mu_{2k} = \sum_r \mu^r_{2k} \quad ,
\ee
the recurrents  which are obtained by Eqs.~(\ref{Delta-mu}) are highly entangled functions of the separate $\mu_{2k}^r$'s.

Similarly, the Hall coefficient  summation formulas are  formally correct, but cease to be useful for separable currents. Since the total Hall coefficient is,
\be
R_{\rm H} = {\sum_{r=1 }^{N_r} (d\sigma^r_{xy}/dB) \over \left( \sum_{r=1}^{N_r } \sigma^r_{xx} \right)^2 } \quad ,
\label{RH-tot}\ee
for multichannel separable currents $N_r>1$,  ratios between different conductivities must enter the  Hall coefficient, which therefore cannot be captured accurately  just by $R_{\rm H}^{(0)}$.
 For example, Fermi surface electrons with strongly dependent $\bk$-dependent lifetime $\tau_{\bk}$.
 At weak scattering, the currents of different wavevectors are separable.  While
$R_{\rm H}^{(0)}$ of Eq.~(\ref{RH}) recovers the single lifetime expression of Eq.~\ref{Sab-tauconst}, it does not agree (within a factor of order unity) with the $\bk$-dependent lifetime result of Boltzmann's equation in Eq.~\ref{Sab-Boltz}. The difference between $R_{\rm H}^{(0)}$ and Boltzmann's result is contained in the  higher order correction terms, which are unwieldy.

In conclusion, continued fractions and Hall coefficient summation formulas are best suited for a non-separable current governed by one relaxation timescale.  Otherwise,
in Hamiltonians which support two or more separable currents with different relaxation tiemscales, the conductivities of each current should be calculated separately  and later
assembled in Eq.~(\ref{RH-tot}) to obtain the  total Hall coefficient.\footnote{We thank Steve Kivelson for emphasizing the anistropic lifetime problem in the Hall coefficient formula.}

\subsection{Renormalized Hamiltonians at low temperatures}
\label{sec:Renorm}
Effective Hamiltonians are crucially important for efficient use of thermodynamic approaches to DC transport coefficients.
As can be seen from the Lehmann representation of longitudinal conductivities~(\ref{Kubo-Lehmann}),  and the DPP formulas of the Hall conductivities~(\ref{DPP-full}), the important part of the spectrum and eigenstates resides in an energy window of order $T$ above the ground state.
The Hall coefficient formula, which is derived from 
Eq.~(\ref{dSdB2}), also depends on that part of the spectrum as can be seen from,
\be
\left({dL^{xy}_{ij}(\ve,V)\over dB}\right)= -\pi^2 
    {2\over \cV}  \Im \sum_{nmn'} W_{nm}  \langle n |J_i^x|m\rangle\langle n' |J_j^y |n\rangle \cM_{nm'}^{n'n}  \delta_\ve(E_n-E_{m'}) \delta_\ve(E_{n'}-E_{n}) \quad .
\ee
where $\delta_\ve(x)=\ve/(x^2+\ve^2)$. Thus, in the DC limit, $dL_{ij}^{xy}/dB$ is also an {\em on-shell} expression.
It is greatly advantageous to replace the microscopic Hamiltonian $H^{\rm micro}$  by an effective
Hamiltonian $H^{\rm eff}(B)$ which operates
in the low energy Hilbert space  and  reproduces the correct spectrum at $\{ E_n \le {\rm const}~T\}$. 

The renormalized currents and magnetization within this reduced Hilbert space  are defined by,
\bea
\bj_i^{\rm eff} &=&  i[H^{\rm eff},\bP^{\rm eff}_i] \quad ,\nonumber\\
M^{\rm eff} &=&  -{\partial H^{\rm eff}\over dB} \quad ,
\eea
where $\bP^{\rm eff}_i$ are also projected onto the reduced  Hilbert space. 
A consequence of renormalization in certain models is to reduce the magnitude of  currents' non-conservation, i.e.
\be
||[H^{\rm eff},\bj_i^{\rm eff}]||\le ||[H^{\rm micro} ,\bj_i]|| \quad .
\ee
The relative magnitude of the first recurrent depends on this commutator. Since all terms in the
Hall coefficient's correction terms $R_{\rm corr}$
are proportional to $\Delta_1$, this ensures reduction of the relative magnitude $R_{\rm corr}/R_{\rm H}^{(0)}$.

An simple example of the value of renormalization is demonstrated by comparing the zeroth Hall coefficients  of the microscopic Hamiltonian of electrons in a periodic potential, to that of the conduction band effective Hamiltonian.
The microscopic Hamiltonian is $H^{\rm particles}$
goven in Eq.~(\ref{Hparticles}).
Its magnetization is given by
\be
M^z ={1\over 2c}  \sum_{i=1}^{N_p}  \bx_i \times {\bp_i\over m} \quad .
\ee
The CSR and CMC are completely independent on the potential terms and hence,    
\bea
\chi^{\rm particles}_{\rm csr}&=& {N_p e^2 \over \cV m }\nonumber\\
\chi^{\rm particles}_{\rm cmc}&=& {N_{\rm tot} e^3 \over \cV m^2 c}\nonumber\\
R_{\rm H}^{(0),^{\rm particles}} &=&  {\cV \over N_p  ec  } \quad .
\eea
Where $N_p$ includes the electrons in  {\em all} the filled bands  and core states. The zeroth Hall coefficient must clearly be very far from the full answer, since we know from Boltzmann's equation (\ref{RH-Boltz}) that $R_{\rm H}$ depends only on the filling of the {\em conduction} band. The correction term is therefore expected to be significant. Indeed, the magnitude of the first recurrent , which is a factor in $R^{\rm corr}$ is of order
\be
\Delta_1^2  = {1\over \cV \chi_{\rm csr}} \sum_{\bk,\bq,n,n'} {f_{\bk,n}-f_{\bk+\bq,n'} \over  \epsilon_{\bk+\bq,n}-\epsilon_{\bk,n'} }   (q^x)^2 |V(\bq)|^2 \quad ,
\ee
which can be very large for crystal potentials $|V(\bq)|\gg |\epsilon_F|$.

In comparison, if we use the  renormalized conduction band Hamiltonian,
\bea
H^{\rm eff} &=&  \sum_{\bk} {\bk^2\over 2m^*}c^\dagger_\bk c^{}_\bk + {1\over \cV} \sum_{\bk,\bk'} U^{\rm dis}_{\bk\bk'}c^\dagger_\bk c^{}_{\bk'} \quad , \nonumber\\
\bj^{\rm eff} &=&  \sum_{\bk}  {\bk\over m^*}  c^\dagger_\bk c^{}_\bk \quad .
\label{Hj}
\eea
where $U^{\rm dis}\ll \epsilon_F$ is due to impurities, then
\bea
\chi^{\rm eff}_{\rm csr}&=& {(e^*)^2  ~n_{\rm cond} \over m^*} \quad , \nonumber\\
\chi^{\rm eff}_{\rm cmc}&=& {(e^*)^2 ~ n_{\rm cond} \over (m^*)^2 ~c} \quad ,\nonumber\\\nonumber\\
\left(R_{\rm H}^{(0)}\right)^{\rm eff} &=&  {1\over n_{\rm cond} ~e^*~c} \quad .
\eea
$n_{\rm cond}$ is just the density of the partially filled conduction band. $R_{\rm H}^{(0)}$ 
recovers Drude-Boltzmann theory (\ref{RH-Drude}), with a small correction term which is suppressed by a factor of $U^{\rm dis}/\epsilon_F\ll 1 $.

For dirty semimetals~\cite{Abhisek}, where the interband gap is of the same order as $U^{\rm dis}$,  the Hall coefficient can be well approximated by $R_{\rm H}^{(0)}$  of the  two-band Hamiltonian, including  both intraband and interband current matrix elements. The Hall coefficient formula provides a   simpler route to the Hall coefficient than  the coupled two-band Boltzmann equation~\cite{Bulbul,Culcer,Tserkov}.

\part{Numerical algorithms for thermodynamic coefficients}
\label{Part:Thermo-numerics}
In Parts \ref{Part:Thermo I} and \ref{Part:Thermo II}, longitudinal conductivities and Hall-type coefficients were expressed in terms of thermodynamic expectation values, for which a host of  statistical mechanics algorithms exist. Below, as a preparation for the calculations of high temperature conductivities of strongly correlated lattice models in Parts \ref{part:SCE} and \ref{part:SCB}, we elaborate on the application of high temperature series and stochastic (Monte-Carlo) methods to evaluate the following relevant thermodynamic coefficients:
\begin{enumerate}
    \item The CSR, as defined by Eq.~\ref{CSR-def}.
       \item The conductivity moments $\mu_{2k},k=1,2,\ldots$ as defined by Eq.~(\ref{EV}). The recurents can be obtained from themoments by Eq.~(\ref{Delta-mu}).
    \item  CMC susceptibilities, which can be written as expectation values,
    \be
    \chi_{\rm cmc} =
2 \lim_{\cV\to \infty}~{\hbar\over \cV}  \Re \Tr~\left(\rho \left[P^x, ,\left[M,J^y\right]\right]\right) 
\quad ,
\ee
\item The hypermagetization matrix elements, of Eq.~\ref{HyperM}.
\end{enumerate}
At low temperatures, the use of variational wavefunctions to compute thermodynamic expectation values has seen enormous progress in recent years.. As examples, we refer the reader to the density matrix renormalization group~\cite{DMRG,DMRG-Schollwok}, projected entangled-pair states~\cite{PEPS} and  tensor networks~\cite{TN}.
For lack of space, these algorithms will not be reviewed here.

\section{High temperature series}
\label{sec:HighT}
For interacting Hamiltonians, an advantage of thermodynamic expectation values over dynamic (time dependent) correlations,  is its amenability to high temperature series expansion. Each term in the series involves traces over powers of the Hamiltonian, which are  much less costly than Hamiltonian diagonalization over large lattice sizes.  

The expansion works as follows. The density matrix $\rho=\exp(-\beta H)/Z $ in powers of $\beta$, 
\bea
\exp(-\beta H)&=& \sum_{n=0}^\infty { ( -\beta )^n  \over n!}~H^n\quad,\nonumber\\
Z(\beta) &=& \sum_{n=0}^\infty {(-\beta)^n\over n!}  ~ \Tr H^n\quad.
\eea
For simplicity we restrict the following discussion to Hamiltonians which are sums of traceless bond operators on a bipartite lattice. Such Hamiltonians 
are used 
to describe the strongly correlated electrons and bosons models in Parts \ref{part:SCE} and \ref{part:SCB}
respectively. Hence, traces over even powers of $H$ are assumed to vanish. 

Expansion of the expectation value of a local operator $O$ up to fourth order in $\beta$ yields,
\bea
\langle O \rangle&=& \Tr (O) -\beta \Tr (HO)  + {\beta^2 \over 2}  \left(\Tr (H^2 O)-\Tr(O)\times \Tr(H^2) \right)\nonumber\\
&&~~~- {\beta^3\over 3!} \Big(\Tr (H^3 O)  -3\Tr  (H O) \times\Tr (H^2) \Big)  \nonumber\\
&&~~~ + {\beta^4\over 4!} \Big(  \Tr (H^4 O) - 6 \Tr (H^2 O) \times  \Tr (H^2) -4 \Tr (H O) \times  \Tr (H^3)-\Tr(O)\times\Tr(H^4) \Big) +\cO(\beta^5) \quad .
\label{HighT}
\eea

The thermodynamic coefficients are obtain by tracing over a lattice  large enough to  include  {\em at least} all the connected operators in $H^4$ which share  sites with  $O$.

\subsection{Linked clusters theorem}
Let's consider a local  operator $O$ with $\Tr O \ne 0$, which occupies a cluster of sites $c_O=\{\br_1,\ldots \br_{n_O}\}$.
The strings of site operators  which have sites in $c_O$  define a connected cluster to $c_O$.
The trace  $\Tr H^n O$  is a sum of strings of operators
which can be factorized as a trace over sites in the cluster $c_O^{k}$ which are connected by $k$ powers of $H$ to $c_0$,  and a trace over the rest of the powers of $H$ as follows:
\be
\Tr (H^n O) = \sum_{k=0}^n {n!\over k! (n-k)!} \Tr_{ \in c^{(k)}_O}(H^k O) \times \Tr_{\notin c^{(k)}_0 } (H^{n-k}) \quad .
\label{cluster}
\ee

\begin{figure}[h]
\begin{center}
 \includegraphics[width=8cm,angle=0]{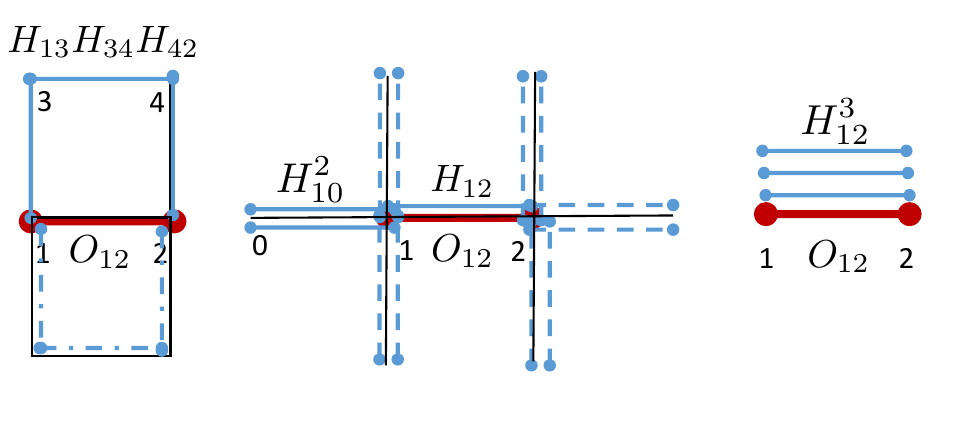}
\caption{High temperature expansion of a nearest neighbor Hamiltonian on the square lattice. The graphs describe the three types of connected operator products in $\Tr (H^3 O_{12})_{\in c^{(3)}_O}$ of Eq.~(\ref{HighT-cc}). The linked cluster  $c_O^{(3)}$ is defined by the  lattice of 7 bonds which is shown in  the central graph, on which the counter-term $-3\Tr_{ \in c^{(3)}_O}  (H O) \times\Tr_{ \in c^{(3)}_O} (H^2)$ is also defined. }
 \label{fig:HighT-3}
\end{center}
\end{figure}

 A cancellation of  all the non-connected clusters $\notin c_0$  in Eq.~(\ref{HighT}) occurs between the numerator and the expansion of the partition function at each order in $\beta^n$.
 This leaves us  to compute only traces over strings of operators in $(H^{k} O)$ which include sites in $c^{(k)}_O$ (linked clusters of $H^{k}O$), and traces of $H^{n-k}$   for $k=0,1\ldots n$ (a.k.a linked cluster theorem). Eq.~(\ref{HighT}) then reads as
 \bea
\langle O \rangle&=& \Tr (O)  -\beta \Tr_{ \in c_O} (HO)  + {\beta^2 \over 2}  \left(\Tr_{ \in c^{(2)}_O} (H^2 O)-\Tr_{ \in c^{(2)}_O}\times \Tr_{ \in c^{(2)}_O}(H^2) \right)\nonumber\\
&&~~~- {\beta^3\over 3!} \Big( \Tr_{ \in c^{(3)}_O} (H^3 O)  -3\Tr_{ \in c^{(3)}_O}  (H O) \times\Tr_{ \in c^{(3)}_O} (H^2) \Big)  \nonumber\\
&&~~~ + {\beta^4\over 4!} \Big(  \Tr_{ \in c^{(4)}_O} (H^4 O) - 6 \Tr_{ \in c^{(4)}_O} (H^2 O) \times  \Tr_{ \in c^{(4)}_O} (H^2) -\Tr(O)\times\Tr_{ \in c^{(4)}_O}(H^4) \Big) +\cO(\beta^5) \quad .
\label{HighT-cc}
\eea

An example of a calculation of $\beta^3 \Big( \Tr_{ \in c^{(3)}_O} (H^3 O) \Big) $ is described in Fig.~\ref{fig:HighT-3}. Here, $H$ is taken as a sum over traceless bond operators $H_{ij}$ on the square lattice, while $O_{12}$ is any traceless bond operator.

\section{Stochastic algorithms}
\label{QMC-Hall}
Stochastic sampling of thermodynamic expectation values a.k.a. Quantum Monte Carlo (QMC), have proven to be very efficient in extending high temperature series to lower temperatures.
The algorithms described below  are particularly
tailored for interacting  quantum particles
on lattices. These are the Determinant QMC (DQMC) for fermions and the Directed Loop Algorithm (DLA) for lattice bosons and spins. 

\subsection{Determinant Quantum Monte Carlo}
\label{sec:DQMC}
In DQMC, the goal is to rewrite the grand-canonical partition function of a given fermionic Hamiltonian as a sum of “statistical weights" over a space of field configurations~\cite{blankenbecler}. 
To illustrate the method~\cite{scalettarnotes}, we consider the single band Hubbard model (HM),
\be
H^{\rm HM} = -t\sum_{\ij\sigma} (c^{\dagger}_{i\sigma}c_{j\sigma} + h.c.) - \mu\sum_{i\sigma}n_{i\sigma} + U\sum_{i}\left(n_{i\uparrow} - {1 \over 2}\right)\left(n_{j\downarrow} - {1 \over 2}\right) \quad .
\label{HM}
\ee
where $ c^\dagger_{i\sigma}$ creates a fermion of spin $\sigma=\uparrow,\downarrow$ on lattice site $i$, and $n_{i\sigma}=c^\dagger_{i\sigma}c^{}_{i\sigma}$. $\mu$ is the chemical potential.

One now divides $\beta$ into $L_{\tau}$ `Trotter' steps, such that $\beta=\Delta \tau L_{\tau}$. The partition function is rewritten as-
\be
Z = Tr(e^{-\beta H}) = \Tr(e^{-\Delta\tau H})^{L_\tau} \approx \Tr\underbrace{(e^{-\Delta\tau K}e^{-\Delta\tau V}e^{-\Delta\tau K} e^{-\Delta\tau V} \cdots)}_{L_\tau} \quad , 
\ee
where $K=-\sum_{<ij>\sigma} t_{ij}(c^{\dagger}_{i\sigma}c_{j\sigma} + h.c.) - \mu\sum_{i\sigma}n_{i\sigma}$ is the kinetic term and $V= U\sum_{i}\left(n_{i} - {1 \over 2}\right)\left(n_{j} - {1 \over 2}\right)$ is the potential term. The second equality is approximate as $K$ and $V$ don't commute. The approximation improves upon increasing $L_{\tau}$ (decreasing $\Delta\tau$).
To handle the interaction $V$ term, we employ the discrete  Hubbard-Stratonovich (HS) transformation:
\be
e^{-U\Delta \tau \left(n_{i} - {1 \over 2}\right)\left(n_{j} - {1 \over 2}\right)} = {1 \over 2} e^{-{U\Delta\tau \over 4}}\sum_{s}e^{\lambda s(n_{\uparrow} - n_{\downarrow})} \quad ,
\ee
where $\cosh\lambda = e^{U\Delta\tau/2}$ and $s=\pm 1$ is an Ising variable.

The $e^{-\Delta\tau K}$ terms are bilinear in the fermion operators. For each factor of the $L_{\tau}$ terms $e^{-\Delta\tau V}$, we introduce N number of HS fields, one for each spatial site, where we have an on–site interaction to decouple. The HS field $s(i,l)$ therefore has two indices, space $i$ and imaginary time $l$. Now, the $e^{-\Delta\tau V(l)}$ also become bilinear in the fermion operators. We put an argument $l$ on the $V$-terms to emphasize that while the $K$-terms are all identical, the $V(l)$ contain different HS fields on the different imaginary time slices.

Finally, the trace evaluates to 
\be
Z = \sum_{s(il)} \det M_{\uparrow} \det M_{\downarrow} \quad , 
\ee
where $M_{\sigma} = I + e^{-k}e^{-v_{\sigma}(1)}e^{-k}e^{-v_{\sigma}(2)} \cdots e^{-k}e^{-v_{\sigma}(L_{\tau})}$. Here, the matrix elements of $k$ are equal to $\Delta\tau k_{ij}$. $v_{\sigma}(l)$ equals the diagonal matrix, whose elements are $\lambda s(i,l)$, obeying $v_{\uparrow}(l) = -v_{\downarrow}(l)$. $I$ is the $N \times N$ identity matrix.

Thus, we get a determinant for each of the two spin species. The quantum partition function has now been re-expressed as a classical Monte Carlo problem. One needs to sum over the possible configurations of the real, classical variables $s(i,l)$ with the ``Boltzmann weight'', which is the product of the two fermion determinants. 
The workhorse of the algorithm is the Green's function $G_{\sigma}$, which is first initialized as-
\be
G_{\sigma} =  [I + e^{-k}e^{-v_{\sigma}(1)}e^{-k}e^{-v_{\sigma}(2)} \cdots e^{-k}e^{-v_{\sigma}(L_{\tau})}]^{-1} \quad .
\ee
Next, the HS fields are updated one by one, and the new $G_{\sigma}$'s are computed. Finally, after completing an entire set of updates on all space-time points of the lattice, measurements are made from these Green's functions. A couple of simple examples are-
\bea
\langle n_{i\sigma} \rangle &=& 1 - [G_{\sigma}]_{ii} \\ \nonumber
\langle S^{+}_{i}S^{-}_{j} \rangle &=& -[G_{\uparrow}]_{ji}[G_{\downarrow}]_{ij} \quad ,
\eea
where $\langle n_{i\sigma} \rangle$ denotes the local spin density at site $i$ and $\langle S^{+}_{i}S^{-}_{j} \rangle$ expresses (equal-time) spin correlations between sites $i$ and $j$.
The algorithm scales in CPU time as $N^4L$. The reason is that re–evaluating the determinant of $M$ takes $N^3$ operations, and we must do that $NL$ times to sweep through all the HS variables.
There are tricks to reduce this computational cost to $N^{3}L$, relying on the sparse nature of the “difference matrix" $dM$, after updating the HS fields.

Some salient features of DQMC simulations are as follows-
\begin{itemize}
    \item Measurement of correlation functions with non–zero imaginary time separation is possible, but requires more work. Analytic continuation of such correlations is required to get the dynamical response. That is significantly more difficult and often requires uncontrolled approximations~\cite{shao}.
    \item The product of matrices required in constructing $M$ and hence $G=M^{-1}$ is numerically unstable at low temperatures and strong couplings. Special ``stabilization'' is required to do the matrix manipulations~\cite{bauer}.
    \item The determinants of the matrices can turn negative. This is the infamous ``fermion sign problem''~\cite{mondaini1,mondaini2,iskakov}. The sign problem does not occur for certain special cases. For example, if $U$ is negative (the “attractive” Hubbard model), the individual determinants can go negative, but the matrices are always equal and hence the determinant appears as a perfect square. If $U>0$, but the chemical potential $\mu=U/2$ (``half–filling''), it's also fine. The matrices are not identical in this case, but the determinants are nevertheless related by a positive factor, i.e., they again have the same sign, so their product is always positive.
    \item Alternate HS transformations are possible~\cite{karakazu}. One may couple more symmetrically to the spin, which does not single out the $z$-component. Alternatively, one can couple to pair creation operators. So far, all such alternatives give a worse sign problem than the one mentioned here.
\end{itemize}

A determinant QMC calculation for expectation values of lattice fermions with discrete auxiliary fields can be implemented using the ALF package~\cite{ALFQMC}. 
The statistical fluctuations can be evaluated by  {\em``Jackknife
resampling''}~\cite{efron} (a method used for error estimation).
The  average fermion-induced sign in the QMC sampling is defined as
\be
\langle S\rangle=\langle {\rm sgn}({\rm det}) \rangle \quad .
\label{QMC-sign}
\ee
Empirically, the value of $\langle S\rangle$ is found to low temperature, reflecting the diminishing effects of exchange processes in fermion world lines. It has been argued to be safe to use the expectation value of $S$ to normalize the QMC averages, provided $\langle S\rangle>0.8$. 

The CMC and CSR susceptibilities are computed by first expressing the operators in terms of fermionic creation and annihilation operators.  Since the Boltzmann eights are given by imaginary-time action of free fermions with time dependent fields, the operator expectation values
can be factorized as products of Green's functions (using Wick's theorem), which can be sampled
over configurations of the auxiliary fields.
The displayed data in Section~\ref{sec:tJM} is restricted to the regime of $\langle S\rangle \ge 0.8$, which for  $U=8t$ and all doping range is satisfied at $T \ge  t/2\approx J$.  

\subsection{Directed Loop Algorithm}
\label{sec:QMC-DLA}
The Stochastic Series Expansion~\cite{SSE} (SSE) and QMC world-line and worm algorithms~\cite{prokofiev,Worm}  have been found to be very efficient for  lattice bosons and spin models~\cite{sandvik1,boninsegni}.

A particular variant of these algorithms is the Directed Loop Algorithm (DLA)~\cite{DLA}, which may be formulated along both world-line and SSE approaches. For simplicity, we review the DLA scheme within the SSE method, where the bonds of the Hamiltonian are decomposed as 
\be
H = -\sum_{a}\sum_{b}H_{a,b} , 
\ee
where $H_{a,b}$ satisfies $H_{a,b} |\alpha \rangle = |\alpha^{\prime} \rangle$ for a chosen basis $\{ |\alpha \rangle \}$. The index $a$ refers to operator type (either  diagonal or off-diagonal  in the chosen sites basis), and index $b$ labels a particular bond position on the lattice. $H_{0,0}$ denotes the unit operator. Using a Taylor expansion, the partition function is expressed as,
\be
Z = \sum_{\alpha} \sum_{S_{M}} {\beta^n (M-n)! \over M!} \left \langle \alpha \left \vert \prod^{M}_{p=1} H_{a_{p},b_{p}} \right \vert \alpha \right \rangle \quad ,
\label{SSE}
\ee
where $M$ is the truncation order of the expansion, $S_{M} = [[a_{1},b_{1}],[a_{2},b_{2}],\cdots,[a_{M},b_{M}]]$ denotes the operator product and $n$ is the number of non-$[0,0]$ elements.
$M$ is adjusted during the equilibration. Typically $M\sim\beta N$ suffices for most purposes, where $N$ denotes the number of lattice sites. Next, one defines a state $|\alpha(p)\rangle$ by   
\bea
|\alpha(p)) \equiv \prod^{p}_{i=1} H_{a_{i},b_{i}}|\alpha(0)\rangle\quad, \nonumber\\
|\alpha(p) \rangle = {|\alpha(p))/(\alpha(p)|\alpha(p))^{1\over 2}}\quad.
\eea
The periodicity constraint $|\alpha(M)\rangle = |\alpha(0)\rangle$ needs to be satisfied for a non-vanishing contribution to Z. The sum is sampled by transitions $(\alpha, S_{M}) \rightarrow (\alpha^{\prime}, S^{\prime}_{M})$ satisfying detailed balance.

To concretize the procedure, we consider the anisotropic Heisenberg model in a magnetic field, whose Hamiltonian is- 
\be
H = \sum_{<ij>} J [S^{x}_{i}S^{x}_{j} + S^{y}_{i}S^{y}_{j} + \Delta S^{z}_{i}S^{z}_{j}] - h\sum_{i} S^z_{i} \quad .
\ee
We'll use the standard $z$-component basis ($|\alpha \rangle = |S^{z}_{1}S^{z}_{2}\cdots S^{z}_{N} \rangle$). The diagonal and off-diagonal bond operators are defined as-
\bea
H_{1,b} &=& \epsilon + \Delta/4 - \Delta S^{z}_{i(b)}S^{z}_{j(b)} + h_{b}[S^{z}_{i(b)} + S^{z}_{j(b)}] \\ \nonumber
H_{2,b} &=& -{1 \over 2} [S^{+}_{i(b)}S^{-}_{j(b)} + S^{-}_{i(b)}S^{+}_{j(b)}] \quad , 
\eea
where $i(b)$, $j(b)$ are the sites connected by bond $b$ and $h_{b}$ is the bond field ($h_{b} = h/2d$ on a d-dimensional cubic lattice). The constant $\epsilon+\Delta/4+h_{b}$ is added to the diagonal bond operator to ensure that all its matrix elements are positive. Moreover, on a bipartite lattice, the minus sign in front of the off-diagonal operator is irrelevant, as it can be ``gauged out''. Storing the operator sequence $S_{M}$ and a single state $|\alpha(p)\rangle$, diagonal updates of the form $[0,0] \leftrightarrow [1,b]$ can be carried out sequentially for $p=1,\cdots,M$ at all elements $[a_{p},b_{p}]$ in $S_{M}$ with $a_{p}=0,1$. The Metropolis acceptance probabilities for these processes may be derived as-
\bea
P([0,0] \rightarrow [1,b]) &=& N_{b}\beta \langle S^{z}_{i}(p)S^{z}_{j}(p) | H_{1,b} | S^{z}_{i}(p)S^{z}_{j}(p) \rangle / (M-n) \\ \nonumber 
P([1,b] \rightarrow [0,0]) &=& (M-n+1) / [N_{b}\beta \langle S^{z}_{i}(p) S^{z}_{j}(p) | H_{1,b} | S^{z}_{i}(p)S^{z}_{j}(p) \rangle] \quad ,
\eea
where $P>1$ is interpreted as probability one. 

We next discuss the ``operator loop'' updates, which are efficient off-diagonal updates. To formulate these, it is useful to introduce a different representation of the SSE configurations. It is not necessary to store the full states $|\alpha(p)\rangle$ as they contain a great deal of redundant information. One can represent the matrix element in Eq.~(\ref{SSE}) as a linked list of ``vertices''. First, we note that the weight of a configuration $(\alpha,S_{M})$ can be written as-
\be
W(\alpha,S_{M}) = {\beta^n (M-n)! \over M!} \prod^{n}_{p=1} W(p) \quad ,
\label{SSE-VW}
\ee
where the product is over the n non-$[0,0]$ operators in $S_{M}$. The $W(p)$ are called ``bare vertex weights'' and are given by-
\be
W(p) = \langle S^{z}_{i(b_{p})}(p)S^{z}_{j(b_{p})}(p) | H_{b_{p}} | S^{z}_{i(b_{p})}(p-1)S^{z}_{j(b_{p})}(p-1) \rangle \quad , 
\ee
where $H_{b} = H_{1,b} + H_{2,b}$ is the full bond operator. 
A vertex represents the spins on bond $b_{p}$ before and after the operator has acted. These four spins constitute the legs of the vertex. There are six allowed vertices, with four different vertex weights- 
\bea
W_{1} &=& \langle \downarrow\downarrow | H_{b} | \downarrow\downarrow \rangle = \epsilon \quad , \\ \nonumber 
W_{2} &=& \langle \downarrow\uparrow | H_{b} | \downarrow\uparrow \rangle = W_{3} = \langle \uparrow\downarrow | H_{b} | \uparrow\downarrow \rangle = \Delta/2 + h_{b} + \epsilon \quad , \\ 
\nonumber
W_{4} &=& \langle \uparrow\downarrow | H_{b} | \downarrow\uparrow \rangle = W_{5} = \langle \downarrow\uparrow | H_{b} | \uparrow\downarrow \rangle = 1/2 \quad , \\ \nonumber
W_{6} &=& \langle \uparrow\uparrow | H_{b} | \uparrow\uparrow \rangle = \epsilon + 2h_{b} \quad .
\eea
The building of a loop in the linked-vertex representation consists of a series of steps. The starting point of a loop is chosen at random. One creates two link-discontinuities (worm ``head'' and ``tail'') here and propagates the head, keeping the tail fixed. The head enters a vertex at one leg (the entry leg) and an exit leg is chosen according to probabilities that depend on the entry leg and the spin states on all the legs. The entrance to the following vertex is given by the link from the chosen exit leg. The spins at all visited legs are flipped, except in the case of a ``bounce'', where the exit is the same as the entrance leg, and only the direction of movement is reversed. When the head and tail meet each other, the loop closes, and a configuration contributing to $Z$ is generated.
Measurements of diagonal operators (like the magnetization $<S^{z}>$) are done over these class of configurations. The statistics obtained from the worm propagation moves (while the loop is open) corresponding to a fixed space-time separation of the worm head and tail directly provide an estimate for the two-point off-diagonal correlation functions 
(like $<S^+_{i}(\tau)S^-_{j}(0)>$).

The probabilities for the different exit legs, given the type of the vertex and an entrance leg, are chosen such that detailed balance is satisfied. This leads to the directed-loop equations. These equations underdetermine the required probabilities, and additional considerations (like minimizing the number of bounces) are required to constrain them. For brevity, we refrain from providing a detailed discussion of this issue here and point to the references mentioned above~\cite{sandvik1,sandvik2}. The relation between the SSE and world-line approaches is explained in Ref.~\cite{sandvik3}. The algorithm becomes slightly more complicated if longer range interactions are included in the Hamiltonian~\cite{sandvik4}.

\begin{figure}[h]
\begin{center}
 \includegraphics[width=10cm,angle=0]{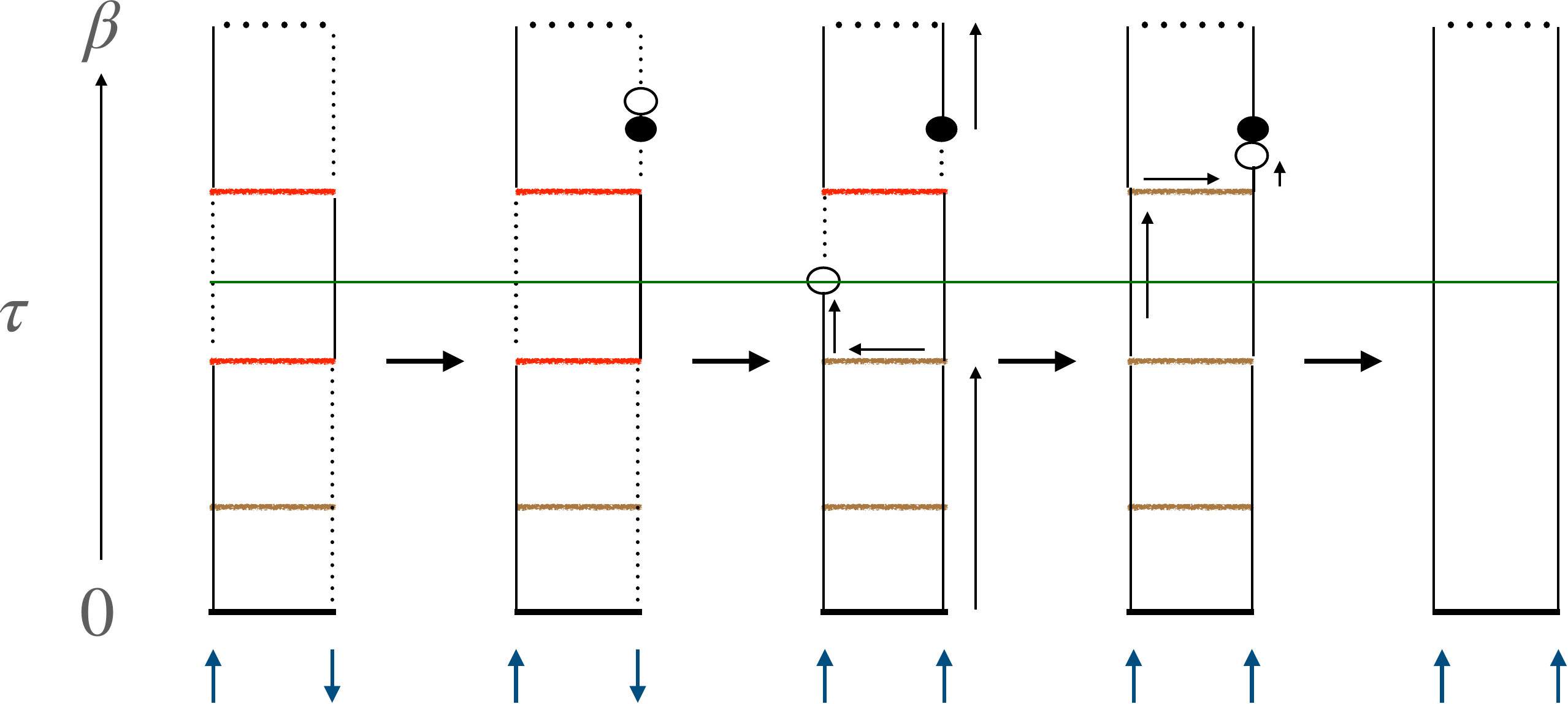}
\caption{ The propagation of a worm within DLA for the $H = \vec{S}_{1}\cdot\vec{S}_{2}$ model. The vertical axis represents the imaginary-time axis running from $0$ to $\beta$, with $\tau=0$ and $\tau=\beta$ identified. The short blue arrows at the bottom denote the starting ($\tau=0$) spin states.  The brown (red) horizontal bars depict the ``diagonal''(``off-diagonal'') Hamiltonian vertices. Solid and dashed vertical lines stand for the propagation of up and down spins respectively. The worm `head' and `tail' are represented by open and filled ovals. Once the worm head propagates, it flips the spins on the go and also may (or may not) change the vertex type it encounters. Moreover, the worm can move from site to site through these vertices, as shown in the third diagram of the sequence. The worm movement is shown by black arrows. Finally, once the worm head and tail meet each other, they annihilate and a new vertex configuration needs to be generated. The number of times the worm head crosses the green line contributes as counts for the $<S^+_{1}(\tau)S^-_{2}(\tau)>$ correlation function.}
 \label{fig:QMC-bosons}
\end{center}
\end{figure}

As a concrete implementation of the DLA within a world-line scheme, we used the DSQSS package~\cite{Motoyama}. In Fig.~\ref{fig:QMC-bosons}, we show a toy example of a worm update in this scheme. One starts with the Heisenberg model ($H = \vec{S}_{1}\cdot\vec{S}_{2}$) of two spins on a bond (in a full lattice model with nearest neighbour interactions, it's enough to consider this for simplicity, since one can always decompose the total Hamiltonian into a sum of nearest neighbour bond operators).
The imaginary time axis is periodic, with $\tau=0$ and $\tau=\beta$ identified. The blue arrows exhibit the starting ($\tau=0)$ spin configuration. One first distributes the ``vertices'' on various nearest neighbor bonds, with \textit{a priori} weights consistent with the temperature and Hamiltonian matrix elements in the $S^{z}$ basis (within SSE, these are given by Eq.~\ref{SSE-VW}).

As depicted in Fig.~\ref{fig:QMC-bosons}, in each configuration, the solid (dashed) black vertical lines denote the propagation of up (down) spins in imaginary time. There are both ``off-diagonal'' (red) and ``diagonal'' (brown) vertices, which either flip the entering spins or not, respectively.

The worm update is depicted in five consecutive steps, whose sequence is marked by solid black arrows. A pair of ``worm heads'' (open and filled black ovals) are introduced at a randomly chosen space-time point. This represents the insertion of an off-diagonal operator in the Ising basis (such as $S^{x}$). Thereafter, the filled oval (called the ``tail'') is kept fixed while the other (called ``head'')  propagates vertically, with periodic boundary condition in the vertical direction, until it encounters a vertex. Each vertex has two entry and two exit points (also called ``legs''). Depending on the ``entry leg'' and the ``leg probabilities'' satisfying directed loop equations, the worm head may travel horizontally (instead of vertically), and may also change the vertex type. It always flips the spins on the legs it goes through. Finally, the propagation terminates when the head meets the tail. The green horizontal arrow in Fig.~\ref{fig:QMC-bosons} denotes an imaginary time slice. The fact that this line is crossed by the worm head once during its propagation,  contributes an integer count to the expectation value of $<S^+_{1}(\tau)S^-_{2}(\tau)>$. 
 
 \newpage
\part{Strongly Correlated Electrons}
\label{part:SCE}
\section{The Hubbard Model}
\label{sec:HM}
Electrons in a tight binding bandstructure, with strong local interactions can be minimally described by the HM~\cite{Hubbard,Kivelson-Hubbard}, as defined by Eq.~(\ref{HM}).

For historical reasons, the electron filling in the HM is often measured by the ``hole concentration'' relative to half filling, i.e. $x=\langle n^h_i\rangle$, as defined by
$n^h_i=1-\sum_{s=\uparrow,\downarrow} n_{is}$.

For $U=0$ the HM reduces to the non-interacting tight binding model $H^{\rm SL}$  of Eq.~(\ref{SL}). As shown in Section \ref{sec:RH-WS}, $R_{\rm H}^{(0)}$  is given by the isotropic scattering time solution of Boltzmann's equation which is depicted in Fig.~\ref{fig:SL-RH}.

\begin{figure}[h!]
\begin{center}
\includegraphics[width=8cm,angle=0]{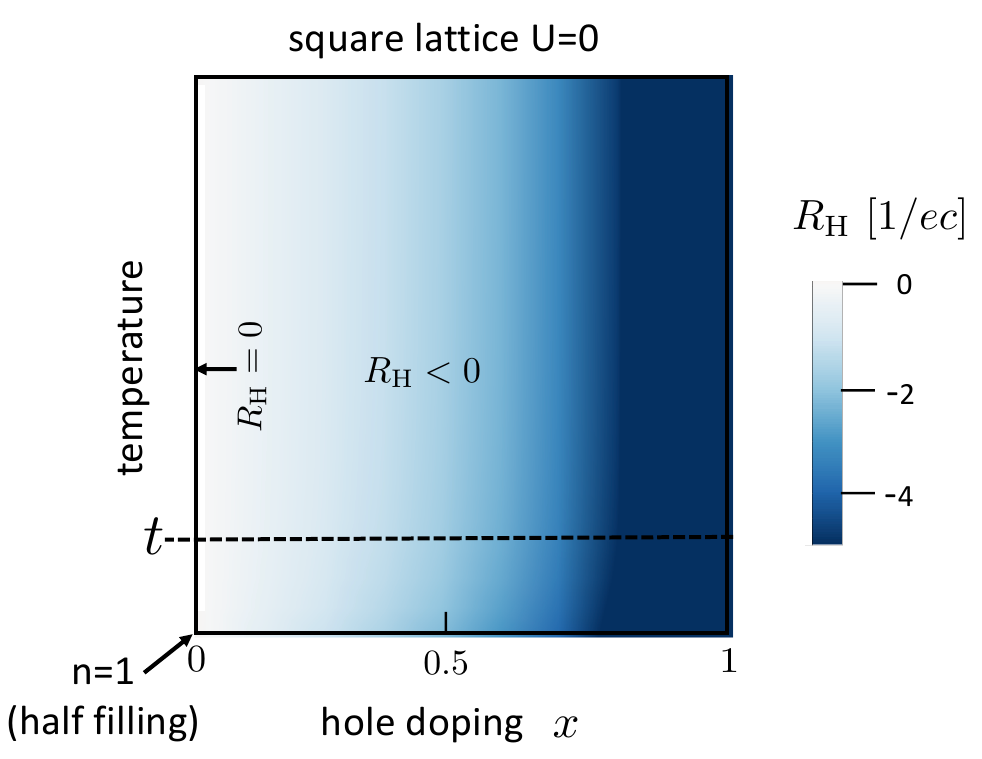}
\caption{ 
Hall coefficient  of  weakly disordered ($U=0$) square lattice.  $R_{\rm H}^{\rm SL}$ given by Eq.~(\ref{RH-Boltz}).  The strongest prediction of Boltzmann theory for this model is that  $R_{\rm H}<0$  for $0\!<\!x\!<\!1$, and does not  diverge  anywhere except at $x\to 1$. }
 \label{fig:HallMap-SL}
\end{center}
\end{figure}
The Hall coefficient temperature-doping map of the nn SL model, is depicted in Fig.~\ref{fig:HallMap-SL}. We note that the Hall sign is negative for all $x>0$, and vanishes half filling  $x_{\rm sign}=0$. This Hall map will be contrasted with the Hall map in the $U\gg t$ regime. 

\section{The t-J Model}
\label{sec:tJM}
In the strong interacting regime of $U  \gg t$, the electronic correlations of the HM differ drammatically from the predictions of the weakly interacting model on the SL.
This is true especially near  half filling $x=0$, where a charge gap of order $U$ appears between the singly occupied spin and hole configurations, and configurations which contain  doubly occupied sites. At zero temperature
$T\ll U$, the half filled HM describes a Mott insulator, which orders antiferromagnetically. Away from half filling, the ground state phase diagram is still under debate, with variational studies indicating possible charge and spin density wave order, and/or $d$-wave superconducting order~\cite{PEPS,Sorella,Kivelson-Hubbard}.

As argued in subsection \ref{sec:Renorm}, at  temperatures $ T\ll U$, the thermodynamic approaches converge much better after the HM Hamiltonian is renormalized onto its lower energy  Gutzwiller-projected (GP) subspace of no-double-occupancies. There, for $U\gg t$,  the HM maps onto the t-J model (tJM)~\cite{IEQM,tJ-Spalek}, which to second leading order in $t/U$,
\bea
H^{\rm tJM} &=& \cP_{\rm GP}\left(  H^t+H^J +H^{J'} \right) \cP_{\rm GP}+\cO(t^3/U^2) \quad ,\nonumber\\
H^t  &=&  -\tt\sum_{\ij }  K^+_{ij} \quad ,\nonumber\\
H^{J}   &=& J \sum_{\ij}  \bs_i\cdot \bs_j (1\!-\!n^h_i)(1\!-\!n^h_j) \quad ,\nonumber\\
H^{J'}   &=& - {J \over 4} \sum_{\ij \jk}   \left(   K^+_{ik} - 2{\bf \Sigma}^+_{ik}\cdot \bs_j\right)(1\!-\!n^h_j) \quad .
\label{tJM}
\eea
The GP bond operators are
\bea
&&{K}^{\pm}_{ij}  \equiv    \sum_s \tc^\dagger_{is} \tc^{\nd}_{j s} \pm \tc^\dagger_{js} \tc^{\nd}_{i s} \quad , \nonumber\\
&& {\bf \Sigma}^{\pm}_{ij}   \equiv   \sum_{ss'} \tc^\dagger_{is} \vec{\sigma}_{ss'} \tc^{\nd}_{js'}  \pm  \tc^\dagger_{js}\vec{\sigma}_{ss'} \tc^{\nd}_{is'} \quad ,
 \label{K-def}
\eea
where $\tc^\dagger_{is}  \equiv  c^\dagger_{is} (1-n_{i,-s})$. 
$K^+$ ($K^-$) describes the bond kinetic energy (current)  and $\Sigma^\pm$ describes the  spin-dependent hopping.  
$J=4t^2/U$ is Anderson's  antiferromagnetic superexchange energy. 

\subsection{Linear Resistivity slope}

\begin{figure}[h!]
\begin{center}
\includegraphics[width=7cm,angle=0]{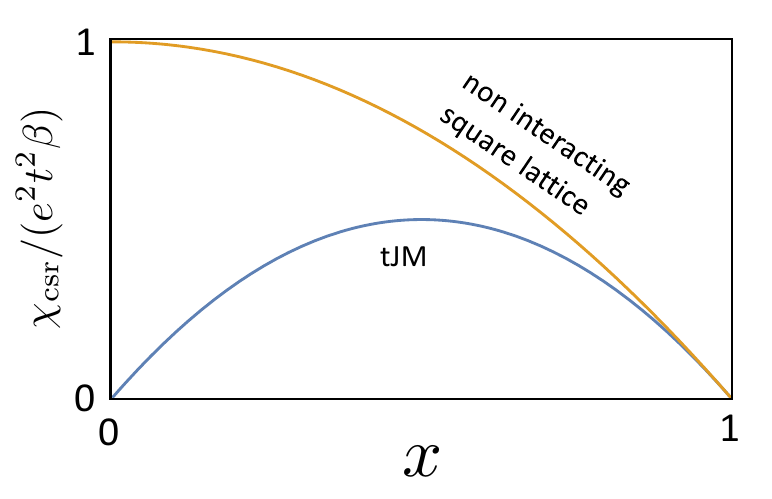}
\caption{Doping dependent CSRs at intermediate temperatures $t<T\ll U$.
In contrast to the non-interacting case, the CSR of the tJM vanishes toward the Mott phase as $x\to 0$. This results in an interaction induced divergence of $R_{\rm H}$ and the resistivity slope.}
 \label{fig:CSR-tJ}
\end{center}
\end{figure}

$H^t$, (the ``$t$-model''), dominates the metallic charge transport at $U\gg T\gg J$. The doping dependence of the CSR 
of the $H^t$ was  calculated by Jaklic and Prelovsek~\cite{Jak} and Perepelitsky \etal~\cite{Perepel} up to order $\beta^3$ :
\be
\chi_{\rm csr}=2\beta e^2\tt^2 x(1-x) 
+{ \beta^3 e^2  \tt^4   \over 6}  x(1-x) (-9 + 2 x + x^2 ) + \cO(\beta^5\tt^6) \quad .
\label{CSR-hiT}
\ee
As depicted in  Fig.~\ref{fig:CSR-tJ}, the CSR vanishes toward $x\to 0$ as a consequence of the GP. It is also suppressed relative to the non-interacting limit even far from half filling. Vanishing of the CSR  produces a divergence of the Hall coefficient by Eq.~(\ref{RH0}).

\begin{figure}[h]
\begin{center}
\includegraphics[width=8cm,angle=0]{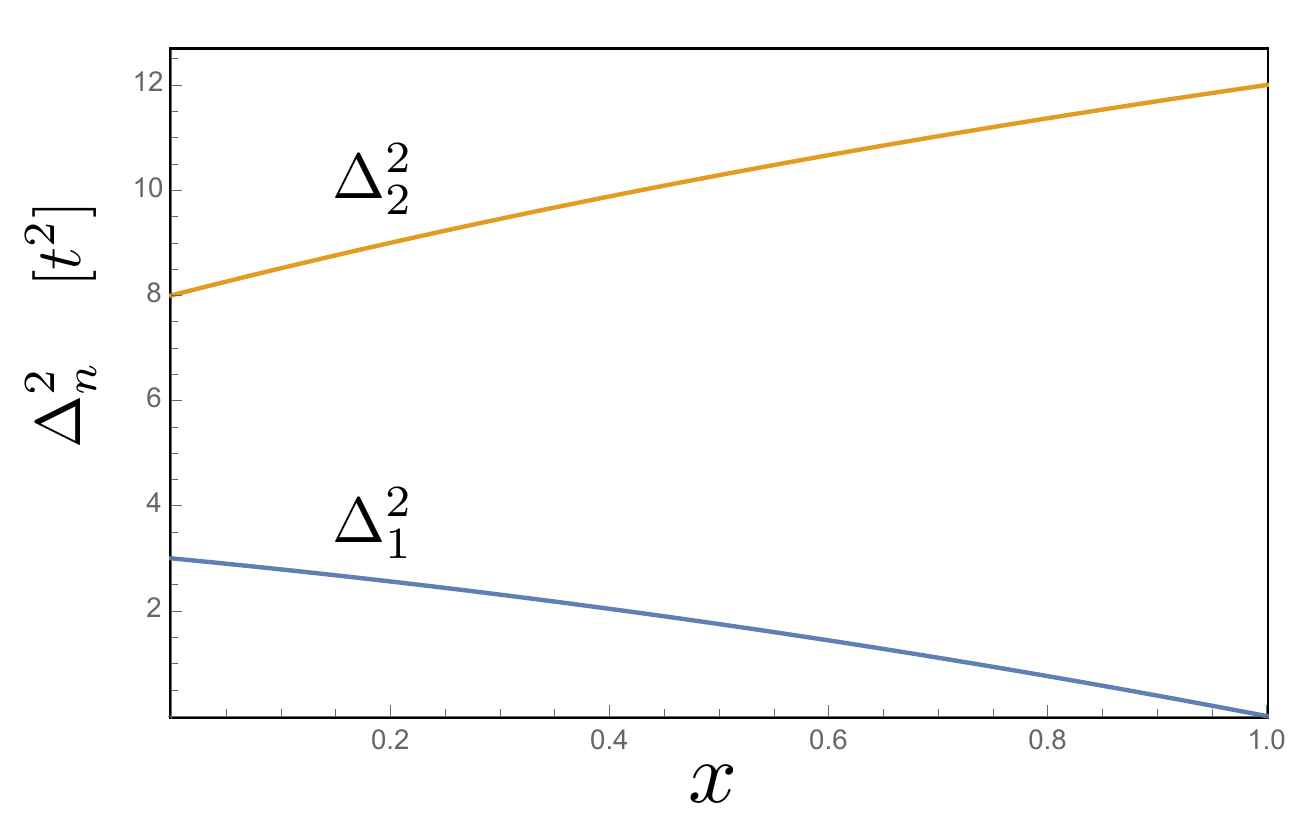}
\caption{The first two recurrents of $H^t$ at order $\beta^0$ as a funcrtion of hole doping $x$.}
 \label{fig:DeltaRatios}
\end{center}
\end{figure}
The high temperature limit of the first two conductivity recurrents of $H^t$ are 
\be
\Delta_1^2=   t^2( 3  - 2x -x^2) ) \quad , 
~~\Delta_2^2 =  t^2 {24(1+x)\over 3 +x }\quad ,
\label{recurrents}
\ee
which are shown as functions of doping $x$ in Fig.~\ref{fig:DeltaRatios}.
$\Delta_1^2$ was calculated by Jak \etal~\cite{Jak}. $\Delta_2^2$
was calculated by Khait \etal  in Ref.~\cite{tJM}.

The DC conductivity as a CF is,
\be
 \sigma_{xx}(\omega) =-\beta \chi_{\rm csr}\lim_{\ve\to 0}  \Im {1 \over    i\ve -{\Delta_1^2  \over i\ve - \Delta_2 G^>_{22}(i\ve) } } \quad .
 \ee
\begin{figure}[h]
\begin{center}
\includegraphics[width=8cm,angle=0]{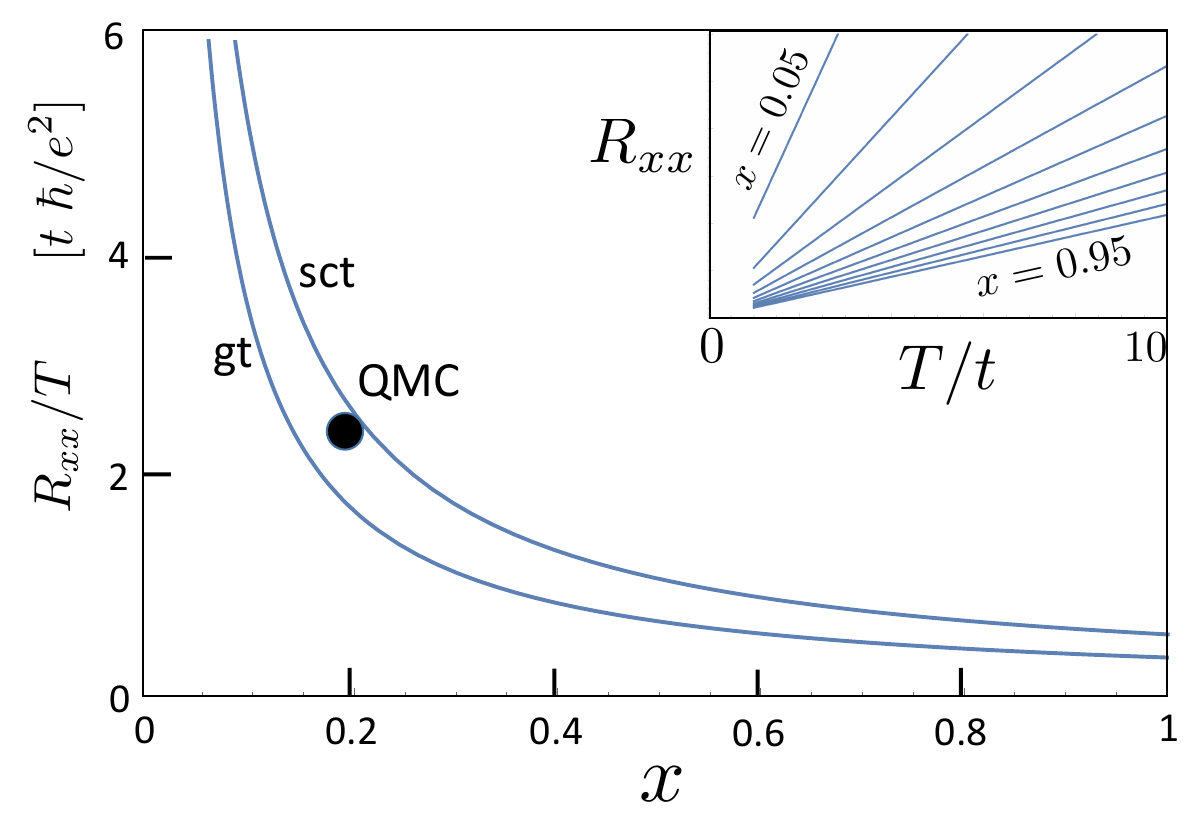}
\caption{High temperature resistivity slopes of  $H^t$.
sct and  gt denote two different continued fraction extrapolation schemes  (see text).  Inset depicts  $R_{xx}(T)$ for  gt approximation.
The solid circle marks the Quantum Monte Carlo result for the $U\!=\!8t$ HM, reported in Fig. S3 of Ref.~\cite{Dev-Rxx}.  \label{fig:RXXvsT} }
\end{center}
\end{figure}
 $\Delta_1$ and $\Delta_2$, can be extrapolated to evaluate $\Im G^>_{22}$  by semicircle termination (SCT), Eq.~(\ref{SCT})),
 \be
 (\bar{\Delta}^{\rm sct}_n)^2 \Rightarrow  \Delta^2_2 \quad , \quad n=2,3,\ldots \infty \quad ,
 \ee
 and by Gaussian termination (GT), Eq.~(\ref{GT})),
 \be 
 (\bar{\Delta}^{\rm gt}_n)^2 \Rightarrow {1\over 2} n \Delta^2_2 \quad , \quad n=2,3,\ldots \infty \quad . 
 \ee 
In Fig.~\ref{fig:RXXvsT}, the high temperature DC resistivity slopes of $H^t$  are plotted for both termination approximations. We see that the slopes diverge toward the Mott limit, as expected by the suppression of the CSR depicted in Fig.~\ref{fig:CSR-tJ}. 
Interestingly, the  resistivity is finite at high temperatures even in the dilute electron density limit   $x\to 1$. The slopes of Fig.~\ref{fig:RXXvsT}  are in qualitative agreement with the HM calculation in Ref.~\cite{Dev-Rxx}.

Note: Perepelitsky \etal~\cite{Perepel} had evaluated the high temperature resistivity slope of the large $U/t$ HM in the high dimensional lattice (large-$d$) limit as,
\bea
R_{xx}^{d\to \infty}(x,T) &\approx& {1\over x}~ R_0(x=1,T)\quad ,\nonumber\\
R_0(x=1,T)&\propto& {T\over t}\quad.
\eea
In the large $d$ limit, vertex corrections to the conductivity bubble can be neglected~\cite{Perepel}. As $x\to 1$, the vanishing compressibility effectively cancels the  diverging effective scattering time. This conclusion agrees with the CF extrapolation of Ref.~\cite{tJM} at $d=2$,
\be
R_{xx}^{d=2}={T\over t e^2} {2\over \sqrt{\pi}}  {\Delta^2_1 \over \chi_{\rm CSR} ~\Delta_2 }\propto {T\over x t} \quad ,
\label{Sxx-GT}
\ee
where the $x$-dependence  is given by  Eqs.~(\ref{CSR-hiT}) and (\ref{recurrents}).   
In the CF, the  ratio $\Delta_2/\Delta_1^2$ is interpreted as the effective scattering time, which  diverges as $x\to 1$.   
The resistivity slopes therefore saturate at $x\to 1$ as depicted in Fig.~\ref{fig:RXXvsT}.

\begin{figure}[h!]
\begin{center}
\includegraphics[width=8cm,angle=0]{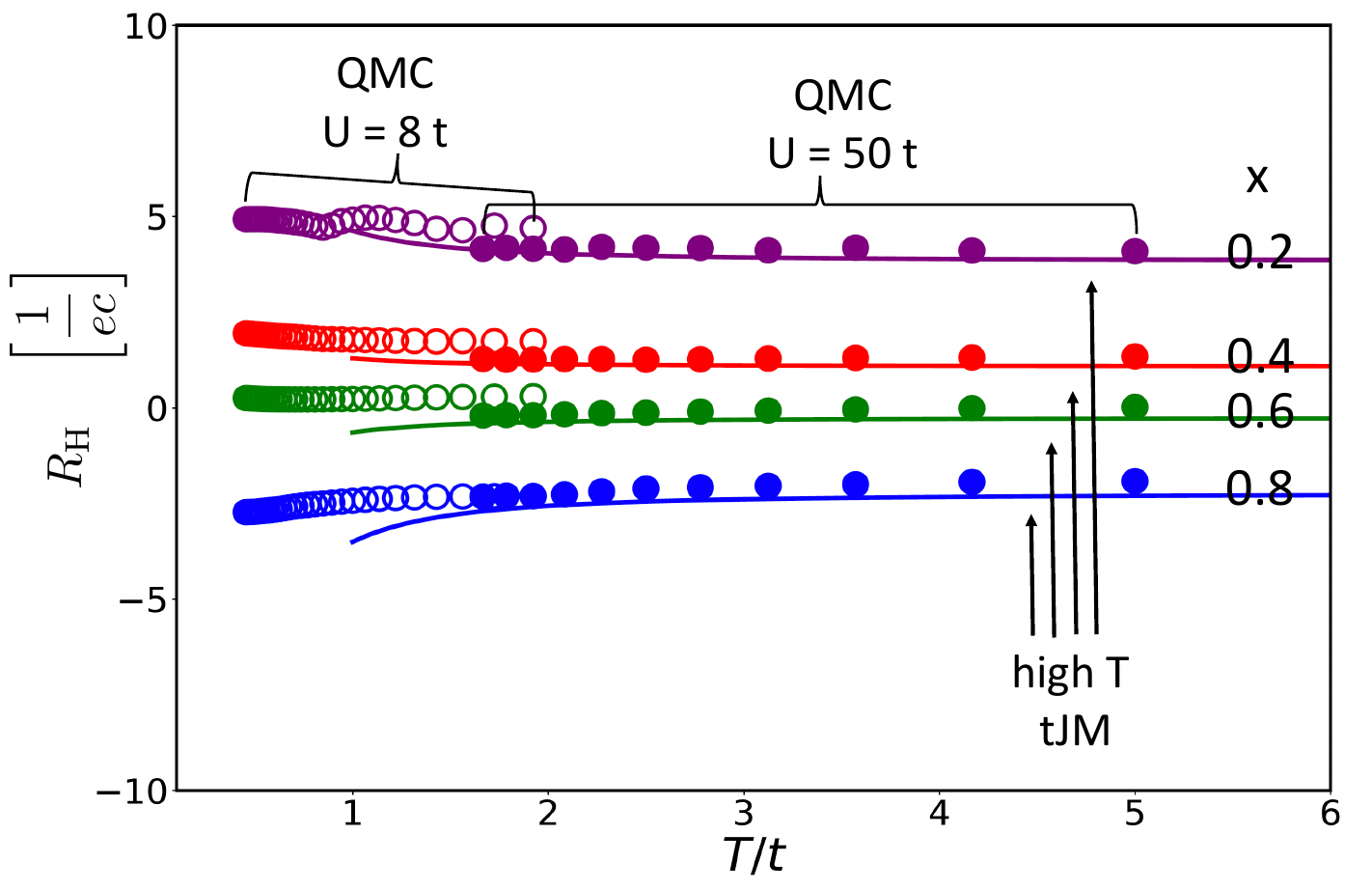}
\caption{
Hall coefficient $R_{\rm H}^{(0)}$  in the intermediate temperature regime $J \le T \ll U$ from Ref.~\cite{tJM}. Lines depict the asymptotic high temperature values for the $t$-model Eq.~(\ref{RH-t}), which are calculated by the high temperature series expansion reviewed in Section \ref{sec:HighT}. The solid and open  circles are obtained by DQMC (see Section \ref{sec:DQMC})  using  Boltzmann weights  of the HM.
The different values of $U/t$ are used to  cover different temperature regimes. The DQMC results are shown in regimes of negligible fermion sign error. We note that the high temperature expansion agrees with the DQMC data down to $T\simeq 2t$, and  
the DQMC data shows quite weak temperature dependence down to $T\simeq J$.  } \label{fig:QMC-RH}
\end{center}
\end{figure}

\subsection{Hall map for large $U/t$}
$R_{\rm H}$ of the t-model was calculated in Ref.~\cite{tJM} by Eq.~(\ref{RH}), using the high temperature series  described in Section \ref{sec:HighT}. 
The doping dependent CMC  was determined up to order $(\beta \tt)^4 $,
\bea
\chi^t_{\rm cmc} &=&     {\beta^2\tt^4 e^3  \over 2c }   x(1-x)  (-5+10x  + 3x^2 )\nonumber\\
&+&   {\beta^4\tt^6 e^3\over 16 c} x(1-x) (45 - 136 x + 50 x^2 +48 x^3 - 71 x^4) \quad .
\label{CMC-hiT}
\eea
The  zeroth Hall coefficient was determined up  to second order in $\beta t$,
 \be 
 R^{t,(0)}_{\rm H}=  {1\over ec} \left(  {-5  +10 x + 3 x^2\over  8 x(1-x)} \right)+
 \left( (\beta \tt)^2   { -45 - 53 x + 145 x^2 + 225 x^3\over 192  x}\right) \quad .  
\label{RH-t}
\ee 
The calculation of  $R^{t,(0)}_{\rm H}$ was extended down to  the intermediate temperature (IT) regime: $J\le T\ll U$, by DQMC which is reviewed in Section~\ref{sec:DQMC}. The weights are given by the HM, which at large $U/t$ effectively includes the
effects of the spin interactions $H^{J}$  in the tJM.  
The DQMC data is plotted in Fig.~\ref{fig:QMC-RH}. The analytic high temperature series curves of Eq.~(\ref{RH-t}) are depicted by solid lines. 
The DQMC was performed using
Boltzmann weights of the strongly interacting HM on the square lattice.
Typical feasible system sizes are chosen between $8\!\times\!8$ and $12 \times 12$ lattice sites, where we found little size dependence of our results. This indicates a short correlation length in the studied temperature regime.
The imaginary time step was  chosen to render the Trotter errors to be insignificant.
The  number of Monte Carlo sweeps was generally $\sim 10^{5}$.

The DQMC data shown in Fig.~\ref{fig:QMC-RH} was included for temperatures and doping concentrations  where the average fermion-induced sign of the Boltzmann weights $\langle S\rangle$ (see  Eq.~(\ref{QMC-sign})), is  higher than 0.8. Ref.~\cite{tJM} also showed that $\langle S\rangle$ decreases or large values of $U/t$, and at intermediate doping concentrations, while it approaches unity at high temeratures and for $x\to 1$ and $x\to 0$.

In the extremely high temperatures regime $T\gg  U$, the tJM is not valid. One can  use the unprojected HM of Eq.~(\ref{HM}) since the suppression of double occupancies in the HM diminishes since the Boltzmann weights become weakly dependent of $U$.
$R_{\rm H}^{(0)}$ is then simply given by the high temperature limit of the non-interacting SL, i.e.,
\bea
&&  \chi^{\rm HM}_{\rm csr}\sim \beta e^2 t^2 n(2-n) \quad , \quad
 \chi^{\rm HM}_{\rm cmc}\sim   \beta^2 {e^3\over c} t^4 n(2-n)(1-n) \quad ,   \nonumber\\
&& R^{(0)}_{\rm H} =  {2(1-n) \over   n(2-n)ec} +\cO(\beta U)^2 \simeq R_{\rm H}^{\rm SL} \quad ,
\label{RH-HM}
\eea
where $n=1-x$ is the electron density. 

Interestingly, at lower temperatures, which lie within the applicability of the tJM, the recovery of the $U\simeq 0$ behavior at high temperatures is heralded by the effects of $H^{J'}$ in the tJM (\ref{tJM}).
$H^{J'}$ become important in the CMC susceptibility, by adding to the currents and magnetizations 
next neighbor hopping terms, 
\bea
&&{j'}^{\alpha}_{ijk}= -i eJ (1-n^h_j) \left( K^-_{ik;\alpha}-   2{\bf \Sigma}^-_{ik;\alpha}\cdot \bs_j\right)   \nonumber\\
&&M' =   {1\over 2c}   \sum_{\ij\jk} ( x_i {j'}^{y}_{  ijk}- y_i  {j'}^x_{ ijk}  ) \quad .
\label{OrderJ}
\eea
Since $J\ll t$, these terms are unimportant for the CSR. However, for the CMC they contribute one less power of $\beta H $ than those $\chi_{\rm cmc}^t$, since they encircle a magnetic flux with one less Hamiltonian bond. Hence 
they become dominant at  temperatures of order $T\simeq U $,
\be
\chi^{J'}_{\rm cmc} = {\beta J e^3\over 2} x(1-x)(1+2x-3x^2) \quad .
\label{CMC-J'}
\ee
 We note that $\chi^{J'}_{\rm cmc}$  has  the opposite sign to $\chi^{t}_{\rm cmc}$, and $H^{J'}$ yields the contribution to $R_{\rm H}^{(0)}$ of,
\be
\Delta R^{J'}_{\rm H} =  {4 T\over ecU} \left(  {1+2x-3x^2  \over 8  x (1-x)} \right) \quad .
\label{RH-J}
\ee
Therefore,
as $x\to 0$ and $T\to U$,  $\Delta R^{J'}_{\rm H}$ reduces the positive Hall coefficient divergence, and  the  doping levels of the  Hall sign reversal. Its effect is observed in the upper blueish regions of Fig.~\ref{fig:HallMap-U}.
\subsection{Hall coefficient corrections}
  \label{sec:Corr}
   \begin{figure}[h!]
\begin{center}
\includegraphics[width=8cm,angle=0]{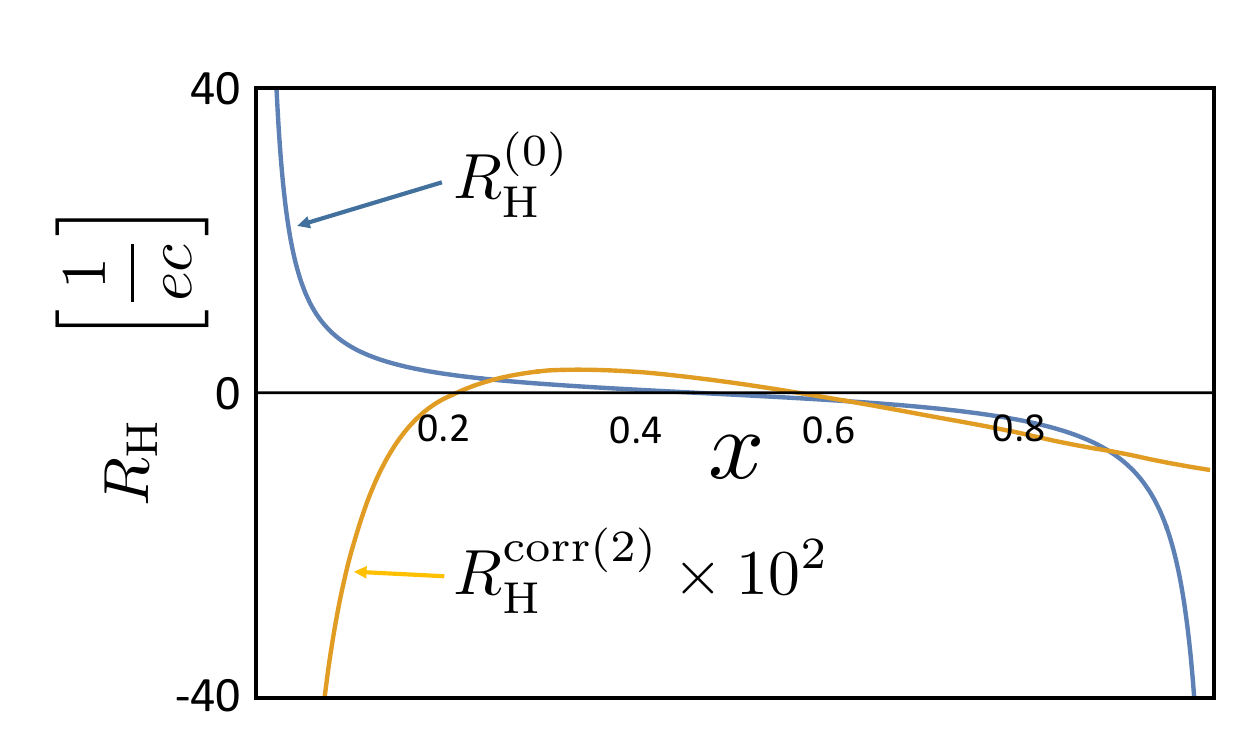}
\caption{Comparison of the the leading order high temperature Hall coefficient of $t$-model, defined in Eq.~(\ref{RH-t}) (blue line), with its second order  correction term (orange line, multiplied by 100 for visibility) defined in Eq.~(\ref{RHcorr2}).   
The ratio of magnitudes  vanishes at $x\to 1$, and approaches $0.06$ at $x\to 0$. 
}
 \label{fig:RHcorr}
\end{center}
\end{figure}

We calculate the correction term up to second order
\be
R_{\rm H}^{{\rm corr}(2)} =  {1\over \chi_{\rm csr}} \left(  \left({\Delta_1\over \Delta_2}\right)^2 \cM_{2,2}  -\left({\Delta_1\over \Delta_2}\right)\left( \cM_{2,0} +\cM_{0,2} \right) \right)\! \quad .
 \label{RHcorr2}
 \ee
The three hypermagnetization matrix elements $\cM_{02},\cM_{20},\cM_{22}$  were evaluated numerically. 

The calculation of $M_{2,2}''$ to leading order in $(\beta t)$  involved traces over up to $10^5$ operator clusters.

In Fig.~\ref{fig:RHcorr}, we plot the final result for $R_{\rm H}^{{\rm corr}(2)}$ for all doping concentrations.
We see that in comparison $R_{\rm H}^{(0)}$, the
its quantitative effect is  negligible, and maximized toward $x\to 0$ by
 \be
\lim_{x\to 0}  \left| R_{\rm H}^{{\rm corr}(2) }   /R^{(0)}_{\rm H} \right| \to 6\% \quad .
\label{ratio}
 \ee

Based on the  high temperature ($\cO(\beta t)^0$) result in Eq.~(\ref{ratio}) , and the weak temperature dependence found
for $R_{\rm H}^{(0)}$ shown in Fig.~\ref{fig:QMC-RH}, we may assume that the correction term remains negligible throughout the IT regime.
  We note that we have not calculated $R_{\rm H}^{\rm corr}$ for  the HM at $T\ge U$.

\subsection{Numerical calculation Krylov operators and hypermagnetization matrix elements}
\label{sec:numericalOperators}
The high temperature moments $\mu_{2},\mu_4$ which yield the recurrents $\Delta_1,\Delta_2$,  and the hypermagnetization matrix elements $\tilde{\cM}_{02},\tilde{\cM}_{20},\tilde{\cM}_{22}$  defined in Section \ref{sec:Rcorr}, were evaluated numerically by symbolic manipulation. 

The operator $\cL^n j^\alpha_i$ is expanded in terms of  products of site operators $O_{i_1}({\br_1}) \cdot  O_{i_2}(\br_2)\cdot ... \cdot O_{i_N}(\br_N)$ with  real coefficients that are stored separately.
Each application of the Liouvillian or the hypermagnetization changes the hyperstate after  multiplying the operators site-by-site using the tJM multiplication Table~\ref{table:opmult} of the 8 GP fermions and spin operators. 
One  keeps track of the overall order of the fermion  operators $\tilde{c}_i,\tilde{c}_j^\dagger$, and the negative signs produced when reordering the contributions to the same final  hyperstate from different multiplication paths.

\begin{table}[h!]
	\centering
	\begin{tabular}{ |c|c|c|c|c|c|c|c|c| } 
		\hline
		& $\tilde{c}^\dagger_\uparrow$ & $\tilde{c}^\dagger_\downarrow$ & $\tilde{c}_\uparrow$ & $\tilde{c}_\downarrow$ & $n_\uparrow$ & $n_\downarrow$ & $ \tilde{s}^+$ & $\tilde{s}^-$ \\ 
		\hline
		$\tilde{c}^\dagger_\uparrow$ & 0 & 0 & $n_\uparrow$ & $\tilde{s}^+$ & 0 & 0 & 0 & 0 \\ 
		\hline
		$\tilde{c}^\dagger_\downarrow$ & 0 & 0 & $\tilde{s}^-$ & $n_\downarrow$ & 0 & 0 & 0 & 0 \\ 
		\hline
		$\tilde{c}_\uparrow$ & $n_h$ & 0 & 0 & 0 & $\tilde{c}_\uparrow$ & 0 & $\tilde{c}_\downarrow$ & 0 \\ 
		\hline
		$\tilde{c}_\downarrow$ & 0 & $n_h$ & 0 & 0 & 0 & $\tilde{c}_\downarrow$ & 0 & $\tilde{c}_\uparrow$ \\ 
		\hline
		$n_\uparrow$ & $\tilde{c}^\dagger_\uparrow$ & 0 & 0 & 0 & $n_\uparrow$ & 0 & $\tilde{s}^+$ & 0 \\ 
		\hline
		$n_\downarrow$ & 0 & $\tilde{c}^\dagger_\downarrow$ & 0 & 0 & 0 & $n_\downarrow$ & 0 & $\tilde{s}^-$ \\ 
		\hline
		$\tilde{s}^+$ & 0 & $\tilde{c}^\dagger_\uparrow$ & 0 & 0 & 0 & $\tilde{s}^+$ & 0 & $n_\uparrow$ \\ 
		\hline
		$\tilde{s}^-$ & $\tilde{c}^\dagger_\downarrow$ & 0 & 0 & 0 & $\tilde{s}^-$ & 0 & $n_\downarrow$ & 0  \\ 
		\hline
	\end{tabular}
	\renewcommand{\tablename}{Table}
\caption{Multiplication table of GP operators in the tJM. The entry $O_{i,j}=O_i \cdot  O_j$, where $i$ and $j$ are row and column respectively, and 
$\tilde{s}^\alpha=s^\alpha (1-n^h),\quad n_{\uparrow,\downarrow}= (1-n^h)\left({1\over 2}\pm \tilde{s}^z\right)$.
}
	\label{table:opmult}
\end{table}
When evaluating the final traces to obtain the moments or matrix elements,  most operators vanish unless they only contain unit operators and  density operators $n_\gamma$.

\section{Discussion}
Sign reversal of the Hall coefficient at low doping
has been previously obtained by dynamical mean field theory (DMFT)~\cite{DMFT2013,Khomskii},  QMC~\cite{Dev-PR,Comm:QMC-AC} and  determinant QMC~\cite{DMFT2018}. 
These methods have found  evidence of hole pockets in the momentum dependent occupation,
which is qualitatively consistent with our results at low doping.  
 Refs.~\cite{SignSxy-1,SignSxy-2,New-DMFT} calculate (within DMFT) the Hall conductivity of the Hubbard model  at strong magnetic fields. They found that the Hall  sign is reversed relative to band theory, near half filling. These effects were attributed to the Chern numbers of the non-interacting Hofstadter's butterfly bands of the square lattice. It is interesting that  these  sign changes which were predicted at strong fields,
(as measured in strongly correlated flat bands Moir\'e systems~\cite{Science17}), are qualitatively similar to the Hall sign we obtain in the weak field limit. 

Here however, we find that the sign reversal occurs already at  $x\le 0.45$, which may come  as a surprise vis-a-vis the widely
used band  theoretical approaches at much lower doping. The reason is simply related to the
spin and charge entangled commutation relations of GP current operators of the tJM, 
\be
 [K^{-}_{12},K^{-}_{23}]\! =\!  K^{-}_{13} \left( {1+  n^h_{2} \over 2}\right) 
\!+\! {\bf \Sigma}^{-}_{13} \cdot   {\bf s}_{2}(1-n^h_2) \quad ,  
\label{comm}
\ee
which affects the hole density dependence of the CMC, and determines the doping concentration of the sign reversal.

Since the hole density operators have coefficients of order unity, it  is natural that the sign change occurs at a fraction with a 
denominator not much larger than unity.
The important lesson we can learn from this is that
the effects of GP reach far into  the high doping and temperature regimes.

Previous QMC  calculations  of  $R_{\rm H}^{(0)}$ for the  HM~\cite{Dev-PR,Dev-RXY} have used our Eq.~(\ref{RH}),
and neglected $R_{\rm H}^{\rm corr}$.
They have also reported a positive Hall sign near half filling (but no apparent divergence) for the square lattice model.  
However in the regime of $U/t \simeq 16$,  $R^{\rm corr}\propto||[H,j^x]||\sim U/t$, is expected to  dominate over $R_{\rm H}^{(0)}$, and hence cannot be ignored.

The     difference between $R_{\rm H}^{(0)}({\rm HM})$~\cite{Dev-RXY} and  $R_{\rm H}^{(0)}({\rm tJM})$ of Eq.~(\ref{RH-t}),  can be explained by the fact that $R^{\rm corr} ({\rm tJM}) \ll R^{\rm corr}({\rm HM})$  in the IT regime.

We can compare $R_{\rm H}^{(0)}$  of  Eq.~(\ref{RH-t}) to the  infinite frequency Hall coefficient of the t-model calculated at leading order in $\beta$ by Shastry, Shriman and Singh~\cite{SSS}, 
\be
R_{\rm H}^*={d\over dB} \lim_{\beta\to 0,\omega\to \infty} \rho_{xy}(\omega) = {1\over ec} \left( -{1\over 4x}+{1\over 1-x} -{3\over 4}\right) \quad .
\ee
$R_{\rm H}^*$ changes sign at $x=1/3$ and 
diverges as $1/(4x)$ toward the Mott limit. While   Eq.~(\ref{RH-t}) changes sign at $x=0.4415$ and diverges as $5/(8x)$ at small $x$.
 Still, the qualitative similarity  we find between the infinite and zero frequencies
 is surprising, but we cannot infer any general relations from this coincidence. We note that $R_{\rm H}^*$ may be relevant to the optical Kerr effect~\cite{Kerr}.

In summary, the calculation of the DC Hall coefficient as a function of hole doping in the IT regime appears to be well controlled. 
We can learn from it that the Mott insulator phase affects the charge carriers in the nearby metallic phase. This has theoretical implication to any  possible ordered phases at lower temperatures, such as superconductivity and/or other orders~\cite{DMRG,Sorella,Troyer,RVB}.
The superconducting order parameter should  consist of GP holes, with spin entangled commutation relations~(\ref{comm}), rather than 
quasiparticles near the non-interacting Fermi surface. The GP also governs the relation  between superfluid stiffness (which is bounded by the CSR), and the electron density~\cite{Uemura}. It affects the moving vortex charge and the Hall sign  in the flux flow regime~~\cite{RH-Kim,SciPost}.
\newpage
\part{Strongly Correlated Bosons}
\label{part:SCB}
\begin{figure}[h!]
\begin{center}
\includegraphics[width=8cm,angle=0]{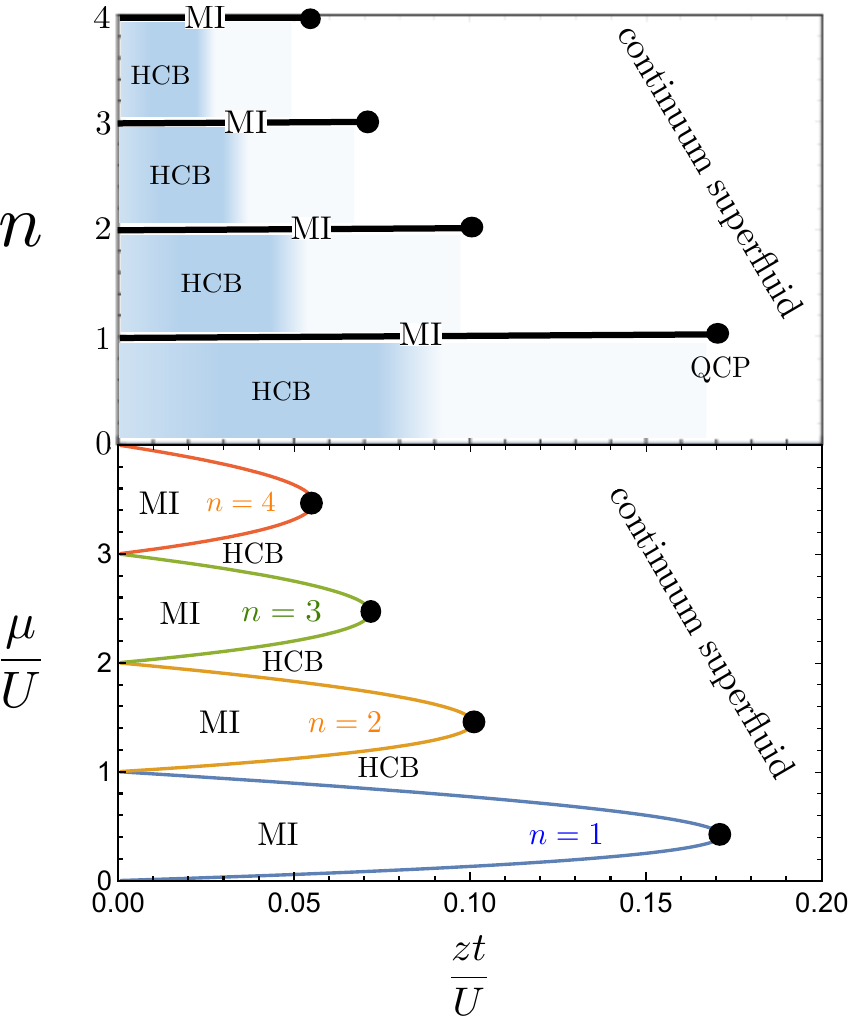}
\caption{$T=0$ phase diagram of the Bose Hubbard Model, Eq.~(\ref{BHM}). At weak interactions $U/zt\to 0$, the lattice periodicity is unimportant, and the bosons condense into a continuum superfluid ground whose excitations are goverened by the Gross-Pitaevskii field theory~(\ref{GP}). The quantum critical points (marked by black disks), 
denote second order transitions into gapped Mott insulators of integer fillings $n$.   Sandwiched between the Mott insulator phases, at fractional average densities, the Hamiltonian renormalizes into Hard Core Bosons (HCB) model Eq.~(\ref{HCB}), whose boson occupations are constrained to fluctuate between two consecutive integers.
}
 \label{fig:BHM}
\end{center}
\end{figure}

\section{The Bose Hubbard model}
The Bose Hubbard model (BHM) is a minimal model of interacting lattice bosons,
\be
H^{\rm BHM}= -t\sum_{\ij} \left( e^{-iqA_{ij}/c} a^\dagger_i a^{}_j + e^{iqA_{ij}/c} a^\dagger_j a^{}_i \right) +   U \sum_i  n_i^2 -\mu\sum_i n_i \quad.
\label{BHM}
\ee
$t$ is the hopping rate and $U$ is the local repulsive interaction, $q$ is the boson charge, and $c$ is the speed of light. The lattice constant is unity.
$a^\dagger_i$ creates a boson on lattice site $i$, and $n=a^\dagger_{i} a^{}_i$ is the boson occupation.  $A_{ij} =  \int_{\bx_i}^{\bx_j}d\bx\cdot \bA$, where $\bA$
is the electromagnetic vector potential. 

$H^{\rm BHM}$ can be realized by cold bosonic atoms trapped in an optical lattice, where $\bA$ can be implemented by light induced artificial gauge fields~\cite{Spielman}.

The BHM also describe a Josephson junction array where the BCS pairing gap in each grain is larger than the Josephson coupling $t$. $U$ would be the inverse capacitance of each grain.

The BHM has a well known phase diagram~\cite{Scalettar} which is shown in Fig.~\ref{fig:BHM}.

\subsection{Weak interactions}
We note that for weak interactions, $t/U\to \infty$,
the bosons condense into the $\bk\simeq 0$ eigenstates  with effective dispersion $\omega_\bk \simeq {|\bk|^2 \over2  m^*}$. The low energy theory is therefore effectively a Galilean symmetric continuum theory described by the Gross-Pitaevskii (GP) action~\cite{pitaevski},
\be
S^{\rm GP}=\int dt d^d x \left( i \psi^*\dot{\psi} - {1\over 2m^*}|\bnabla - i{q\over c} \bA \psi|^2 + U |\psi|^4 - \mu|\psi|^2 \right) \quad ,
\label{GP}
\ee
where $\psi(\bx_i)$ is a complex coherent-state field which represents the lattice boson operator $a_i$.
For dimension $d>2 $, $S^{\rm GP}$ describes a superconductor at low temperatures, and a charged bosons gas above a finite transition temperature. In the presence of impurities, the metallic conductivity tensor  is described by Drude theory Eq.~(\ref{S-Drude}) with scattering rate  which vanishes with impurity concentration. Therefore by (\ref{RH-Drude}), the Hall coefficient is equal to the Galilean symmetric result,
\be
R^{\rm GP}_{\rm H} = {1\over nqc} \quad ,
\label{RH-Gal}
\ee
where $n$ is the boson density.
No Hall sign change is expected at as  function of density or temperature in the weak interactions regime.  

\subsection{Strong interaction}
\begin{figure}[h!]
\begin{center}
\includegraphics[width=8cm,angle=0]{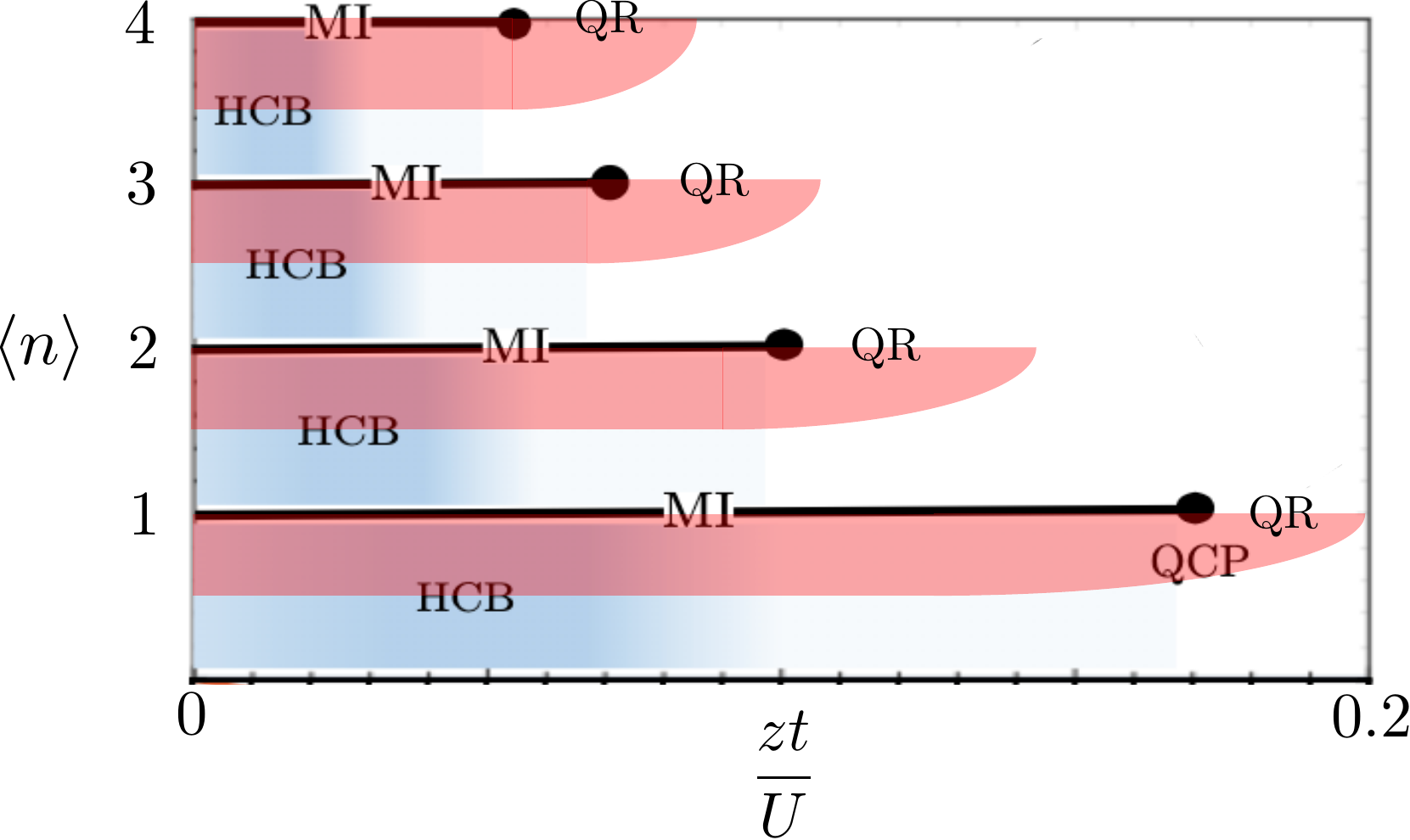}
\caption{Hall signs (for $q>0$) for the BHM phase diagram as calculated by the ground state Chern number Eq.~(\ref{Chern})~\cite{Huber,LAA}. Red regions are negative, and blue and white regions are positive. QR and HCB  regions denote the effective strong interaction models (\ref{HQR}) and  (\ref{HCB}) respectively. The BHM ground state Chern number sign agrees with that of $R_{\rm H}$ at high temperatures.
}
 \label{fig:BHM-RH}
\end{center}
\end{figure}

Fig.~\ref{fig:BHM-RH} depicts the BHM ground state Chern number (Eq.~(\ref{Chern}) as calculated by Huber and Lindner~\cite{Huber}
in the strong interaction regime.
We note the red regions of Hall sign reversal at strong interactions. While at zero temperature the ground state is  a gapless superfluid with zero Hall resistivity, and the Chern calculation should not apply (see Section~\ref{sec:Proxies}). 
Nevrtheless, we find that the Hall sign agrees with the Hall coefficient sign in the the corresponding metallic phases at higher temperature, as described below.

For strong interactions, $t/U<1 $ there are Mott insulator lobes in which  the boson filling is locked into an integer $\langle n_i\rangle=m=0,1,2\ldots$.
The mean field boundaries $\mu_{\rm cr}(U, m)$ for the Mott lobe $m$, is  given by the implicit equation~\cite{QPT-Subir},
\be
zt \left( {m+1\over U m-\mu} -{m+1\over Um-\mu}\right) =1 \quad ,
\ee
which also determines the location of the Quantum Critical Points (QCP) 
$(\mu_{\rm cr}(m), U_{\rm cr}(m))$, by the equation $\left( {\partial \mu \over \partial U}\right)^{-1}=0$.

Near the QCP at density $\bar{n}$, the BHM maps onto an O(2) quantum field theory~\cite{FisherBoson}, described by a complex field $\psi(\bx_i)\propto \langle b_i\rangle$. In a dimensionless scaled form~\cite{Higgs} the O(2) action is,
\be
S^{\rm O(2)} =  {1\over g} \int \int d^{d+1} x  \sum_{\mu=1,d+1}|\left(\partial_\mu -i{q\over c} A_\mu\right) \psi |^2   + {m_0\over 8} \left( |\psi|^2 - 2 \right)^2 \quad .
\label{O2FT}
\ee
$x_{d+1}$ is the imaginary time coordinate,  $A_{d+1}=c\mu/q$. The coupling constant $g \propto U/t$
determines ground state. In $d\ge 2$, $g<g_c$, the field theory describes a superconductor, and for $g>g_c$ it is a Mott insulator.  $g_c$ is the QCP. $m_0$ is the bare Higgs (amplitude mode) mass. The low energy Higgs mass  and the superfluid stiffness vanish at the QCP.

In order to obtain operator expression for the currents and magnetization, which are necessary for the CMC and CSR susceptibilities, the BHM near the QCP can be renormalized onto an effective Quantum Rotators (QR) Hamiltonian~\cite{QPT-Subir},
\be
 H^{\rm QR} = \sum_i   {q^2\over 2 C} (n_i-\chi_c \mu)^2 -  J^{\rm eff}(n_i-\bar{n}) \sum_{\langle ij\rangle}\cos(\phi_i-\phi_j-{q\over c}A_{ij}) \quad ,
\label{HQR}
\ee
where 
\be
\left[n_i,\phi_j\right]=-i\delta_{ij} \quad .
\label{CCR}
\ee

$ H^{\rm QR}$ physically describes a Josephson junction array, whose grain capacitance and Josephson couplings are  related to the BHM parameters,
\bea
&& C \simeq {1\over U}\nonumber\\
&&J^{\rm eff} \simeq 2dt \quad . 
\eea

Near the QCP, $J^{\rm eff}$ is renormalized down to  the 
low energy superfluid stiffness $\rho_s$, as given by  the CSR,
\be
\chi_{\rm csr} =q^2 \left( j^x|j^x\right)=q^2 \rho_s(\bar{n}) \quad .
\ee
$\rho_s$ vanishes at the QCP as $(g_c-g)^\nu$, where $\nu$ is the $d+1$ dimensional correlation length exponent of the O(2) model~\cite{Snir}.
$\rho_s$ also depends on the electron density $n$.
Expanding the density fluctuations about the integer Mott density $\bar{n}$ in the path integral,  we obtain the operator,
\be
\rho_s[n_i]= \rho_s(\bar{n}) + \sum_i {d\rho_s(\bar{n})\over d\bar{n}} (n_i-\bar{n}) \quad .
\ee
Thus, the renormalized current and magnetization operators involve both phase and density operators
\bea
j^\alpha(\bx_i) &=& {q\sum_i \rho_s[n_i]~\bnabla^\alpha\phi(\bx_i) +{\rm h.c}\over 2}\nonumber\\
m(\bx_i) &=& {{q\over 2c}\rho_s[n_i] (x_i  j^y-y_i j^x )+{\rm h.c.}\over 2} \quad .
\eea
The CMC can be  calculated using the cannonical commutations in Eq.~(\ref{CCR}):
\be
\chi_{\rm cmc}= {2\over \cV} \left( j^y,[M^z,j^x]\right)={q\over \cV}\left( j^y,   x  \partial_x  \left({d\rho_s\over dn}\right) j^y \right)={q^3\over c} \rho_s {d\rho_s\over d\bar{n}} \quad ,
\ee
which according to Eq.~(\ref{RH0}):
\be
R_{\rm H}^{(0)} = {\chi_{\rm cmc}\over \chi_{\rm csr}^2 }={1\over q c} {d\log\rho_s\over d\bar{n} } \quad ,
\label{QR-RH}\ee
which qualitatively agrees with similar expressions found for the reactive Hall constant near a Mott insulator~\cite{Zotos-RH},
and for the Hall coefficient of a two leg Luttinger-liquid~\cite{Thierry-RH}.
$R^{\rm corr}$  may be assumed to be small near the QCP since 
it depdns on current relaxation which scales with
the ratio of lattice constant to  the longer correlation length.

Since  the Mott densities $\bar{n}$ minimize $\rho_s[\bar{n}]$, 
The Hall coefficient  vanishes at extrema of $\rho_s(n)$.
 Near the QCP, $R_{\rm H}/q$ changes from negative to positive at the Mott density $\bar{n}=1,2,\ldots$, since the superfluid stiffness vanishes at the QCP.  The switch from  $R_{\rm H}/q >0$ (given by Eq.~(\ref{RH-Gal}) at low densities) to $R_{\rm H}/q <0$ is expected around midpoint between two Mott lobes.
 This behavior is captured by the Hard Core Bosons model below.

\section{Hard Core Bosons on the square lattice}
\label{sec:HCB}

Between consecutive Mott lobes, at temperatures below $T<< |\mu_{\rm cr}(n)-\mu_{\rm cr}(n+1)|$,  there is a superfluid phase  with constrained density fluctuations.  

These fluctuations can be described by the renormalizing the BHM onto the projected space between two Mott lobes, $n_i\in (m,m+1)$. The projected charge fluctuations are described hard core bosons (HCB) operators,  
\bea
&&\tilde{a}_i, \tilde{a}^\dagger_i \quad , \quad (\tilde{a}^{})^2=(\tilde{a}^\dagger)^2=0 \quad ,\nonumber\\
&&0\le \tilde{n}_i \le 1 \quad .
\eea
These operators are faithfully represented by SU(2) spin-half operators,
\be
\tilde{a}^\dagger_i \!\to\! S^+_i \quad , \quad \tilde{n}_i - {1 \over 2} \!\to\!  S_i^z \quad ,\quad \bS^2={3\over 4} \quad.
\ee
which can be used to represent the HCB hamiltonian by the gauged quantum XY model of spin half,
\be
H^{\rm HCB}=- t \left(\sum_{\ij }e^{-i{q\over c} A_{ij}} S^{+}_{i}S^{-}_{j} + h.c\right) - \mu\sum_{i}S^{z}_{i} \quad ,
\label{HCB} 
\ee 
where $\mu,t$,  which depend on $U/t, m$ are renormalized from the corresponding BHM values.
The HCB charge polarizations, currents and magnetization operators are respectively represented by,
\bea
&&\bP = q \sum_{i} ~ \bx_i  S^z_i \quad ,\quad \bj=i[H,{\bf P}]=\sum_{\ij} \bj_{ij} \quad , \nonumber \\ 
&& j^\alpha_{ij}=- i q t (S^+_i S^-_j - S^-_i S^+_j)(x^\alpha_j-x^\alpha_i) \quad ,\nonumber\\
&&M={1\over 4c} \sum_{\ij}  (\bx_i +\bx_j) \times \bj_{ij} \quad.
\label{electric}
\eea
Here $\br_i$ denotes the position of site $i$.
For the square lattice, the density dependent BKT transition temperature  for HCB on the square lattice has been evaluated by QMC:~\cite{Ding-QMC1,Ding-QMC2,Harada}
\be
T_{\rm BKT}(n) \simeq 2.8 t n(1-n) \quad .
\ee

\subsection{Superconducting phase}
HCB on the square lattice exhibit long range superfluid order below the Berezinskii, Kosterlitz and Thouless~\cite{Ber,BKT}  transition (BKT) transition temperature $T-{\rm BKT}$. 
At $T\ge 0$ the
two dimensional superfluid stiffness $\rho_s$ in the classical (large $S$ approximation) is,
\be
 \rho_s \equiv   q^{-2} {d^2 F^{\rm cl}(T,n) \over (dA_x)^2}\Bigg|_{\bA=0} >0 \quad ,
\ee
which yields a maximum at half filling,
\be
\rho_s^{cl}(0,n)= 2t n(1-n) \quad .
\label{rho-cl}
\ee
Quantum corrections  to $\rho_s^{cl}(0,\half)$ enhance it  by
about 7\%~\cite{Stiffness-Sandvik,Stiffness-Troyer}.

\begin{figure}[h]
   \centering
\includegraphics[width=0.7\columnwidth]{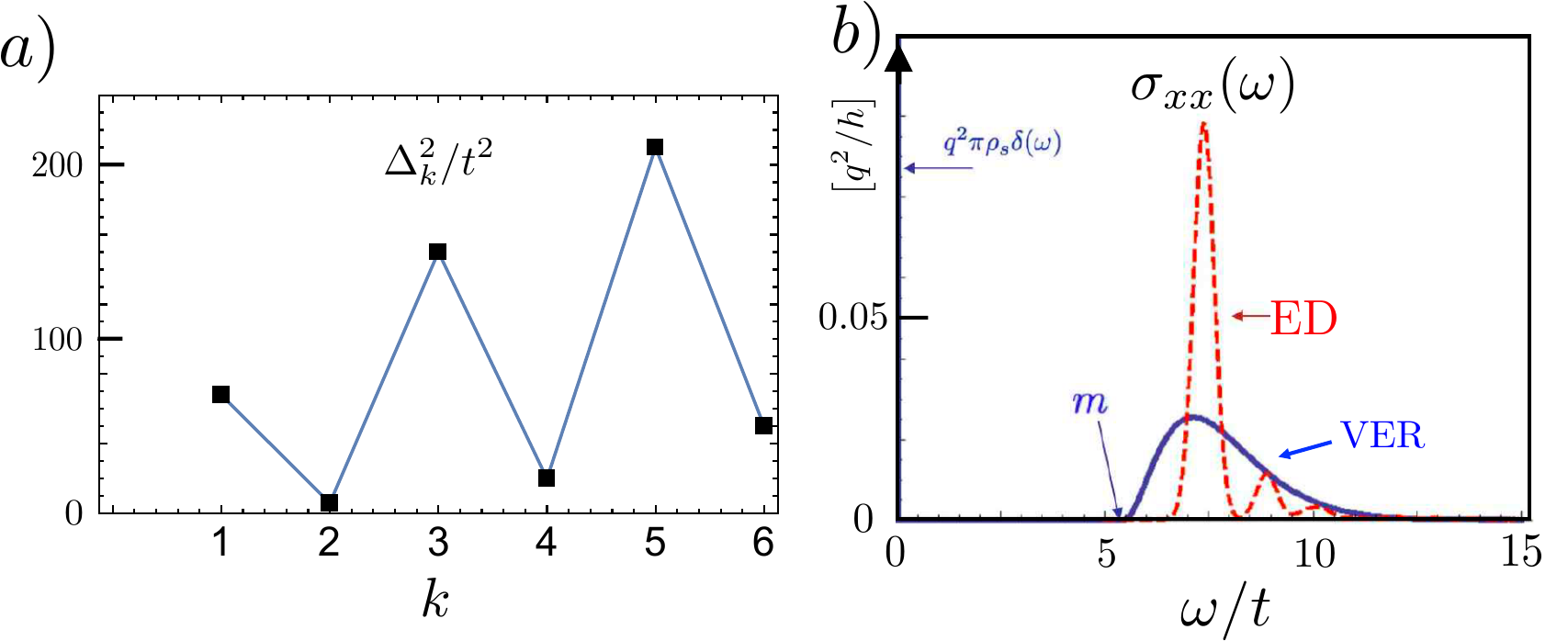}
\caption{Dynamical conductivity $\sigma_{xx}(\omega) $  of HCB in the superfluid phase.  a) 6  ground state recurrents, computed  by ED on  $4\times 4$ lattice.  Even-Odd alternation of the recurrents indicates a strong suppression of $\sigma_{xx}$ at low frequencies. b) $\rho_s$ is the zero temperature superconducting stiffness. Blue line: Conductivity by continued fraction extrapolation by Variation extrapolation of recurents. Red line: Exact Diagonalization calculation of Kubo fomrula . The apparent threshold  energy $m$ is interpreted as Higgs-amplitude mode's mass.}
\label{fig:SxxHCB0}
\end{figure}

The zero temperature AC longitudinal conductivity at half filling $n={1\over 2}$, is evaluated by computing the conductivity moments $\mu_0,\ldots \mu_{12}$ in Eq.~(\ref{EV}) by ED.
The resulting 6 recurrents $\Delta_1,\ldots \Delta_6$ are depicted in Fig.~\ref{fig:SxxHCB0}(a). The large odd $>$ even oscillations of the recurrents indicates that the conductivity is  suppressed at low frequencies. In fact, the continuous form of $\sigma_{xx}(\omega)$, shown in Fig.~\ref{fig:SxxHCB0}(a), exhibits a threshold at around $\omega=5t$. This threshold behavior is  supported by an ED calculation for the Kubo formula on a 4$\times$ 4 lattice.
\be
\sigma_{xx}(A,\ve)={1\over A}\sum_{n>0} |\langle \Psi_0|j^x |\Psi_n\rangle|^2 {\ve \over (E_n-E_0-\omega)^2 + \ve^2 } \quad .
\ee

We note that this threshold  feature  is not expected by BCS theory in a single band, weak coupling superconductor without disorder.
By Mattis and Bardeen~\cite{MB}, such a  threshold 
would appear at $2\Delta$, where $\Delta$ is the pairing gap only if there is sufficient  disorder.
For  HCB model (\ref{HCB}) $\Delta=\infty$, and there is no disorder!

\begin{figure}[h]
   \centering
\includegraphics[width=0.3\columnwidth]{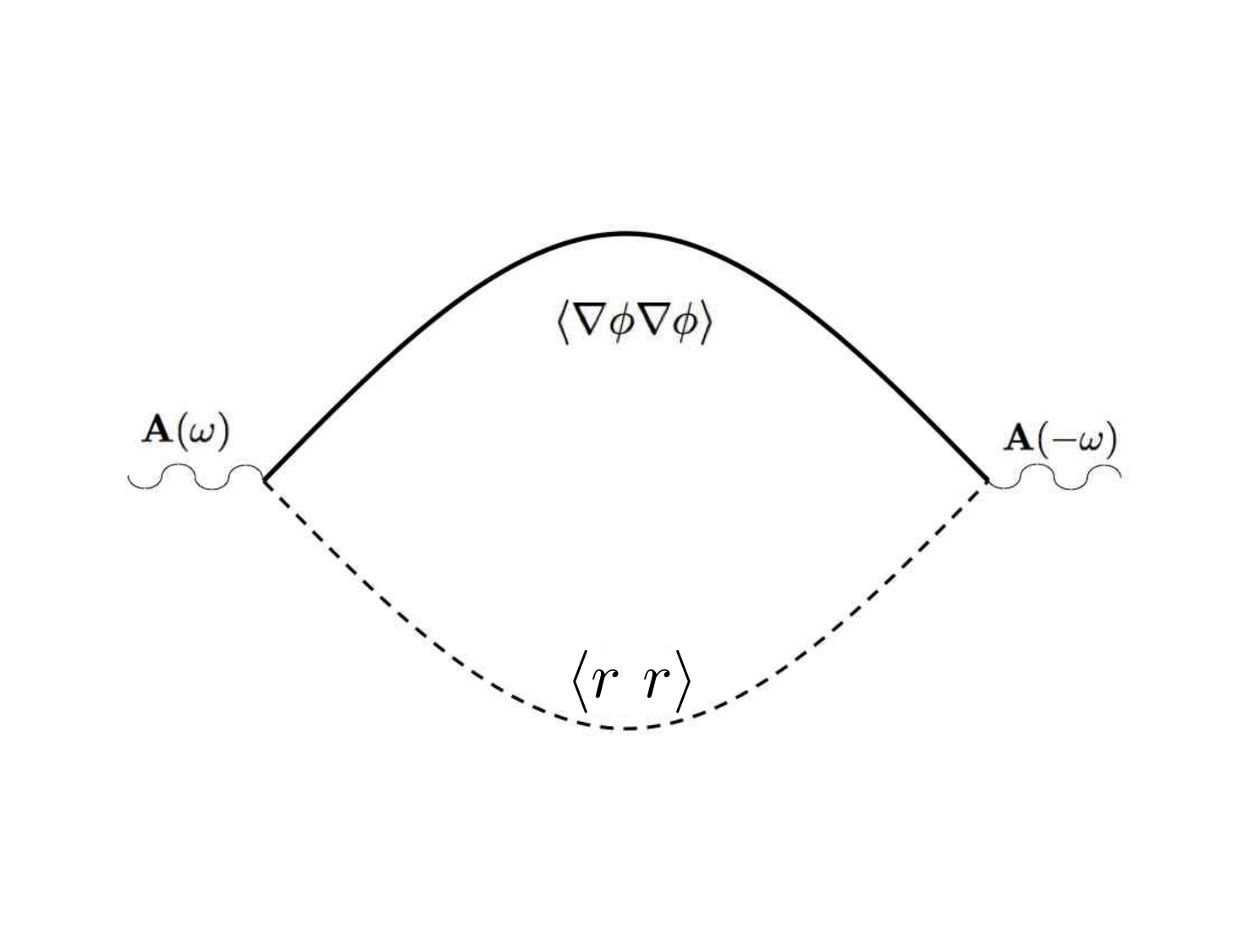}
\caption{Dynamical conductivity of the Relativistic Gross Pitaevskii model Eq.~(\ref{RGP}), at zero temperature in the superconducting phase at weak coupling.  $r$ is the amplitude/Higgs mode  with mass $m$, and $\phi$ the massless Goldstone mode. The conductivity exhibits a threshold at $\omega=m$  as shown in Eq.~(\ref{SxxHiggs}). }
\label{fig:Bubble}
\end{figure}

This threshold was  argued to be the Higgs-amplitude mode of the particle-hole symmetric Relativistic Gross-Pitaevskii (RGP)  theory~\cite{Snir,Higgs} with Eucledean action,
\beq 
S^{\rm RGP}={1\over 2g}\int_\Lambda \! d^{2+1}\!x\left[ \left|(\bnabla-i{q\over c} \bA)\vec{\Phi}\right|^2+{m_0^2\over 4}\left(| \vec{\Phi} |^2-1\right)^2\right] \quad .
\label{RGP}
\eeq 
Here, $\vec{\Phi}$ is the  two component  rotator  field, defined at wavevectors below the cut-off $\Lambda$. The imaginary time component is $x_{3}=c\tau$, where  $c=\sqrt{\rho_s/\chi}$,   $\rho_s \sim t$ is the microscopic stiffness, and $\chi\sim  {1\over U}$
is the local compressibility. $g$ is the quantum parameter, which drives a quantum phase transition~\cite{QPT-Subir} into a bosonic Mott insulator as it grows toward  $g\to g_c$.  

At weak coupling and zero temperature, $\vec{\Phi}(\bx)=\vec{\Phi}_0$,
and the conductivity is evaluated  by fluctuations of the amplitude $r$ and phase $\phi$,
\be
\vec{\Phi}=~\Phi_0 (1+r)~\hat{\br}  +~\phi ~ \hat{z}\times\hat{\br} \quad . 
\ee
At lowest order in $g^0$, the bubble diagram in Fig.~\ref{fig:Bubble} yields a threshold AC conductivity given by~\cite{Higgs},
\be 
\sigma^{\rm RGP}_{xx}(\omega)\sim \omega^5 ~\Theta(\omega^2-m^2) \quad ,
\label{SxxHiggs}
\ee 
where $m$ is the mass-gap of the amplitude field $r$. $m(g)\sim |g-g_c|^\nu$  is the analogue of the Higgs particle mass, which vanishes at the critical point with the three dimensional XY model's correlation length exponent $\nu\simeq 0.671$.

\begin{figure}[ht!]
\begin{center}
\includegraphics[width=8cm]{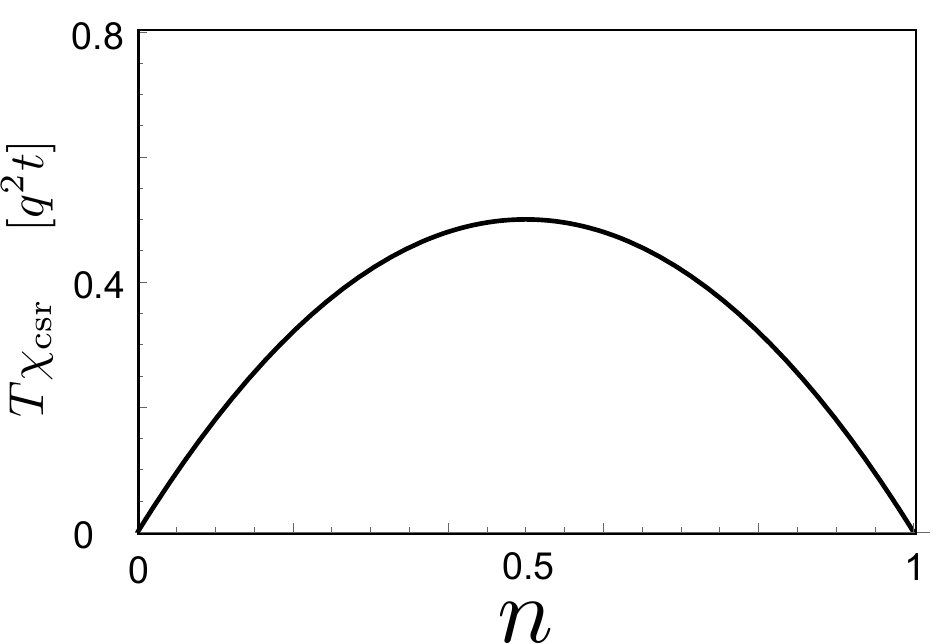}
\caption{CSR of HCB at  order, $\cO(\beta)$,
as give by Eq.~(\ref{CSR-HCB}).}
\label{fig:CSR-HCB}
\end{center}
\end{figure}

\subsection{Metallic phase: longitudinal conductivity}
\label{sec:HCB-cond}
The longitudinal conductivity was evaluated for the metallic phase of $H^{\rm HCB}$
on the square lattice, at $T> T_{\rm BKT}$ In Refs.~\cite{LA,HCB}.
The CF expansion is extrapolated following the Variational Extrapolation of Recurrents~\cite{Ilia} as described in subsection \ref{sec:VER}.
\begin{figure}[h]
\begin{center}
\includegraphics[width=8cm,angle=0]{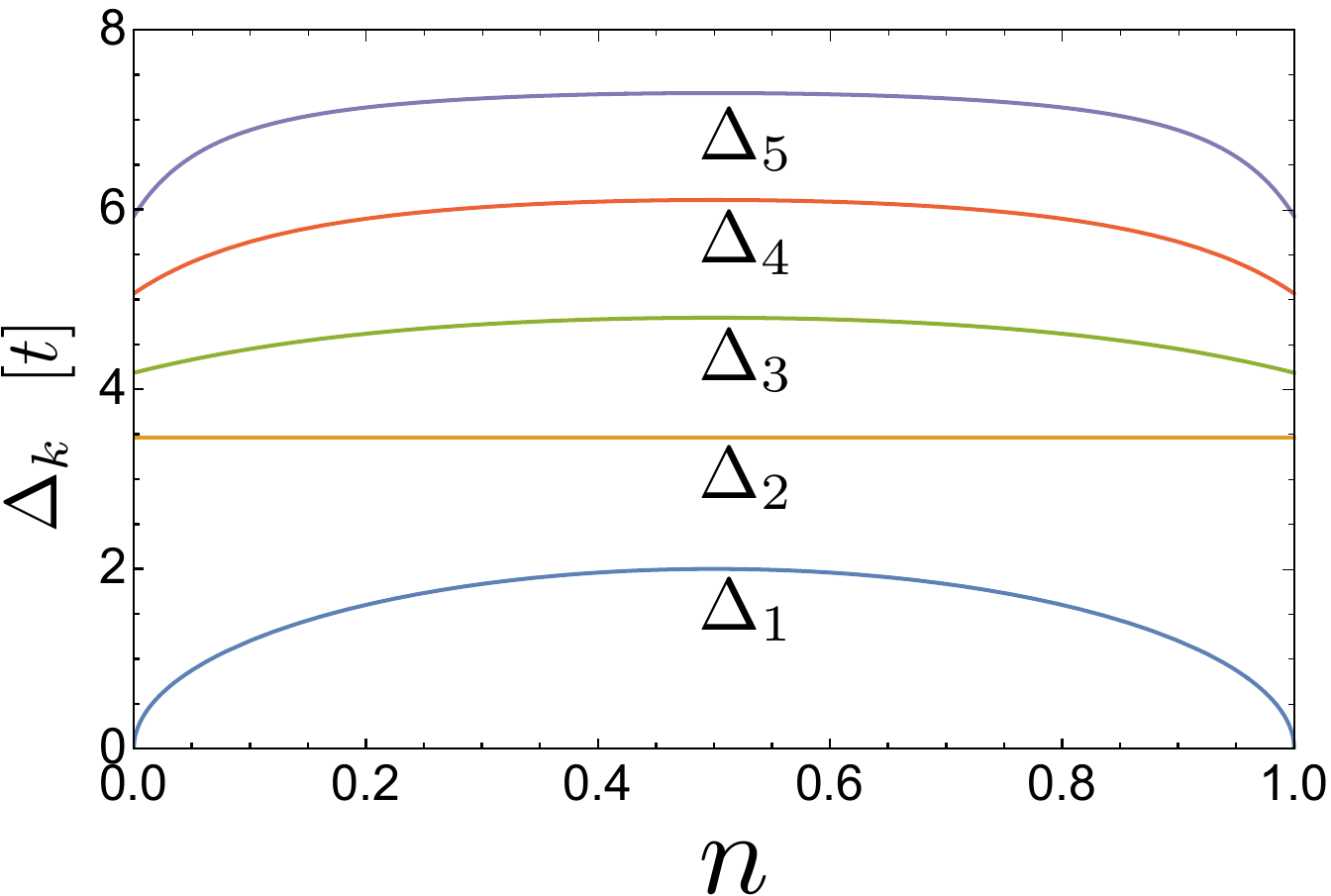}
\caption{Density dependence of the five lowest recurrents. The cancellation of the density dependence of $\Delta_1(n)$ and  $\chi_{\rm csr}(n)$ leads to the weak density dependence of the linear resistivity slopes.}
\label{fig:Deltas}
\end{center}
\end{figure}

The   leading orders CSR is,
\be
\chi_{\rm csr} = 2 q^2 \beta \tt^2 n(1-n) + \cO(\beta t)^3 \quad .
\label{CSR-HCB}
\ee
$\mu_{2k}, k=0,\ldots 5$ were also evaluated to leading order in $\beta t$, which yields  recurrents of order $\cO(\beta^0)$. 
In Fig.~\ref{fig:Deltas}, the  recurrents are plotted  for densities $0<n<1$.
\begin{figure}[h]
\begin{center}
\includegraphics[width=8cm,angle=0]{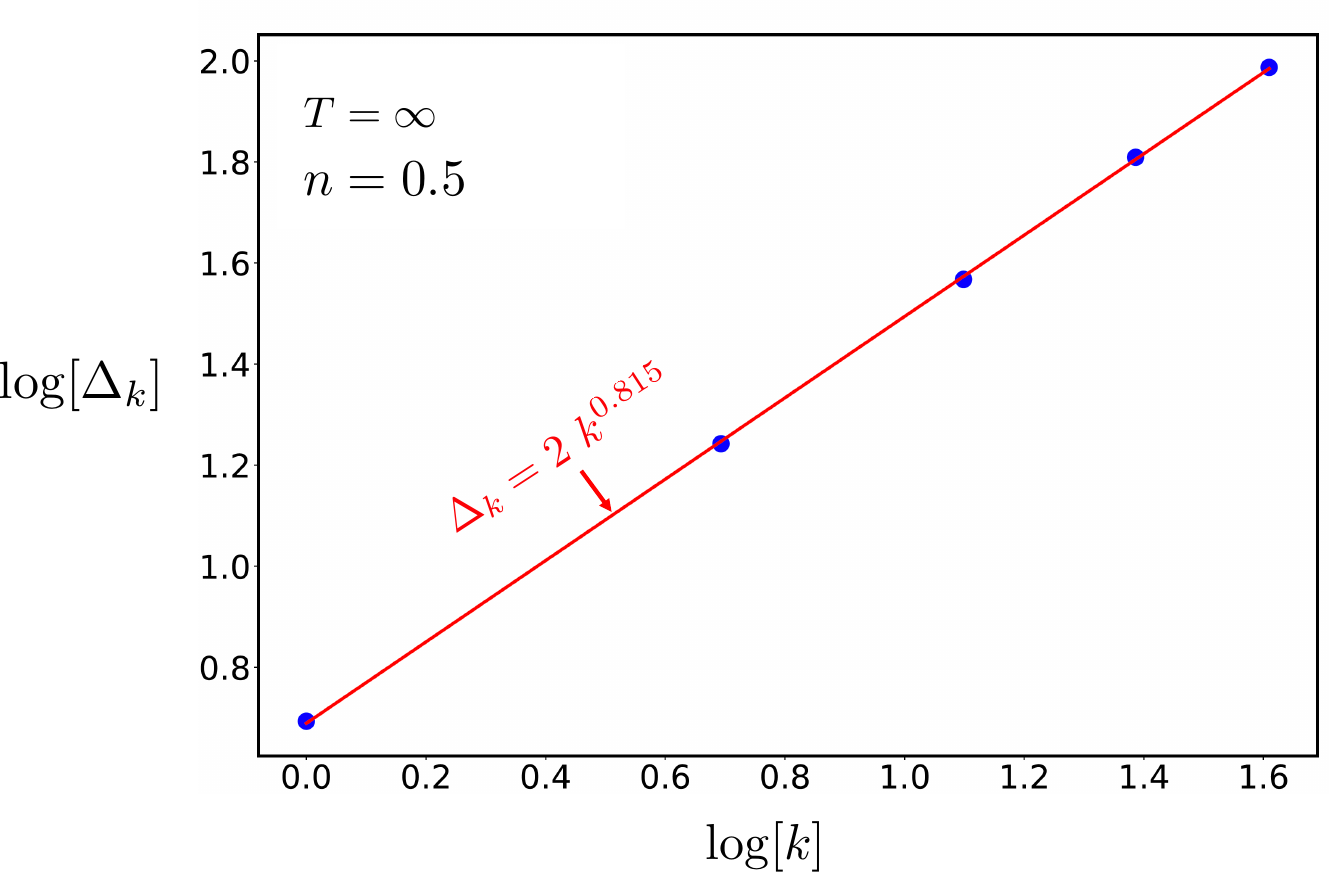}
\caption{Order dependence of HCB recurrents (in units of $t$) at half filling, and infinite temperature}
\label{fig:HalfFill-fit}
\end{center}
\end{figure}

\begin{figure}[h]
\begin{center}
\includegraphics[width=8cm,angle=0]{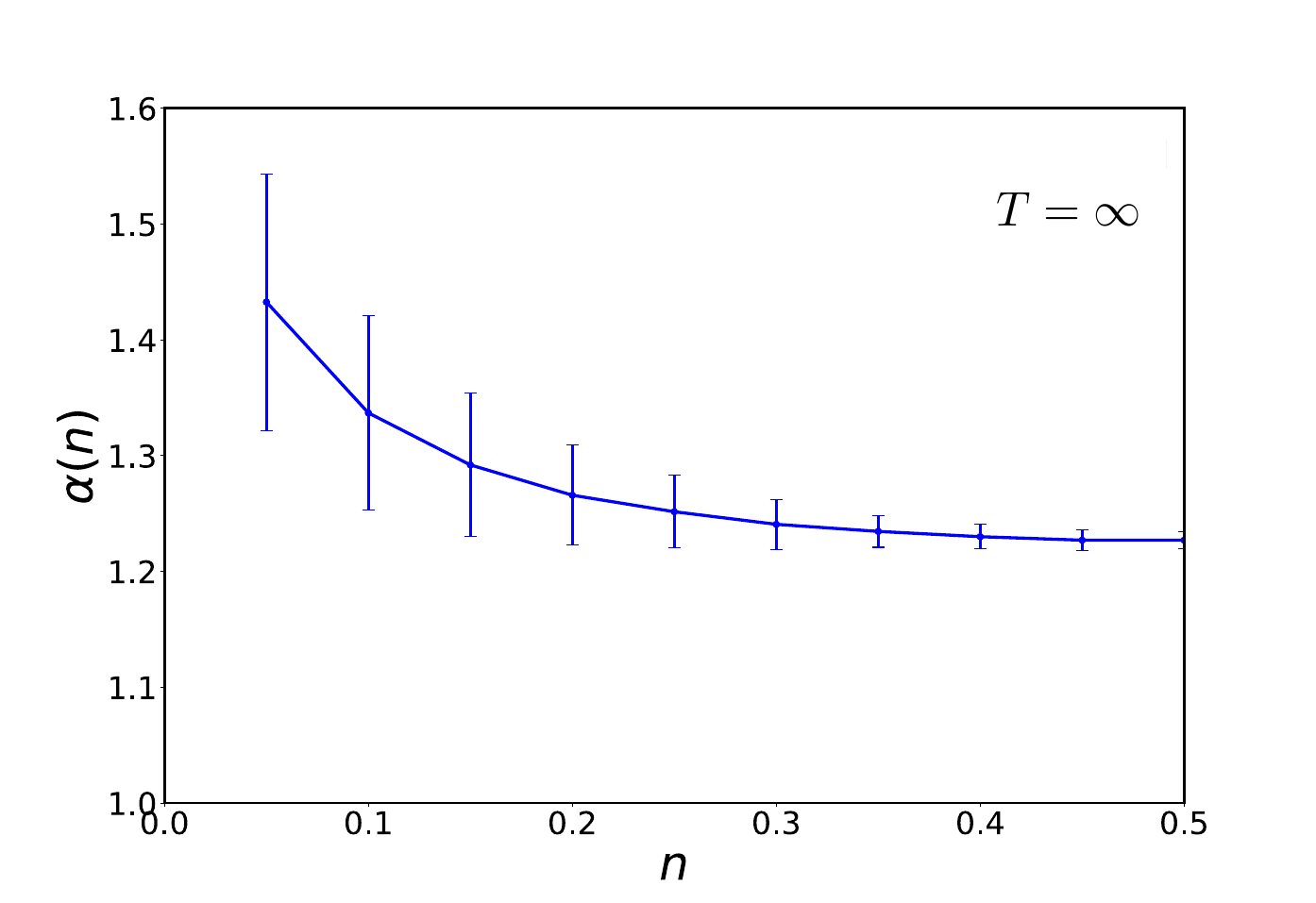}
\caption{The fit exponent $\alpha$ as a function of HCB filling $n$,
which determines the termination function $G^>_{2}$ of the AC conductivity. Error bars describe the mean square deviations for the power law fits of $\Delta_k \sim k^\alpha, k=2,\ldots 5$.}
\label{fig:alpha-n}
\end{center}
\end{figure}

\begin{figure}[h]
\begin{center}
\includegraphics[width=8cm,angle=0]{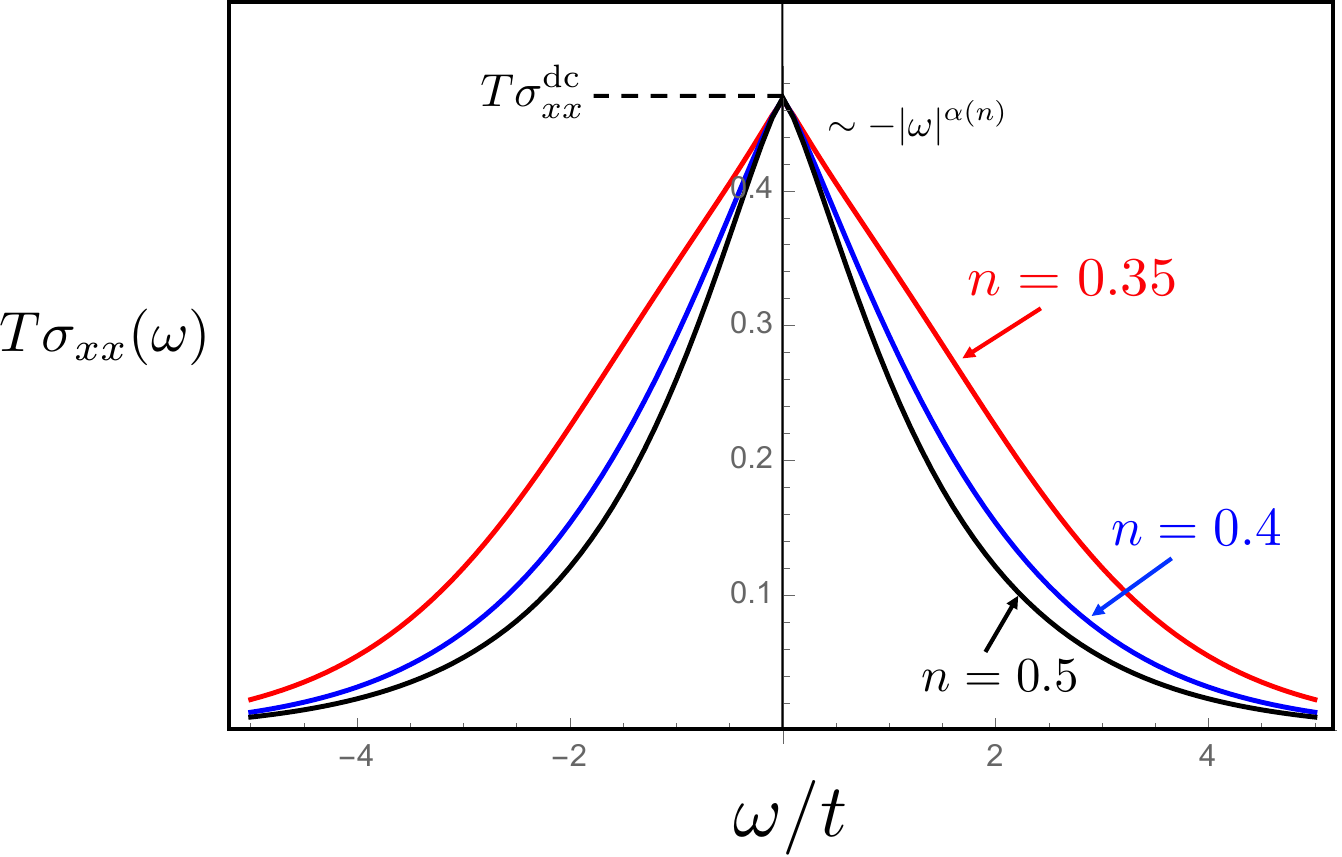}
\caption{HCB dynamical conductivity at high temperatures. The DC value of $T\sigma_{xx}(0)$ is independent of boson density $n$, while the width of the conductivity increases away from half filling. The singularity at zero frequency
depends on $\alpha$ as determined by Fig.~\ref{fig:alpha-n}.}
\label{fig:Sxx-highT}
\end{center}
\end{figure}
According to (\ref{Freud}) and Fig.~\ref{fig:Freud}, the termination function appears to be  well fit by the stretched exponential form i.e.
\bea
\bar{G}_{00}''(\omega,\alpha) &=&  N_\alpha \exp\left(  -\left|{\omega\over \Omega_\alpha }\right|^{\alpha}\right) \quad ,\nonumber\\
\bar{G}_{00}'(\omega,\alpha) &=&  {1\over \pi} {\rm PV}\int_{-\infty}^\infty\!d\omega' { \bar{G}_{00}''(\omega,\alpha)\over \omega-\omega' } \quad ,
\label{barG-HCB}
\eea
where the normalization 
\bea
N_\alpha &=&{\pi \alpha\over  2 \Omega_\alpha \Gamma\left({1\over \alpha}\right) } 
\eea
ensures that  
$\int d\omega\bar{G}_{00}''(\omega,\alpha)=\pi $.

The $S^z\to -S^z$ symmetry of the HCB hamiltonian leads to $\sigma_{xx}(n)=\sigma_{xx}(1-n)$.
For $n\ne 0.5$, $\Delta_1$ is an outlier of the power law line, since it vanishes as $n\to 0,1$.
The higher recurrents are fit by
\be 
\bDelta_k=\Omega_\alpha ~ a_\alpha  k^{\alpha(n)} \quad ,\quad k=2,\ldots
\label{dD-HCB}
\ee
 where $\alpha(n)$ is obtained by least square fit in Fig.~\ref{fig:alpha-n}. The error bars increase  away from half filling. We find that fluctuations of the   recurrents about Eq.~(\ref{dD-HCB}) induces large uncertainty in th extrapolated $\sigma_{xx}(\omega)$.
 Hence  the use of Eq.~(\ref{dD-HCB}) to the regime $0.35\le n\le 0.5$.

The high temperature dynamical conductivities are plotted in Fig.~\ref{fig:Sxx-highT}. 
Since the first two recurrents are,
\be
\Delta_1^2 = 16 n(1-n)t^2 \quad ,\quad   \Delta^2_2= 12t^2 \quad .
\label{Delta1-HCB}
\ee
The conductivity is proportional to the  density independent ratio,
\be
\sigma_{xx}\propto {\chi_{\rm csr}(n)\Delta_2\over \Delta^2_1(n)} \quad .
\label{DCratio-HCB}
\ee
Eq.~(\ref{DCratio-HCB}) is interpreted  as the ratio of kinetic energy ($\chi_{\rm csr}(n)$) to scattering rate ($\Delta_1^2(n)/\Delta_2$), which have a similar density dependence. Hence, since $\alpha$ changes very little ($\sim 1\%$ in the range $n\in (0.35,0.5)$), the linear resistivity slope in this regime is very weakly density dependent, and given by
\be
R_{xx}^{\rm dc} \simeq  0.33  {T\over  t} R_Q \quad ,~~~~~R_Q\equiv {h\over q^2} \quad .
\label{RT}
\ee

The CF extrapolation also yields a singular frequency dependence at low frequencies,
as shown in Fig.~(\ref{fig:Sxx-highT}), 
\be
{ \sigma_{xx}(\omega)- \sigma^{\rm dc}_{xx}\over \sigma^{\rm dc}_{xx} } \sim -   \left|{\omega\over  \Omega_\alpha  }\right|^{\alpha} \quad .
\ee
The sharp zero frequency cusp is consistent with Mukerjee,  Oganesyan, and Huse~\cite{Vadim}, (MOH)  who found a similar singularity in  the high temperature conductivity of a one dimensional non-integrable fermion model, and predicted similar behavior in higher dimensional  models.

\begin{figure}[ht!]
\begin{center}
\includegraphics[width=8cm]{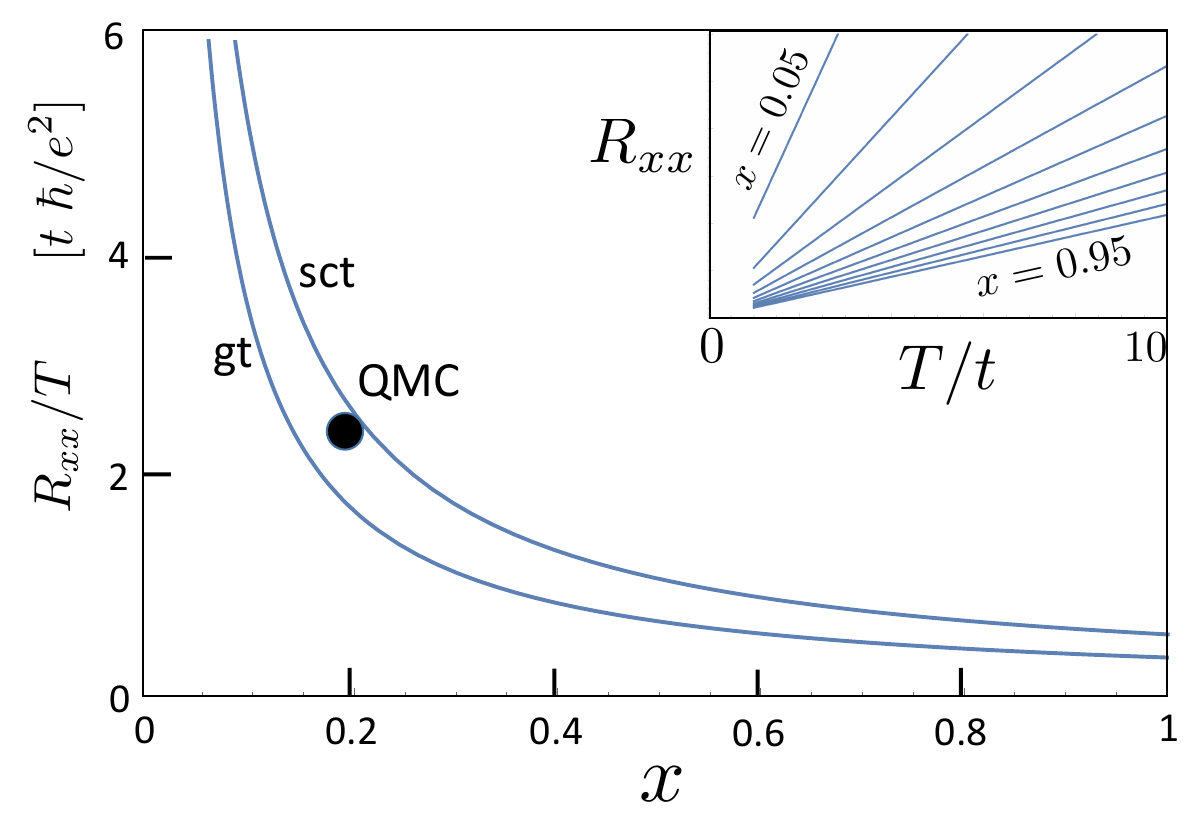}
\caption{Resistivity versus temperature of HCB at half filling. The high temperature linear slope of Eq.~(\ref{RT}) (solid line) is matched toward the superconducting transition temperature $T-{\rm BKT}$ using Halperin and Nelson's (HN) free vortices theory, Eq.~(\ref{HN}) (Dashed line). }
\label{fig:RXX}
\end{center}
\end{figure}

 The high temperature linear temperature slope of the HCB resistivity can be connected to the superconducting transition at  $T_{\rm BKT}\simeq 0.7 t$. 
 Halperin and Nelson (HN)~\cite{HN} described the superconducting fluctuation region just  $T_{\rm BKT}$, where the resitivity rises as 
\bea
R^{\rm HN}_{\alpha\beta} &\simeq& 2.7 R^{\rm n}_{\alpha\beta}(T)\left({\xi_+\over \xi_c}\right)^{-2} \nonumber\\
&=&  2.7 R^{\rm n}_{\alpha\beta}(T)~ \exp\left(-2b\left( {T_{\rm BKT}\over T-T_{\rm BKT} } \right)^{1\over 2} \right) \quad .
\label{HN}
\eea
$\xi_+$ is the BKT correlation length, and $\xi_c$ is of the order of the HCB lattice constant and $b\simeq 1$.  
For the HCB model, the ``normal state'' resistivity $R^{\rm n}_{xx}(T)$ is obtained from our Eq.~(\ref{RT}).
We use these values to plot the crossovers from HN theory Eq.~(\ref{HN}) to higher temperatures as dashed lines in Fig.~\ref{fig:RXX}.

\subsection{Metallic phase: Hall coefficient}
The zeroth term,
\be
R_{\rm H}^{(0)} = {\chi_{\rm cmc}\over \chi_{\rm csr}^2 } \quad .
\label{RH0}
\ee
is calculated by the high temperature expansion of the CSR, Eq.~(\ref{CSR-HCB}) and the 
CMC susceptibility,  as prescribed in Section \ref{sec:HighT}.  The leading order in $\beta$ evaulated in Ref.~\ref{HCB} is
\bea
\chi_{\rm cmc} &=& {2\over c} q^3  t^2 \Tr \left( \rho(\beta,n)    (S^{+}_{1}S^{-}_{3}+S^{-}_{1} S^{+}_{3})S^{z}_{2} \right) \nonumber\\
&=&  {1\over c} \beta^2 q^3 t^4   (1-2n)(1-(2n-1)^2) \quad .
\label{chi-CMC}
\eea

One observes that $\chi_{\rm csr}$ ($\chi_{\rm cmc}$) is ``particle-hole'' symmetric (antisymmetric) under $n\to 1-n$. The zeroth Hall coefficient expanded to second order in $\beta$~\cite{HCB} is
\be
R_{\rm H}^{(0)}=   \frac{1}{qc}\left(\frac{2n-1}{n(n-1)} + \frac{2}{3}(\beta t)^2\left(n-\frac{1}{2}\right)\right) \quad . 
\label{RH0-highT}
\ee
Notably at low density, the Hall coefficient recovers the continuum Galilean invariant result
$R_{\rm H}^{(0)} \to (nqc)^{-1}$. Near half-filling, $R_{\rm H} \!\sim \!  -8({n-{1\over 2})/(qc)}$, reflecting the effects of lattice Umklapp and hard core scattering.

\begin{figure}[ht!]
\begin{center}
\includegraphics[width=8cm]{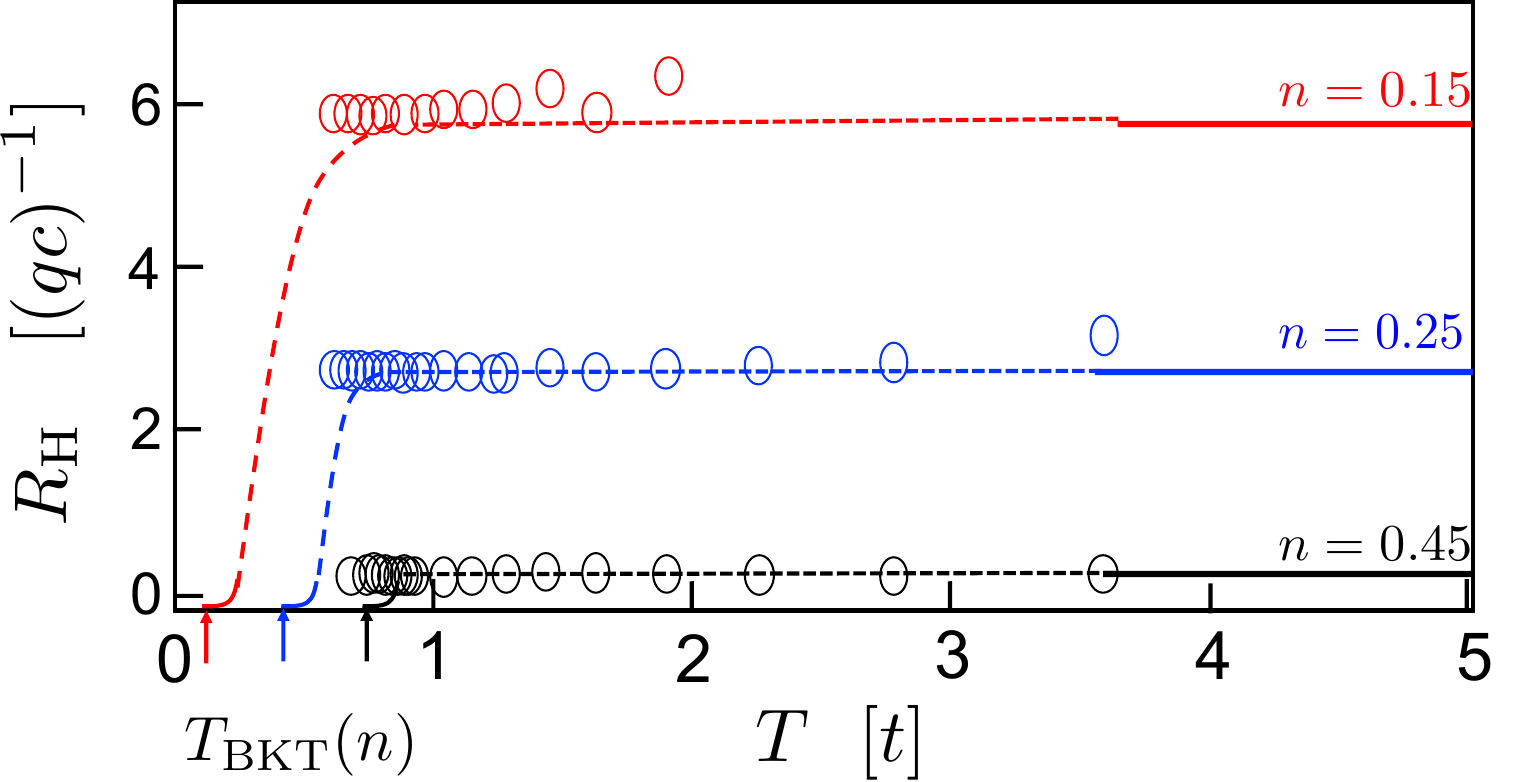}
\caption{Metallic Hall coefficient $R_{\rm H}^{(0)}$ of HCB, from Ref.~\cite{HCB}.
High temperature asymptotes are depicted by solid lines, and  lower temperature  QMC calculations are plotted with circles. Dashed lines interpolate between $R_{\rm H}^{(0)}(T)$  and vortex fluctuations theory~\cite{HN} 
near the superfluid transition at $T_{\rm BKT}(n)$. }
\label{fig:RHvsT}
\end{center}
\end{figure}

Eq.~(\ref{RH0-highT}) was extended to lower temperatures numerically~\cite{HCB} by a path-integral based QMC 
for bosonic lattice models using the DSQSS package~\cite{Motoyama}, with a DLA algorithm~\cite{DLA}. The method is reviewed Section~\ref{sec:QMC-DLA}. 

We used the equal-time ($\tau_1-\tau_2=0$) correlation function results for evaluating the static expectation values. For two-body operators (e.g. as needed for the CSR), one can directly use this result, for nearest-neighbour separation. For the CMC, we had to `re-weight' the measurements of the DSQSS code (which can only measure two-point correlation functions) to account for the additional $S^{z}$ operator at the nearest neighbour location. Since the algorithm is explicitly formulated in the $S^{z}$ basis, the lattice configuration of these values are easily accessed at each imaginary time slice. Hence, a simple modification of the correlation function measurement part of the source code sufficed for our purpose.

For acceptable statistical error,  we used $24 \times 24$ size lattices, which were sufficient for converging expectation values at temperatures higher than the expected superfluid transition. The number of Monte Carlo sweeps was $\sim 10^{6}$.
In Fig.~\ref{fig:RHvsT}, the solid lines are the analytical results of Eq.~(\ref{RH0-highT}), while the QMC data are depicted by open circles. We see that the Hall coefficients above the HN regime, rapidly saturate  to their high temperature limit. 

\begin{figure}[t]
\begin{center}
\includegraphics[scale=0.3]{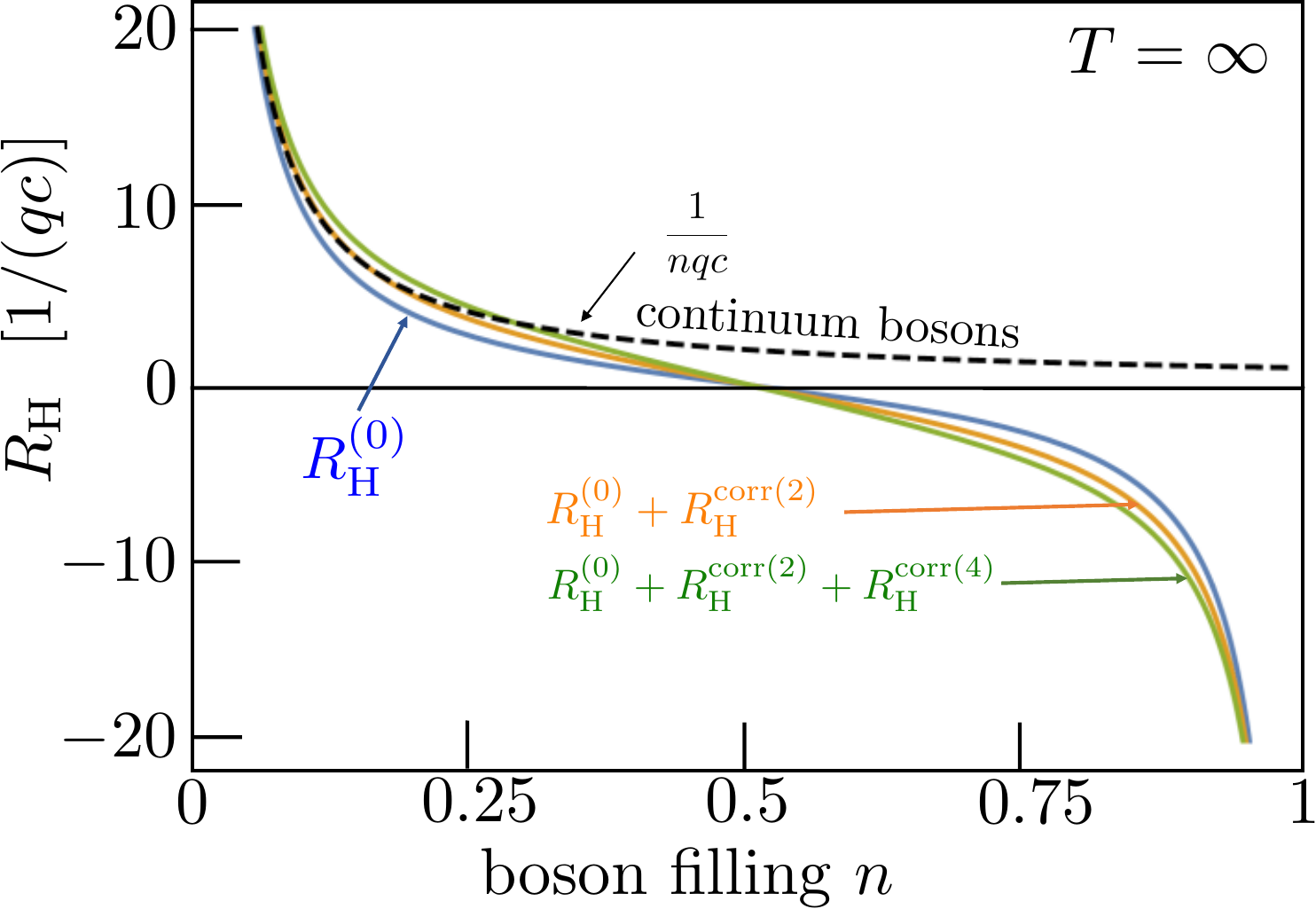}
\caption{Density dependent Hall coefficient $R_{\rm H}(n)$ of HCB on a square lattice at high temperature from Ref.~\cite{HCB}. $R_{\rm H}^{(0)} and R^{\rm corr}_{\rm H}$ are defined in Eqs.~(\ref{RH0}) and (\ref{Rcorr}).
Monotonous convergence of the corrections  up to fourth order are shown by the yellow and green curves. The  Hall sign change at half filling is a consequence of the the hard core interactions and the lattice effects on the bosons. 
}
\label{fig:RHwCorr}
\end{center}
\end{figure}
\subsection{Hall coefficient corrections}
\label{sec:Rcorr-HCB}
In general,  $R^{\rm corr}$ can be a cumbersome calculation. The relative simplicity of the HCB Hamiltonian permits a feasible computation of $R^{\rm corr}$  up to fourth Krylov order. The calculated correction~\cite{HCB} is the truncated sum,
\bea
&&R_{\rm H}^{\rm corr(4)}= {1\over \chi_{\rm csr}}  \sum_{i,j=0}^4 R_i R_j  (1 -\delta_{i,0} \delta_{j,0}) \cM_{2i,2j}\nonumber\\
&& \cM_{2i,2j}=  \Im \left( \kexpect{2i; y} {\cM} {2j; x} - \kexpect{2i; x} {\cM} {2j; y} \right) \quad ,\nonumber\\
&&R_{i>1} = \prod_{r=1}^i \left(-{\Delta_{2r-1}\over \Delta_{2r} } \right) \quad ,\quad R_0 =1 \quad.
\label{Rcorr}
\eea

The high temperature hypermagnetization matrix elements $\cM_{2i,2j}$ are calculated as described in  Section \ref{sec:Rcorr}.
$\tilde{\cM}_{44}$  required  numerical evaluation of the traces of over $\sim 10^7$ operator products, using a multiplication table similar to the calculations described in Section  Section \ref{sec:numericalOperators}.
The density dependent high temperature Hall coefficient is depicted in Fig.~\ref{fig:RHwCorr}. 

We see that these corrections do not qualitatively change the zeroth term's behavior especially near the densities $n=0,1,{1\over 2}$, although they  converge slower around intermediate densities $n= 0.25,0.75$. Since the Hall coefficient is finite for any metal, the
summation over higher order corrections must converge. 
Thus, $R^{(0)}_{\rm H}$ appears to be qualitatively correct at high temperatures. 
At lower temperatures  $R_{\rm H}^{(0)}$, as evaluated by QMC in Fig.~\ref{fig:RHvsT},   appears to be blind to the onset of long range superconducting phase correlations and vortices, as described by HN theory. Therefore toward $T\to T_{\rm BKT}$, the correction term is expected to grow  relative to $R_{\rm H}^{(0)}$ and cancel it completely at $T_{\rm BKT}$.

\section{Discussion}
From the calculations shown above, one concludes that near half filling, the `weakly interacting continuum bosons' description fails for HCB. The quantum mechanical effects of lattice periodicity and  constraints of no-double occupancies play a crucial role in the transport coefficients.

We note that metallic phases of HCB and the tJM electrons share their proximity to Mott insulators. As shown in Part \ref{part:SCE}, the Hall sign of the tJM  also diverges toward the Mott phase, and reverses its sign relative to that predicted by models of weak interactions.

\newpage

\part{Summary and Future Directions}
 \label{Part:Summary}

This report reviews recent theoretical advances in quantum transport theory with particular emphasis on application to strongly interacting, gapless phases of matter. The DPP formulas for DC Hall-type conductivities  can reduce the computational cost of Kubo formulas in the Lehmann representation. They also help clarify conceptual dilemmas about the role of  gapless eigenstates
 in carrying the Hall currents on OBC, and the ultimate irrelevancy of the magnetization subtractions in thermal Hall conductivities.
The DPP formulas generalize the Berry curvature expressions  derived in Section \ref{sec:Berry}, for clean non interacting models. This may help us understand the interplay between Berry curvatures  and effects of boundaries and disorder~\cite{Sinitsyn} within the Kubo formula framework, without appealing to semiclassical approximations.

A  significant fraction of this report is devoted to thermodynamic approaches which include continued fractions for AC conductivities, and thermodynamic summation formulas for the electric, thermoelectric  and thermal Hall coefficients. These  approaches converge better by first renormalizing the Hamiltonian down to the temperature scale
of interest, and calculating conductivities which do not consist of separable contributions. 

As examples, continued fractions and Hall coefficient summation formulas  are applied to basic models of interacting lattice fermions and bosons: the Hubbard, t-J,  and  HCB Hamiltonians.  Conductivities of strongly interacting models are contrasted with their weakly interacting limits in order to illuminate  the qualitative effects of strong interactions.

In future studies, thermodynamic approaches can take advantage of reliable variational ground states, e.g.  DMRG~\cite{DMRG}, to obtain low temperature dynamical correlations in other strongly correlated lattice models of spins, fermions and bosons.
For example: thermal Hall coefficient formula of  strongly frustrated magnetic insulators~\cite{Matsuda-Kitaev},
the Hall conductivity of  Heavy fermion metals~\cite{RH-HF-Analytis}, and nearly ferroelectric 
 material SrTiO$_4$~\cite{Behnia-SrTiO}, conductivities of 
 interacting flat band multilayers of graphene~\cite{TBG} and transition metal dichalcoginides~\cite{NPJ1,NPJ2}. 
 
 Experimental measurements of conductivities of cold fermionic and bosonic atoms trapped in optical lattices with artifical gauge fields~\cite{Bakr} could test the accuracy of thermodynamic calculations.  
 On the mathematical side,  further understanding of the relation between low order recurrents and the low frequency behavior of  conductivities would be very worthwhile.

\section{Acknowledgements}
We thank Gil Refael for his encouragement to write the Report. We are indebted to  Netanel Lindner, Snir Gazit, Ilia Khait, Noga Bashan, Abhisek Samanta and Ari Turner,  whose results are reviewed in this Report. We are   grateful for critical comments by Steve Kivelson, Steve Simon, Bruno Uchoa, and Daniel Arovas, which improved the readability of the manuscript. A.A. acknowledges the
Israel Science Foundation (ISF) Grant No. 2081/20. This Report was written in part at the  Aspen Center for Physics, which is supported by National Science Foundation grant PHY-2210452, and at the Kavli Institute for Theoretical Physics (KITP) supported in part by Grant Nos. NSF PHY-1748958, NSF PHY-1748958 and NSF PHY-2309135.

\bibliographystyle{elsarticle-num} 

\bibliography{refs-PR.bib}

\begin{thebibliography}{100}
\bibitem{Comm:QMC-AC}
{Imaginary time QMC conductivities require analytic continuation, which is
  limited to frequencies higher than temperature (see Appendix B in
  \cite{Snir}). The DC conductivities are often deduced by proxies to the
  analytical continuation~\cite{Dev-PR}.}
\end{thebibliography}
\end{document}